\documentclass[11pt]{article}
\usepackage{amsmath,amssymb,amsfonts,latexsym,graphicx,amsthm}
\usepackage{fullpage,color}
\usepackage[colorinlistoftodos,prependcaption,textsize=tiny]{todonotes}
\usepackage{url,hyperref}
\usepackage{comment}
\usepackage[linesnumbered,boxed,ruled,noend]{algorithm2e}
\SetKwRepeat{Do}{do}{while}
\usepackage{framed}
\usepackage[shortlabels]{enumitem}
\usepackage{cleveref}
\usepackage{xcolor}
\usepackage{algorithmic}
\usepackage{setspace}
\usepackage[all]{nowidow}
\usepackage{tcolorbox}
\usepackage{authblk}
\usepackage{multicol}
\usepackage{tikz}
\usetikzlibrary{backgrounds}

\usetikzlibrary{decorations.pathreplacing,calligraphy}
\usetikzlibrary{decorations.pathmorphing}

\usepackage[margin=1in]{geometry}

\newtheorem{theorem}{Theorem}[section]
\newtheorem{lemma}{Lemma}[section]
\newtheorem{obs}{Observation}[section]
\newtheorem{definition}{Definition}[section]

\newtheorem{corollary}{Corollary}[section]
\newtheorem{claim}{Claim}[lemma]

\newtheorem{question}{Question}
\newtheorem{mdresult2}[question]{Question}

\newtheorem{invariant}{Invariant}[section]

\newtheorem{assumption}{Assumption}[section]

        {\medskip}

\newenvironment{shay}[1]{\noindent \textbf{\textcolor{red}{Shay: #1}}}

\newcommand{\etal}{{\em et al.}\ }

\newcommand{\eps}{\epsilon}

\newcommand{\opt}{\mathsf{OPT}}

\newcommand{\ceil}[1]{\lceil #1 \rceil}
\newcommand{\floor}[1]{\lfloor #1 \rfloor}

\newcommand{\wts}{\omega}

\newcommand{\alg}{\mathsf{alg}}
\newcommand{\sets}{\mathcal{S}}
\newcommand{\univ}{\mathcal{U}}
\newcommand{\cost}{\mathsf{cost}}
\newcommand{\brac}[1]{\left(#1\right)}
\newcommand{\lev}{\mathsf{lev}}
\newcommand{\plev}{\mathsf{plev}}
\newcommand{\tlev}{\mathsf{toplev}}
\newcommand{\tarlev}{\mathsf{tarlev}}
\newcommand{\asn}{\mathsf{asn}}
\newcommand{\cov}{\mathsf{cov}}

\newcommand{\ind}{\mathbf{1}}
\newcommand{\reset}{\mathsf{reset}}
\newcommand{\old}{\mathrm{old}}
\newcommand{\odd}{\mathrm{odd}}
\newcommand{\even}{\mathrm{even}}

\definecolor{BrickRed}{rgb}{.72,0,0}
\def\EMPH#1{\emph{\textcolor{BrickRed} {#1}}}

\begin{comment}

\begin{document}

\begin{titlepage}
%\title{Dynamic Greedy Set Cover with Worst-Case Update Time}
%\title{A Lossless Deamortization Approach for Dynamic Set Cover}
\title{A Lossless Deamortization for Dynamic Greedy Set Cover}
\author{Shay Solomon\thanks{Tel Aviv University, \href{}{solo.shay@gmail.com}}\and
	Amitai Uzrad\thanks{Tel Aviv University, \href{}{amitaiuzrad@mail.tau.ac.il}}\and 
	Tianyi Zhang \thanks{Tel Aviv University, \href{}{tianyiz21@tauex.tau.ac.il}}}
\date{}

\maketitle
\thispagestyle{empty}

\end{comment}

\title{A Lossless Deamortization for Dynamic Greedy Set Cover \footnote{A preliminary version of this paper was accepted to the proceedings of FOCS 2024.}}

\author[1]{Shay Solomon\thanks{Funded by the European Union (ERC, DynOpt, 101043159). Views and opinions expressed are however those of the author(s) only and do not necessarily reflect those of the European Union or the European Research Council. Neither the European Union nor the granting authority can be held responsible for them. This research was also supported by the Israel Science Foundation (ISF) grant No.1991/1, and by a grant from the United States-Israel Binational Science Foundation (BSF), Jerusalem, Israel, and the United States National Science Foundation (NSF).}}
\author[1]{Amitai Uzrad\thanks{Funded by the European Union (ERC, DynOpt, 101043159). Views and opinions expressed are however those of the author(s) only and do not necessarily reflect those of the European Union or the European Research Council. Neither the European Union nor the granting authority can be held responsible for them. This research was also supported by the Israel Science Foundation (ISF) grant No.1991/1.}}
\author[1]{Tianyi Zhang\thanks{Funded by the European Union (ERC, DynOpt, 101043159).
Views and opinions expressed are however those of the author(s) only and do not necessarily reflect those of the European Union or the European Research Council. Neither the European Union nor the granting authority can be held responsible for them.}}
\affil[1]{Tel Aviv University}

\begin{document}

\date{\empty}

\begin{titlepage}
\def\thepage{}
\maketitle

\begin{abstract}
	The dynamic set cover problem has been subject to growing research attention in recent years.
In this problem, we are given as input a dynamic universe of at most $n$ elements and a fixed collection of $m$ sets, where each element appears in a most $f$ sets and the cost of each set is in $[1/C, 1]$, and the goal is to efficiently maintain an approximate minimum set cover under element updates.
 
 Two algorithms that dynamize the classic \emph{greedy} algorithm are known, providing $O(\log n)$ and $((1+\epsilon)\ln n)$-approximation with {\em amortized} update times  
 $O(f \log n)$ and $O(\frac{f \log n}{\eps^5})$, respectively [GKKP (STOC'17); SU (STOC'23)].
The question of whether one can get
approximation $O(\log n)$ (or even worse) with low \emph{worst-case} update time has remained open --- \emph{only the naive $O(f \cdot n)$ time bound is known}, even for unweighted instances.
 
In this work we 
devise the first amortized greedy algorithm that is amenable to an efficient deamortization, and also develop a {\em lossless} deamortization approach suitable for the set cover problem, 
the combination of which yields
a $((1+\epsilon)\ln n)$-approximation algorithm with a worst-case update time of $O\brac{\frac{f\log n}{\epsilon^2}}$.
\emph{Our worst-case time bound --- the first to break the naive $O(f \cdot n)$ bound ---  matches the previous best amortized bound, and actually improves its $\epsilon$-dependence}.

Further, to demonstrate the applicability of our deamortization approach, we employ it, in conjunction with the primal-dual amortized algorithm of [BHN (FOCS'19)],
to obtain a 
$((1+\epsilon)f)$-approximation algorithm with a worst-case update time of $O\brac{\frac{f\log n}{\epsilon^2}}$, improving over the  previous best bound of $O(\frac{f \cdot \log^2(Cn)}{\eps^3})$ [BHNW (SODA'21)].

Finally, as direct implications of our results for set cover, we (i) achieve the first nontrivial worst-case update time for the \emph{dominating set} problem, and (ii) improve the state-of-the-art worst-case update time for the \emph{vertex cover} problem. 
\end{abstract}
\end{titlepage}

\begin{spacing}{1.3}
	\tableofcontents
\end{spacing}

\thispagestyle{empty}
\clearpage
\pagenumbering{arabic}
\setcounter{page}{1}

\newpage

\section{Introduction}

In the classical set cover problem, the input is a set system $(\univ, \sets)$, where
$\univ$ is a universe of $n$ elements and $\sets$ is a family of $m$ sets 
$s \in \sets$ of elements in $\univ$, each with a cost $\cost(s)\in [\frac{1}{C}, 1]$.
The {\em frequency} $f$ of the set system $(\univ, \sets)$ is the largest number of sets in $\sets$ that any element in $\univ$ can possibly belong to. A subset of sets $\sets^\prime\subseteq \sets$ is called a {\em set cover} of $\univ$ if every element in $\univ$ resides in at least one set in $\sets^\prime$. The basic goal is to compute a {\em minimum set cover}, i.e., a set cover 
$\sets^*\subseteq \sets$ whose total cost $\cost(\sets^*) = \sum_{s\in \sets^*}\cost(s)$ is minimized. A well-known {\em greedy} algorithm  achieves a $(\ln n)$-approximation, and using a {\em primal-dual} approach one can obtain an $f$-approximation; these two approaches are believed to be optimal, as one cannot achieve a $(1-\eps) \ln n$-approximation unless P = NP \cite{williamson2011design,dinur2014analytical}, nor an $(f-\epsilon)$-approximation assuming the unique games conjecture \cite{khot2008vertex}.

There has been a recent growing endeavor to understand the set cover problem in the dynamic setting. In the dynamic set cover problem, we are given as input a dynamic universe $\univ$ of at most $n$ elements and a fixed collection $\sets$ of $m$ sets, and the goal is to maintain a set cover $\sets_\alg\subseteq \sets$ of small total cost, ideally matching the best approximation for the static setting, 
within a low update time.
%hile the universe $\univ$ undergoes element insertions or deletions. 
Given the aforementioned hardness results, one can hope for an approximation factor that approaches either $\ln n$ or $f$, while achieving an update time that approaches $O(f)$, which is the time required to specify an update explicitly.
%, as the input size of each element insertion is $O(f)$. 
Next, we survey the known results, distinguishing between the \emph{low-frequency regime} ($f = O(\log n)$) and the \emph{high-frequency regime} ($f = \Omega(\log n)$).

\paragraph{Low-Frequency Regime.} The vast majority of work on dynamic set cover has been devoted to the low-frequency regime, based on the primal-dual approach. An $O(f^2)$-approximation with $O(f\log(m+n))$ amortized update time was given in \cite{bhattacharya2015design}, and an $O(f^3)$-approximation with $O(f^2)$ amortized update time was given in \cite{gupta2017online}. A near-optimal approximation of $(1+\epsilon)f$ for the unweighted setting ($C \equiv 1$) was achieved for the first time in \cite{abboud2019dynamic}, with (expected) amortized update time $O(f^2\log n / \epsilon)$, which was  improved to (expected) amortized update time $O(f^2)$ (without any $\eps$-dependency). The randomized algorithms of \cite{abboud2019dynamic,assadi2021fully} were
strengthened to the general weighted setting via deterministic algorithms with similar update time, still for the near-optimal approximation of $(1+\eps)f$ \cite{bhattacharya2019new,bhattacharya2021dynamic}. Very recently,  this line of work on primal-dual algorithms with amortized time bounds culminated %\tianyi{I think the right grammar should remove `was'. Maybe `improved' not `culminated', as it is not optimal yet} 
in a $((1+\eps)f)$-approximation algorithm that achieves a near-optimal amortized update time of $O\left(\frac{f}{\epsilon^3}\log^*f  + \frac{f\log C}{\epsilon^3}\right)$ \cite{bukov2023nearly}.

The algorithm of \cite{bhattacharya2019new} yields an amortized update time of $O(\frac{f \cdot \log(Cn)}{\eps^2})$, and it is an \emph{inherently global} algorithm, in the sense that (1) it allows the underlying invariants to be violated to some extent \emph{in a global way} (i.e., in some average sense), 
and (2) it applies ``global clean-up'' procedures to restore the invariants.
Importantly, the global nature of that algorithm is what makes it amenable to efficient \emph{deamortization}, as done in \cite{bhattacharya2021dynamic} to obtain a deterministic $((1+\epsilon)f)$-approximation 
algorithm with $O(\frac{f\log^2(Cn)}{\epsilon^3})$ \emph{worst-case} update time; note that the deamortization of \cite{bhattacharya2021dynamic} loses a factor of $\frac{\log(Cn)}{\eps}$ in the update time. (See \Cref{comp} for a summary of the results.)

%Summarizing, there has been extensive work on dynamic set cover in the low frequency regime in the past decade, with some breakthrough papers achieving near-optimal results.

%Present the main results in this regime. Say that there are several known results that are good. When getting to BHN, say that it is fully global and thus amenable to efficient deamortization (with short explanation of what fully global is, from previous abstract). 

\paragraph{High-Frequency Regime.} In contrast to the low-frequency regime, only two algorithms that dynamize the classic \emph{greedy} algorithm are known, achieving $O(\log n)$- and $((1+\epsilon)\ln n)$-approximation with {\em amortized} update times $O(f \log n)$ and $O(\frac{f \log n}{\eps^5})$, respectively \cite{gupta2017online,solomon2023dynamic}. 

It seems inherently harder to dynamize the greedy algorithm (in the high-frequency regime), as compared to the primal-dual algorithm (in the low-frequency regime). 
We will try to substantiate this claim in the technical overview of  \Cref{tech}; however, the large gaps between the state-of-the-art results in the two regimes may already provide partial evidence.
%First, the gaps between the known results are quite significant.
For amortized bounds, 
the algorithm of \cite{bukov2023nearly}
in the low-frequency regime
incurs only a tiny extra $\log^*f \le \log^* n$ factor in the update time over the ideal $O(f)$ time bound 
(ignoring the dependencies on $\eps$ and $C$),
whereas the algorithms of \cite{gupta2013fully,solomon2023dynamic}
%In \cite{solomon2023dynamic}, although the approximation factor is near-optimal, the amortized update time is
incur an extra $\log n$ factor. 
For {\em worst-case} bounds, the algorithm of \cite{bhattacharya2021dynamic} in the low-frequency regime provides a low {worst-case} update time,
%On the other hand,
whereas the question of whether one can get approximation $O(\log n)$ (or even worse) with low \emph{worst-case} update time has remained open; only the naive $O(f \cdot n)$ time bound is known, even for unweighted instances.
Technically speaking, the algorithms of \cite{gupta2017online,solomon2023dynamic} in the high-frequency regime apply ``local clean-up'' procedures whenever {\em any} invariant is violated, which is problematic to deamortize; alas, in contrast to the low-frequency regime, designing a dynamic greedy algorithm  of global nature seems highly challenging, as discussed in detail in \Cref{tech}. 

\paragraph{Focus.} This work focuses on the dynamic set cover problem with {\em worst-case} update time, primarily in the {\em high-frequency} regime --- where no nontrivial worst-case time bound is known. One may consider the gaps in our understanding of the dynamic set cover problem with worst-case time bounds from two different perspectives: 

\begin{enumerate}
\item In the high-frequency regime, the gap between the state-of-the-art amortized ($O(\frac{f \log n}{\eps^5})$ \cite{solomon2023dynamic}) and worst-case (the naive $O(f \cdot n)$) time bounds. 
\item For the state-of-the-art worst-case time bounds, the gap between the low-frequency 
($O(\frac{f\log^2(Cn)}{\epsilon^3})$ \cite{bhattacharya2021dynamic})
and the high-frequency (the naive $O(f \cdot n)$) regimes. 
\end{enumerate}
%there is somewhat of a large gap between the known results for each regime, mainly regarding non-trivial \emph{worst case} update times, where in the low frequency regime the authors of \cite{bhattacharya2021dynamic} obtained an algorithm yielding a near-optimal approximation factor of $(1+\epsilon)f$ in $O(f\log^2(Cn)/\epsilon^3)$ worst case update time, whereas in the high frequency regime there is no known algorithm yielding a non-trivial approximation factor in any non-trivial worst case update time. 

\noindent The following fundamental question naturally arises:

\begin{tcolorbox} [width=\linewidth, sharp corners=all, colback=white!95!black]
\begin{question} \label{q1}
Can one achieve an approximation of $O(\log n)$ (or even worse) for dynamic set cover with any nontrivial \emph{worst-case} update time? 
%Furthermore, can one achieve the near-optimal approximation of $(1+\eps)\ln n$  with a worst-case update time that matches the best amortized bound?
\end{question}
\end{tcolorbox}

%\noindent In addition, due to the gap between both regimes, we are interested in the following secondary question:

\subsection{Our Contribution} 

This work provides the first dynamization of the greedy algorithm with a low {\em worst-case} update time. To this end:
\begin{enumerate}
    \item We first overcome the aforementioned challenge by presenting the first amortized greedy algorithm of {\em global} nature; see \Cref{fullyg} for the details. 
    \item Second, we develop a {\em lossless}  deamortization approach, i.e., the resulting worst-case time-bound is just as good as the best amortized bound; see \Cref{sec:deam} for the details. 
\end{enumerate}
    
\noindent By employing our deamortization approach in conjunction with our new  global amortized algorithm, we obtain the following main result of this work (see \Cref{comp} for a summary of results).

\begin{theorem} [High-frequency set cover] \label{wc}
For any set system $(\univ, \sets)$ that undergoes a sequence of element insertions and deletions, where the frequency is always bounded by $f$, and for any $\epsilon \in (0, \frac{1}{4})$, there is a dynamic algorithm that maintains a $((1+\epsilon)\ln n)$-approximate minimum set cover in $O\brac{\frac{f\log n}{\epsilon^2}}$ deterministic worst-case update time.
\end{theorem}

\noindent Not only does \Cref{wc} resolve \Cref{q1} in the affirmative, but it also achieves \EMPH{optimal} bounds on both the approximation factor and the worst-case update time, given the current state-of-the-art amortized result, excluding the $\epsilon$ dependencies.
%, where not only does the algorithm yield a near-optimal approximation factor of $(1+\epsilon)\ln n$ in a non-trivial worst case update time, but 
Moreover, %In fact, 
%not only does the worst-case update time match the previous best amortized bound,
%even when considering the $\epsilon$ dependencies, 
our worst-case update time actually improves the  $\epsilon$-dependence
of the previous best amortized bound \cite{solomon2023dynamic}
from $\eps^{-5}$ to $\eps^{-2}$. 
Therefore, we achieve an \EMPH{optimal deamortization} of the previous best amortized algorithm in the high-frequency regime. We stress that while the deamortization itself is optimal, the update time bound of $O_\eps(f \log n)$ is not necessarily optimal;  whether this time bound can be improved (even for amortized bounds)
remains an intriguing open question. 

To demonstrate the applicability of our deamortization approach, we employ it, in conjunction with the aforementioned amortized algorithm of \cite{bhattacharya2019new} in the low-frequency regime,
to obtain the following result,
%a 
%$((1+\epsilon)f)$-approximation algorithm with a worst-case update time of $O\brac{\frac{f\log n}{\epsilon^2}}$, 
which improves over the  worst-case time bound of $O(\frac{f \cdot \log^2(Cn)}{\eps^3})$ \cite{bhattacharya2021dynamic}, first by shaving a factor of $\frac{\log(C n)}{\eps}$, and then by removing the dependency on the aspect ratio $C$.

\begin{theorem} [Low-frequency set cover]
\label{wc2}
For any set system $(\univ, \sets)$ that undergoes a sequence of element insertions and deletions, where the frequency is always bounded by $f$, and for any $\epsilon \in (0, \frac{1}{4})$, there is a dynamic algorithm that maintains a $((1+\epsilon)f)$-approximate minimum set cover in $O\brac{\frac{f\log n}{\epsilon^2}}$ deterministic worst-case update time. 
 \end{theorem}

We note that our deamortization approach that proves \Cref{wc} in the high-frequency regime {\em seamlessly extends} to prove \Cref{wc2} in the low-frequency regime.
Consequently, we provide a {\em unified algorithmic approach} to the dynamic set cover problem with worst-case time bounds.
We emphasize that our approach is naturally suitable specifically for the set cover problem. It would be interesting to explore the possibilities of extending our approach beyond the set cover problem; we leave this as an intriguing open question. 
Nonetheless, the set cover problem is a fundamental covering problem, which encapsulates several other important problems. As such, we believe that an approach suitable for set cover is of rather general interest.
In particular, our approach leads directly to the following implications for the (minimum) \emph{dominating set} and \emph{vertex cover} problems.
%of our approach to these problems. 

%\color{blue}
In the minimum dominating set problem, we are given a graph $G = (V,E)$, where $n=|V|$, and each vertex has a cost assigned to it. The goal is to find a subset of vertices $D \subseteq V$ of minimum total cost, such that for any vertex $v \in V$, either $v \in D$ or $v$ has a neighbor in $D$. In the dynamic setting, the adversary inserts/deletes an edge upon each update step. 
%\color{black}
We derive the result for the dominating set problem via a simple reduction to the set cover problem (described in \Cref{dssec}), which allows us to use our set cover algorithm  provided by \Cref{wc} as a black box. 

\begin{theorem} [Dominating set] \label{ds-final}
For any graph $G = (V,E)$ that undergoes a sequence of edge insertions and deletions, where the degree is always bounded by $\Delta$, and for any $\epsilon \in (0, \frac{1}{4})$, there is a dynamic algorithm that maintains a $((1+\epsilon)\ln \Delta)$-approximate minimum weighted dominating set in $O\brac{\frac{\Delta \log n}{\epsilon^2}}$ deterministic worst-case update time.
\end{theorem}

%\color{blue}
\noindent We note that \Cref{ds-final} provides the {\em first} non-trivial worst-case update time algorithm for the (unweighted or weighted) minimum dominating set problem 
%with better-than-$\Delta$ approximation 
\footnote{One could have used our simple reduction from dynamic dominating set to dynamic set cover, 
%Although it was not reported in the literature before ,one can use 
in conjunction with the worst-case primal-dual set cover algorithm in \cite{bhattacharya2021dynamic} as a black-box, to obtain a $(1+\epsilon) \cdot \Delta$ approximation
with a worst-case update time of $O(\frac{\Delta \cdot \log^2(Cn)}{\eps^3})$.
Such a result has not been reported in the literature, but more importantly, its approximation ratio $\approx \Delta$ is far worse than the approximation ratio $\approx \ln \Delta$ that we aim for.}
%\color{black}
as with our set cover results, there is no dependence whatsoever on the costs. Our worst-case time bound matches the previous best amortized bound for the problem \cite{solomon2023dynamic}, and it also improves its $\epsilon$-dependence from $\eps^{-5}$ to $\eps^{-2}$.

%\paragraph{Dominating Set.} 
%In \Cref{dssec} we derive \Cref{ds-final} 
%Moreover, this reduction greatly simplifies the work of \cite{solomon2023dynamic} on amortized set cover and dominating set, that had to deal with many tedious adaptations. The main difference between the set cover problem and the dominating set problem lies within the dynamic setting. In the set cover problem the adversary inserts (activates) or deletes (deactivates) an element upon each update step, whereas in the dominating set problem the adversary inserts/deletes an edge upon each update step, which can be thought of as creating or removing a connection between an element to a set in the set cover problem. Since this operation does not exist in the basic dynamic setting of the set cover problem, we will treat such an operation as an element deletion (deactivation) followed by an insertion (activation) of the same element, just with a different collection of sets that can cover the element, which will be the same collection plus (respectively, minus) the created (resp., removed) connection. 

%\color{blue}

%The following result for the vertex cover problem is obtained 
Next, for the minimum (weighted) {\em vertex cover} problem, by setting $f=2$ in \Cref{wc2}, we directly get an improvement of the  state-of-the-art worst-case update time bound for $(2+\eps)$-approximate vertex cover:  from $O(\frac{\log^2(Cn)}{\eps^3})$ \cite{bhattacharya2021dynamic} to $O\brac{\frac{\log n}{\epsilon^2}}$.

\begin{table}
	\begin{tabular}{|c|c|c|c|c|}
		\hline
		reference	&	approximation	&	update time	&	worst-case?	&	weighted?\\\hline

\cite{gupta2017online}	&	$O(\log n)$	&	$O(f\log n)$	&	no		&	yes	\\\hline
		\cite{solomon2023dynamic}	&	$(1+\epsilon)\ln n$	&	$O\brac{\frac{f\log n}{\epsilon^5}}$	&	no	&	yes\\\hline 
		\textbf{new}	&	$(1+\epsilon)\ln n$	&	$O\brac{\frac{f\log n}{\epsilon^2}}$	&	\textbf{yes}	&	yes\\\hline\hline\hline

		\cite{bhattacharya2015design}	&	$O(f^2)$	&	$O(f\log(m+n))$	&	no	&	yes	\\\hline
		\cite{gupta2017online,bhattacharya2017deterministic}	&	$O(f^3)$	&	$O(f^2)$	&	no	&	yes\\\hline
		\cite{abboud2019dynamic}	&	$(1+\epsilon)f$	&	$O\brac{\frac{f^2}{\epsilon^5}\log n}$	&	no	&	no\\\hline
		\cite{bhattacharya2019new}	&	$(1+\epsilon)f$	&	$O\brac{\frac{f}{\epsilon^2}\log(Cn)}$	&	no	&	yes	\\\hline
		\cite{bhattacharya2021dynamic}	&	$(1+\epsilon)f$	&	$O\brac{\frac{f^2}{\epsilon^3} + \frac{f}{\epsilon^2}\log C}$	&	no	&	yes	\\\hline
		\cite{bhattacharya2021dynamic}	&	$(1+\epsilon)f$	&	$O\brac{\frac{f\log^2(Cn)}{\epsilon^3}}$	&	\textbf{yes}	&	yes\\\hline
		\cite{assadi2021fully}	&	$f$	&	$O\brac{f^2}$	&	no	&	no\\\hline
		\cite{bukov2023nearly}	&	$(1+\epsilon)f$	&	$O\left(\frac{f}{\epsilon^3}\log^*f  + \frac{f\log C}{\epsilon^3}\right)$	&	no	&	yes\\\hline
  \textbf{new}	&	$(1+\epsilon)f$	&	$O\brac{\frac{f\log n}{\epsilon^2}}$	&	\textbf{yes}	&	yes\\\hline

		%\textbf{new}	&	$(1+\epsilon)f$	&	$O\brac{\frac{f\log n}{\epsilon^2}}$	&	\textbf{yes}	&	yes\\\hline

	\end{tabular}
	\caption{Summary of results on dynamic set cover.}\label{comp}
\end{table}

\section{Technical Overview} \label{tech}

In this section we give a technical overview of our contribution.
In \Cref{known} we set up the ground by surveying the known techniques and approaches.
In \Cref{ourapp} we discuss the main technical challenges left open by previous work, and then turn to presenting the key technical novelty behind our work and demonstrating how it overcomes the main challenges. Along the way, we try to convey some conceptual highlights of this work.
We refer to Sections \ref{threes}, \ref{remove} and \ref{primdu} for the full, formal details.

\subsection{The Known Amortized Algorithms} \label{known}

\paragraph{Hierarchical Data Structure.} 
%Let us start with the hierarchical approach from \cite{gupta2017online}. 
%In the original approach from \cite{gupta2017online}, 
Every set $s\in\sets$ is assigned a level value in the range $[-1,O(\log(Cn))]$, where $-1$ is reserved for sets not in the cover. Every element $e\in \univ$ is assigned to a unique set $\asn(e)$ in the dynamic set cover solution, where $e$ shares the same level $\lev(e) = \lev(\asn(e))$ as the set to which it is assigned; inversely, we have the {\em cov(ering) set} $\cov(s)=  \{e\mid \asn(e) = s \}$ of $s$, which consists of all elements assigned to set $s$. 

\subsubsection{The Fully Local Approach}

In the original approach from \cite{gupta2017online}, their algorithm  maintains the following invariant.

\begin{invariant}[$O(\log n)$-approximation, \cite{gupta2017online}]\label{local}
	The following two conditions regarding the hierarchical structure hold at any time (i.e., before any update step).
\begin{enumerate}[(1)]
	\item For any set $s$ in the current solution, it holds that $|\cov(s)| / \cost(s) \in [2^{\lev(s)}, 2^{\lev(s)+10}]$.\footnote{In \cite{gupta2017online} the levels are negative and they consider the ratio $\frac{\cost(s)}{|\cov(s)|}$. In this paper we will consider the inverse ratio $\frac{|\cov(s)|}{\cost(s)}$ and so the levels will be positive. The two are completely equivalent.} 
	\item For any set $s$ and level $k$, $\{e\mid e\in s, \lev(e) = k\}$ has size at most $2^{k+10}\cdot \cost(s)$.
\end{enumerate}	
\end{invariant}

\noindent \cite{gupta2017online} used what we shall refer to as a \emph{fully local} approach to maintain both conditions of \Cref{local} 
%of the hierarchical data structure 
at any time; namely, whenever \Cref{local}(1) or \Cref{local}(2) is violated, even for a single set $s$, the algorithm performs a {\em local change}, which aims at restoring the condition for set $s$. At the core of such a local change --- which we shall refer to as a {\em local fall} or {\em local rise} of $s$ (depending on whether the level of $s$ increases or decreases) --- is a change to the level of $s$, which is accompanied with changes to levels of elements that join or leave $\cov(s)$. 
%involves a change to the levels of some elements associated with the violation for set $s$, and thus also a change to the level of $s$,
%aiming to restore the condition for $s$. 
%We sometimes refer to such changes in levels (for both elements and sets) as \emph{falls} and \emph{rises}, and we may also refer to the fall/rise of a single set as a \emph{local fall/rise}, where a single local fall/rise triggers possibly many falls/rises of elements.
Of course, a local fall/rise of a single set
may trigger further violations of the conditions, which are handled by performing further local falls and rises. The resulting cascade of local falls and rises is repeated   
until the conditions hold.
%as necessary
%and so the algorithm handles any further violations iteratively by performing a cascade of rises and falls, 
%(if any exist).

\subsubsection{The Partially Global Approach: From $O(\log n)$ to $(1+\epsilon)\ln n$ Approximation} As observed in \cite{solomon2023dynamic}, \Cref{local} has some inherent barriers against achieving a $(1+\epsilon)\ln n$ approximation. Therefore, in \cite{solomon2023dynamic}, a different set of conditions were proposed in order to optimize the constant factor preceding $\ln n$ to $1+\eps$, as given in the following invariant. We remark that this invariant was not maintained by \cite{solomon2023dynamic}; only a relaxation of the invariant was maintained, as discussed below. 

%\shay{(2) wasn't maintained (only globally), so maybe rephrase. And as for $N_k(s)$, I changed $\lev(e) \le k$ (and increased the upper bound accordingly), aiming to better coincide with the current paper and GKKP}.

\begin{invariant}[$(1+\epsilon)\ln n$-approximation, \cite{solomon2023dynamic}]\label{global}
	Set $\beta = 1+\epsilon$. The following two conditions regarding the hierarchical structure hold at any time.
	\begin{enumerate}[(1)]
		\item For any set $s$ in the current solution, it holds that $|\cov(s)| / \cost(s) \geq \beta^{\lev(s)-1}$.
		\item For any set $s$ and level  $k$, $N_k(s) = \{e\mid e\in s, \lev(e) \le k\}$ has size less than $\beta^{k+2}\cdot \cost(s)$.
	\end{enumerate}	
\end{invariant}

\noindent There are several differences between
\Cref{global} and \Cref{local}, which are 
 crucial for improving the  approximation from $O(\log n)$ to $(1+\eps)\ln n$.
One difference is the usage of $\beta = 1+\eps$ rather than $2$.
Another difference 
%between \Cref{global} and \Cref{local} 
lies in the second condition: While in \Cref{local} it bounds the
number of elements in a set $s$ at level exactly $k$, \Cref{global} provides a stronger bound on the 
total number of elements 
belonging to $s$ at all levels $\le k$.\footnote{In \cite{solomon2023dynamic}, $N_k(s)$ is defined as $N_k(s) = \{e\mid e\in s, \lev(e) < k\}$ and the upper bound on $|N_k(s)|$ is $\beta^{k+2}\cdot \cost(s)$; this is of course an equivalent formulation (where $k$ is replaced by $k+1$).} 
%This stronger bound is crucial
Clearly, \Cref{global} is stronger than \Cref{local}, and it turns out to be problematic to maintain efficiently.

%The main difference between the algorithms of \cite{solomon2023dynamic} and \cite{gupta2017online}, which was 
The key behind the improvement of \cite{solomon2023dynamic}
to the approximation factor, while achieving the same amortized update time,
%from $O(\log n)$ to $(1+\epsilon)\ln n$ within low amortized update time in 
is to abandon the fully local approach of \cite{gupta2017online},
%the {\em algorithmic side} of implementing these invariants. 
%In \cite{gupta2017online}, 
which performs a cascade of local falls and rises 
until both conditions of the invariant are maintained, following any update step. 
%as long as any set violates any condition
Instead, the approach taken by \cite{solomon2023dynamic}, 
which we shall refer to as \emph{partially global}, is to maintain only the second condition of the invariant for any set;
%only the second property is always preserved by local procedures; 
that is, whenever there is any violation of \Cref{global}(2), the algorithm perofrms a local rise. %performs changes to some elements and sets around the violation, called a \emph{rise}. 
On the other hand, the first condition is only maintained in a \emph{global} manner in \cite{solomon2023dynamic}; more specifically, the algorithm waits until \Cref{global}(1) is widely violated in many places in the hierarchical structure, and then performs a \emph{reset} procedure on a carefully chosen part of the hierarchical structure --- which amounts to running the standard greedy algorithm on that part --- to restore \Cref{global}(1); roughly speaking, \Cref{global}(1) only holds in an {\em average sense} (or for an average set), and does not necessarily hold for any set $s$. The authors of \cite{solomon2023dynamic} prove that the approximation factor is $(1+\epsilon)\ln n$ even by assuming that \Cref{global}(1) only holds in an average sense, in a proof that follows closely the standard proof of $\ln(n)$-approximation for the classic static greedy algorithm.

Summarizing:
\begin{itemize}
    \item The fully local algorithm of \cite{gupta2017online} {\em locally} maintains
both conditions of the invariant,
by persistently performing local falls and rises to sets that violate the conditions.
\item In the partially global algorithm of \cite{solomon2023dynamic}, only the second condition is locally 
maintained, by performing local rises to sets that violate it. On the other hand, the first condition is maintained \emph{globally}, which in particular means that {\em no local falls occur}.
\item In both previous algorithms \cite{gupta2017online,solomon2023dynamic}, the proofs of the approximation factor follow rather closely known analyses of the classic greedy algorithm.
%; of course, one can view the reset procedure as performing falls, but this is the only way that sets may fall in the algorithm of \cite{solomon2023dynamic}. 
\end{itemize}

%We shall therefore refer to the algorithm of \cite{gupta2017online} as \emph{fully local} since both criteria of the invariant are maintained \emph{locally}, meaning rises and falls of specific sets occur, and we shall refer to the algorithm of \cite{solomon2023dynamic} as \emph{partially global} since the first criterion is maintained \emph{globally}, meaning no specific falls occur, unless if part of a reset, and the second is maintained \emph{locally}, because rises of specific sets still occur.

\subsection{Our Approach} \label{ourapp}

\subsubsection{A Fully Global Amortized Algorithm} \label{fullyg}
We remind that our goal is to obtain the \emph{first} greedy-based set cover algorithm with a low \emph{worst-case} update time. The naive algorithm would recompute from scratch the greedy algorithm on the entire system following each update step, but this yields an update time of $O(f \cdot n)$. 
%No other algorithm for the problem is known.
%It is not even clear whether one can obtain an $O(\log n)$-approximation (and even worse) within worst case update time that is better than the time of recomputing a greedy set cover from scratch following every update step. 
To achieve a low  worst-case update time, the first suggestion that comes to mind is to try and deamortize one of the aforementioned amortized algorithms \cite{gupta2017online,solomon2023dynamic}.
%de-amortize an amortized algorithm, by duplicating data structures, and working in the background. 
As mentioned, the  algorithms 
\cite{gupta2017online} of \cite{solomon2023dynamic}
are fully local and partially global, respectively; in particular, both algorithms perform local rises, for any set that violates the second condition of the corresponding invariant.
%--- in order to bound the approximation factor.
The running time 
%a local rise 
%is exponential in the level of the rise, i.e., 
%the time 
of a local rise of any set $s$ to level $j$ is at least linear in the number of elements that join $\cov(s)$, which is by design around $\beta^{j}$. 
%hus one cannot complete even a single high-level local rise, when aiming at a low worst-case update time.
Thus, de-amortizing the algorithms of \cite{gupta2017online,solomon2023dynamic} with a low worst-case update time
implies that one cannot complete even a single high-level local rise.
%must trigger a significant delay, 
Of course, one can perform the required local rises with a sufficient amount of delay, by maintaining a queue of all sets that violate the second invariant and handling them one after another, 
however delaying even a single local rise may blow up the approximation factor; e.g., consider an extreme (unweighted) case where each element is covered by a singleton set at level 0, yet there is 
a single set $s$ that contains all $n$ elements, which needs to perform a local rise to level $\log n$, as a result of which each element will have left its singleton covering set and joined $\cov(s)$. 
We note that this extreme case, which incurs the worst-possible approximation of $n$ (for unweighted instances), is ``invalid'', in the sense that it shouldn't have been created in the first place, as $s$ should have made a local rise to cover many elements well before all of them have been inserted; however, one can embed this invalid instance inside larger instances in obvious ways to create various instances of the same flavor that incur very poor approximation. 

%A natural suggestion would be to prioritize and classify the local rises not just based on the time in which the respective sets violated the second invariant and got added to the queue, but also based on the level of the rise.
%Indeed, performing a single high-level rise could take more time than performing many low-level rises, so it makes sense to give a higher priority to lower levels
%Dealing with this issue is highly nontrivial, because even just 
%are highly non-trivial, for one main reason which is that in both we must deal with a
%basically boils down to playing a ``three-player game" \shay{stop here}. Meaning, in \cite{gupta2017online} we would have to deal with the adversary (element updates), rises and falls, and in \cite{solomon2023dynamic} we would have to deal with the adversary, rises and resets. Rises and falls, or rises and resets could be costly operations (taking time $\Omega(f \cdot n)$), thus they would have to be simulated carefully in the background, possibly over several update steps. Dealing and prioritizing between both in the background and on top of that having adversary updates can cause a great deal of chaos in the background, since the algorithm in the background must be able to deal with three simultaneous operations, and manage to prioritize between the three. 

A natural two-step strategy would therefore be to first obtain a \emph{fully global} amortized algorithm, where we eliminate not just the local falls as in \cite{solomon2023dynamic}, but also the local rises ---
so that both conditions of \Cref{global}
will be maintained only in a global average sense, and in particular they may be violated locally by some sets;
moreover, the conditions of \Cref{global} will be restored only through a global {\em reset} procedure on a carefully chosen part of the hierarchical structure. The second step would be to de-amortize the resulting fully global algorithm, which seems much more natural and promising than de-amortizing the fully local or partially global algorithms  \cite{gupta2017online,solomon2023dynamic}. 
Such a two-step strategy 
%Indeed, efficient deamortizations of global algorithms is a standard approach, which 
was employed before in a similar context:
%for achieving a primal-dual set cover algorithm with low worst-case bounds: 
\begin{enumerate}
\item Bhattacharya \etal \cite{bhattacharya2019new}
dynamized the {\em primal-dual} $f$-approximation algorithm to achieve  $((1+\eps)f)$-approximation with an %{\em worst-case} 
 %or 
 amortized %update times of 
 update time of $O(\frac{f \cdot \log(Cn)}{\eps^2})$, via a fully global algorithm --- which, similarly to the above, may violate the conditions of the underlying invariant locally by some sets, and only tries to satisfy them in a global sense, and to restore them through a global reset procedure on a carefully chosen part of the hierarchical structure. \emph{As mentioned, the two known dynamic greedy algorithms are not fully global.}

 \item  Bhattacharya \etal \cite{bhattacharya2021dynamic} deamortized the fully global amortized algorithm of \cite{bhattacharya2019new}, to obtain a worst-case update time of  $O(\frac{f \cdot \log^2(Cn)}{\eps^3})$.
% the global nature of the amortized algorithm is key to the efficient deamortization.
 \emph{The fully global amortized algorithm satisfies some ``nice'' properties, which are amenable to deamortization; it is unclear if similar properties can be achieved for an amortized greedy algorithm.
 Moreover, the deamoritzation of \cite{bhattacharya2021dynamic} loses a factor of $\frac{\log(Cn)}{\eps}$.}
\end{enumerate}
 Two challenges arise:
 \\\emph{Challenge 1. Fully global amortized algorithm: Primal-dual is easier than greedy.}
 The conditions in the invariant of the primal-dual algorithm of \cite{bhattacharya2019new} are the complementary slackness conditions, which are easier to maintain than the conditions in \Cref{global}; in particular, the analog complementary slackness condition to 
 \Cref{global}(2) (see \Cref{primdu}) is to upper bound the weight $\omega(s)$ of any set $s$ by its cost $\cost(s)$, where $\omega(s) = \sum_{e \in S} \omega(e)$ and $\omega(e)$ is  basically $\beta^{-\lev(e)}$ ($\lev(e)$ is the dynamic level of $e$). On the other hand, the condition in 
  \Cref{global}(2) applies not just to any set, but also to every possible level $k$, which makes it inherently more difficult to maintain. 
  \\\emph{Challenge 2. A lossless deamoritzation.}
%  Once we have a fully global amortized greedy algorithm, we would like to deamortize it. 
As the greedy algorithm appears to be inherently more difficult to dynamize and ``globalize'' than the primal-dual algorithm, it is only natural to expect that the task of deamortizing an amortized greedy algorithm would be harder than for a primal-dual algorithm. Moreover, our goal is to attain a {\em lossless} deamortization, where the worst-case update time does not exceed the amortized bound by a factor of $\frac{\log(Cn)}{\eps}$, as in \cite{bhattacharya2021dynamic}.
Next, we describe Challenge 1 in more detail, and highlight the main insights that we employed in order to overcome it.
The discussion on  Challenge 2 is deferred to \Cref{sec:deam}.
%\shay{maybe should say that Challenge 1 is the main one? (not sure we should).} %\textcolor{orange}{Amitai: Leaving this comment here for now.}

%\shay{I omitted a para, but should mention around here the issue with performing ``partial resets'' (this notion needs to be introduced), i.e.,  up to some level (rather than resetting the whole instance), and the ``tension'' between partial resets and local rises (what should we do with local rises that want to rise much higher than the reset level)}

As mentioned, in \cite{solomon2023dynamic} the resets are executed on only part of the system. To be more precise, they execute a reset up to some \emph{critical} level, which amounts to running the standard static greedy algorithm only on sets and elements that their level is up to the critical level. In a sense, it just ``reshuffles" the system up to that critical level, and this does not clean up the whole system obviously, but the authors show that such a reset does clean up enough for the approximation factor to hold, and that the system has obtained enough ``credits" for each set and element up to the critical level to change levels in the reset. Thus, to obtain a \emph{fully global} amortized algorithm, it seems necessary to use this idea of resets only up to certain levels. 
%One issue with this approach is that we can have several resets up to different levels working simultaneously. This demands dealing with conflicts and inconsistencies in an efficient way if we want to obtain a \emph{lossless} deamortization. As mentioned, this issue will be addressed in \Cref{sec:deam}. The second issue, 0

It turns out that ``globalizing" \Cref{global}(2) is inherently different and harder than globalizing \Cref{global}(1), which is perhaps the reason that the authors of \cite{solomon2023dynamic} settled for a partially global algorithm rather than a fully global one. 
First, let us compare the effect to the approximation factor, of postponing local falls versus that of postponing local rises; recall that local falls and rises correspond to the first and second conditions of \Cref{global}, respectively. 
To simplify the discussion, consider the unweighted case. 
%The main reason is the effect that postponing rises has on the approximation factor versus the effect that postponing falls has on the approximation factor. 
If there exists a set $s$ that violates \Cref{global}(1) (and needs to perform a local fall), even by a lot --- in the extreme case $|\cov(s)| = 0$, then this will not have a direct effect on other sets, 
%and if they obey the two conditions of \Cref{global} then 
and at worst we have caused the set cover size to grow by one (by having $s$ in the set cover even though it may not need to be there). Consequently, one can define a global violation to
\Cref{global}(1) for each prefix of levels
in the obvious way (whenever an $\eps$-fraction of the sets up to level $k$ violate the condition, this prefix is ``dirty''), and it is not difficult to show that the approximation is in check as long as no prefix of levels is dirty.
In contrast, if there is a set $s$ that violates \Cref{global}(2) as $|N_k(s)| \geq \beta^{k+1}$, and its local rise to level $k$ is postponed, this
could affect many sets, since each element in $N_k(s)$ may be  covered by a different set, and also $s$ might not even be in the solution currently. Moreover, a single local rise could create possibly many sets that violate \Cref{global}(1), and they may all become empty following the rise. So one local rise may create possibly many sets that violate \Cref{global}(1), by a lot. Moreover, those violated sets may lie in multiple levels, which makes it harder to quantify the dirt across one level.
If we again consider the extreme case, where $|N_k(s)| = n$, then obviously the optimal set cover size is one, and by not executing the rise our maintained solution can be arbitrarily larger, as discussed in the beginning of \Cref{fullyg}. And indeed, the approximation factor analysis of both \cite{gupta2017online} and \cite{solomon2023dynamic} rely heavily on the fact that \Cref{local}(2) and \Cref{global}(2) (respectively) hold locally for each set. If we aim for a globalization of this condition,   we need to meet three objectives:

\begin{enumerate}

\item We first need to define a \emph{global} notion of ``dirt", meaning a global measure that determines how far off we are from the ``ideal guarantee" --- where each set obeys locally both conditions of \Cref{global}. This definition must take into account local rises that are being postponed (in contrast to \cite{solomon2023dynamic} --- and this is the hard part), and we want this global notion of dirt to be defined for any level, and in particular for any \emph{prefix} of levels (all levels up to a certain level --- as in \cite{solomon2023dynamic}), in order to determine a ``critical level" to do a reset up to. Meaning, we need to be able to determine whether ``the system up to some level $k$ is dirty" or not.

\item Next, we need to come up with a ``global algorithm'', which  would correspond to the global notion of dirt, and in particular would maintain a relevant global invariant by cleaning up the dirt \emph{globally} via \emph{resets} up to a certain critical level.
%, without ever performing any local rises (much) higher than the reset level.

\item Lastly, an \emph{inherently} different approximation factor analysis seems to be necessary, since the known ones crucially rely on the validity of the second condition of the invariants (\Cref{local} or \Cref{global}) for every set; we need a new argument that would correspond to the new global invariant, which is defined by the new notion of global dirt.

\end{enumerate}

%\noindent To summarize, we need to find a way to define \emph{global dirt}, then formulate invariants and an algorithm to maintain them corresponding to this definition of dirt, and lastly prove a good approximation factor given that the aforementioned invariants are satisfied.

%\shay{in the subsequent paragraphs, discuss how each of the three objectives is met}

\paragraph{Naive Attempt.~}
The first attempt for meeting the first objective is to use a binary distinction between  {\em active} and {\em passive} elements.\footnote{This terminology of active and passive elements is from \cite{bhattacharya2021dynamic}.
We believe it is instructive to use the same terminology,
even though our definitions of active/passive are not the same as \cite{bhattacharya2021dynamic}, 
%both to pay tribute to \cite{bhattacharya2021dynamic}, but especially 
since we aim at achieving a unified deamorization approach, applicable also to the low-frequency regime.}
We will say that each element upon insertion is passive, and once it participates in a reset it becomes active. We shall consider each passive element at level $k'$ as a ``dirt unit" at level $k'$. Once the number of dirt units \emph{up to} level $k$ surpasses an $\epsilon$-fraction of the total number of elements up to level $k$, we say that the system is {\em $k$-dirty}. 

For the second objective, we will maintain the invariant that the system is never $k$-dirty for any $k$. To do so, we will execute a reset (static greedy algorithm) on the subsystem of elements and sets that lie up to level $k$ immediately when the system becomes $k$-dirty. Following this reset, by definition of our global dirt, all elements up to $k$ become active, which cleans up all dirt up to level $k$, and the invariant holds. Since passive elements are newly inserted elements that have not yet participated in a reset, if each inserted element arrives with $\frac{1}{\epsilon}$ credits, then we would have one credit for each element participating in a reset. 

The problem with this naive suggestion lies within the third objective, meaning the approximation factor may blow up. Immediately following a reset up to level $k$ indeed an element $e$ participating in this reset does not want to be part of a rise to any level up to $k$ (because the greedy algorithm would have taken care of that), but since these resets are executed on only part of the system, it could be that $e$ still wants to rise to some level higher than $k$. Meaning, if an element $e$ wanted to rise to some level $k'>k$ before the reset to level $k$, a reset to level  $k$ does not change this, since it only ``shuffles'' elements at level up to $k$, so in a sense this element has not been fully ``cleaned" yet. The same elements that wanted to rise to level $k'$ still want to rise to there after the reset to level $k$. Thus, even if all elements are active, which would mean that our system is entirely ``clean", it could be that many rises need to occur, which may blow up the approximation factor. Therefore, this binary definition of active or passive is insufficient, and we need to revise our definition of global dirt --- taking this issue into account.

\paragraph{Meeting Objective $1$ (Global Dirt).~}
It seems that our initial binary definition of active/passive must be ``level-sensitive" for it to work. Meaning, an element will be considered active \emph{up to a certain level}, and then passive \emph{from that level upwards}. Let us define the \emph{passive level} of an element $e$ to be this certain level, denoted by $\plev(e)$, and roughly speaking it will be one level higher than the reset level of the highest reset in which $e$ participated in since it was inserted. An element $e$ will be ``clean" below its passive level, and ``dirty" at or above it, and it is important to have  $\lev(e) \le \plev(e)$ (see the discussion below). Denote by $A_k$ (respectively, $P_k$) the set of all elements $e$ with $\lev(e) \le k$ and  $\plev(e)$ larger than (resp., no larger than) $k$. An element in $A_k$ (resp., $P_k$) will be called $k$-active (resp., $k$-passive). We define the system to be {\em $k$-dirty} if $|P_k| > 2\epsilon \cdot |A_k|$. Since $A_k \cup P_k$ is the set of all elements at level at most $k$, the system is $k$-dirty if roughly more than a $2\epsilon$-fraction of all elements at level up to $k$ have a passive level also up to $k$.

\paragraph{Meeting Objective $2$ (Global Algorithm).~}
%it seems that we must and we must consider the highest level in which each element is considered active. Meaning, following a reset to $k$, we would want each participating element to be considered active only at any level up to $k$, and passive for any higher level. Thus, we introduce the notion of \emph{passive level} for each element. Each element has two levels corresponding to it - its actual level, and its passive level. \shay{let's make it more precise: define $\plev(\cdot)$, define $A_k$ and $P_k$ as we do later, give the exact invariant that we have later (with the new definition of $N_k(\cdot)$), etc; the readers already saw the previous invariant - \Cref{global}, so introducing them to our new invariant makes sense} Intuitively, the passive level of an element will be one level higher than the reset level of the highest reset in which the element participated in since it was inserted (and equal to its actual level if it has not participated yet in a reset). 
We want to maintain the following invariant:

\begin{invariant}[see \Cref{inv} for more details]\label{our}
The following three conditions should hold:
	\begin{enumerate}[(1)]
 \item For any set $s$ in the current solution, we have $\frac{|\cov(s)|}{\cost(s)}\geq \beta^{\lev(s)}$.
		\item Define $N_k(s) = A_k \cap s$. For any set $s\in \sets$ and level $k$, we have $\frac{|N_k(s)|}{\cost(s)} < \beta^{k+1}$.	
  \item For any level $k$, we have $|P_k|\leq 2\epsilon \cdot|A_k|$.
	\end{enumerate}	
\end{invariant}

\noindent The first two conditions of the invariant correspond to the two in \Cref{global}, respectively. It may seem as though the first and second conditions imply local constraints, since they hold for each  set $s$. However, we make two crucial changes in the definitions: In the first condition, $\cov(s)$ is redefined to include also deleted elements that have not gone through a reset, and in the second condition, $N_k(s)$ is redefined to consider only \emph{$k$-active} elements. In a sense, these two conditions only consider ``clean" elements. Lastly,  we need to ensure that the vast majority of elements are indeed ``clean". Meaning, we want to prevent the accumulation of too many $k$-passive elements, for each $k$, otherwise the first two conditions would be meaningless, since a large fraction of  elements in the system would not be considered,
which may blow up the approximation factor.
To summarize, the first two conditions are local constraints that disregard all ``passive" elements (for each level), and the third condition ensures that such passive elements (for each level) are scarce, {\em and this is where the global relaxation for the first two comes into play}. Intuitively, the purpose of the third condition is to \emph{divert} dirt from the first two ``local" conditions (which are analogous to \Cref{global}) to the third, which is \emph{global} by design and thus crucial to achieve a fully global algorithm. 

To maintain this invariant we will execute a \emph{reset} up to level $k$ once $|P_k| > 2\epsilon \cdot|A_k|$, which amounts to running the static greedy algorithm on the subsystem of elements and sets that are up to level $k$, and removing deleted elements up to level $k$ from the system. We want to trigger resets only once there is a violation to the third condition. Thus, when an element is inserted, we will assign its passive level to be its actual level; 
%causing it to not be in any set $A_k$, and thus not in any $N_k(s)$, for any level $k$. In a sense, if the level and passive level of an element are equal, 
in a sense, it can be considered as ``completely passive", 
%since if $\lev(e)=\plev(e)=j$, then $e \in P_{j'}$ for any $j' \geq j$. Therefore, insertions cannot raise
since it cannot be in any set $A_k$ and $N_k(s)$ for any $k$ and $s$. When an element is deleted, we will not remove it from the system yet, and instead just mark it as \emph{dead}, and assign its passive level to be its actual level. Therefore, deletions cannot reduce the size of $\cov(s)$ for any $s$. Thus, insertions and deletions cannot cause violations to the first two conditions, and instead they add ``passiveness" to the system,
%by updating passive levels of elements, 
which will eventually trigger a violation to the third condition. 

We want a reset up to level $k$ to completely ``clean up" everything up to $k$. Meaning, we would want $P_{k'}=\emptyset$ for any $k' \leq k$ following the reset. Thus, by definition of $P_k$, each participating element must have a passive level higher than $k$ following the reset. Not only do we want a reset to level $k$ to clean up everything up to that level, we also require that it would not create more dirt (or ``passiveness") in any higher level, otherwise a reset can trigger another (higher) reset, which could blow up the update time. We want only insertions and deletions to create dirt. Thus, if for example a reset to level $k$ is being executed and as a result a participating set $s$ wants to be created at level $k+5$, because it contains about $\cost(s) \cdot \beta^{k+5}$ participating elements, we will not allow this, since this affects all levels between $k$ and $k+5$. Therefore, we will \emph{truncate} any reset to level $k$ at level $k+1$, meaning we will not allow participating sets to cover at any level higher than $k+1$. This way levels higher than $k$ are not affected by the reset, meaning there is no change to $P_{k'}$ and $A_{k'}$ for any $k' > k$, and notice that all participating elements in a reset up to level $k$ will end up at a level up to $k+1$ following the reset. We conclude that following a reset to level $k$ the first two conditions still hold by design of the greedy algorithm and the level assignment in it, the third condition holds since $|P_{k'}|=0$ for any $k' \leq k$, and the reset has not raised $|P_{k'}|$ or lowered $|A_{k'}|$ for any $k' > k$, so this reset cannot trigger a reset at any higher level. For more details regarding the algorithm description, see \Cref{sec:basAlg}.

Since each participating element must have a passive level higher than $k$ following a reset to level $k$, each participating element will be assigned a passive level of the maximum between $k+1$ and its previous passive level. In this way, notice that the passive level of an element will never be lower than its level, and that the passive level of any element throughout the entire update sequence is {\em monotonically non-decreasing}. This means that the number of different passive levels an element can go through during the update sequence is bounded by the number of levels in the system, and it turns out that this bound is what mandates the amortized update time to be $O(f \cdot \log n)$, neglecting  dependencies on $\eps$ and $C$. To show this, consider a token scheme which gives each inserted element $O(f)$ tokens for each passive level it could be at. The key observation is that even if there are multiple resets to the same  level $k$ throughout the update sequence, each element can only once be part of the collection $P_k$ that triggers the reset to level $k$ once $|P_k| > 2\epsilon \cdot|A_k|$, as afterwards its passive level would be at least $k+1$, and it will never decrease. Thus, we can give each element tokens to be responsible for only one reset for each level throughout the entire sequence. Since a reset to level $k$ occurs once $|P_k|$ is  (roughly) a $2\epsilon$-fraction of all elements up to level $k$, handing out $O(\frac{f}{\epsilon})$ tokens for each element per level would be enough to redistribute the tokens such that each participating element has $O(f)$ tokens, enough to enumerate the sets containing it and update the corresponding data structures regarding the new level. We have thus obtained an algorithm with $O(f \cdot \log n)$ amortized update time \footnote{The exact amortized update time is $O(\frac{f \cdot \log (Cn)}{\epsilon^2})$, since there are roughly $\frac{\log (Cn)}{\epsilon}$ levels. In \Cref{sec:deam} we explain how to get rid of the $C$ factor.}, which maintains \Cref{our} that is based on our definition of global dirt. The final objective is to show that we achieve the desired approximation factor.

\paragraph{Meeting  Objective $3$ (Approximation Factor).~}

We present a highly nontrivial proof for the approximation factor of $(1+\epsilon) \cdot \ln n$, which might be of independent interest. 
Our proof relies on
\Cref{our}, which uses a global notion of dirt, and as such it has to circumvent several technical hurdles that the previous proofs 
%is inherently different from the previous ones 
\cite{gupta2017online, solomon2023dynamic}
did not cope with.
See \Cref{pre-toy-approx} and \Cref{toy-approx} for the details. See Figure \ref{firstfi} for an illustration of the definitions and procedures given in the last few paragraphs.
%We believe this new proof of approximation is of independent interest.
%is an additional contribution of the paper in and of itself. For details and proof see \Cref{warmup}.

%\textcolor{red}{Amitai: I'm not sure if the following two figures should be here or in Section $3$. Perhaps first here and second (which demands describing what the algorithm does for insertion/deletion) in Section $3$.}
%\shay{the second figure indeed doesn't belong in the overview section, should be moved to Section 2. The first figure should be fixed, afterwards I'll look at it again; it might be a good addition to this section to have a figure on the effect of a reset at level 3, after which all elements are removed from $P_3$, some may belong in $A_3$, maybe to $P_4$, etc}

	\begin{figure} 
		\center{\includegraphics[scale=0.35]{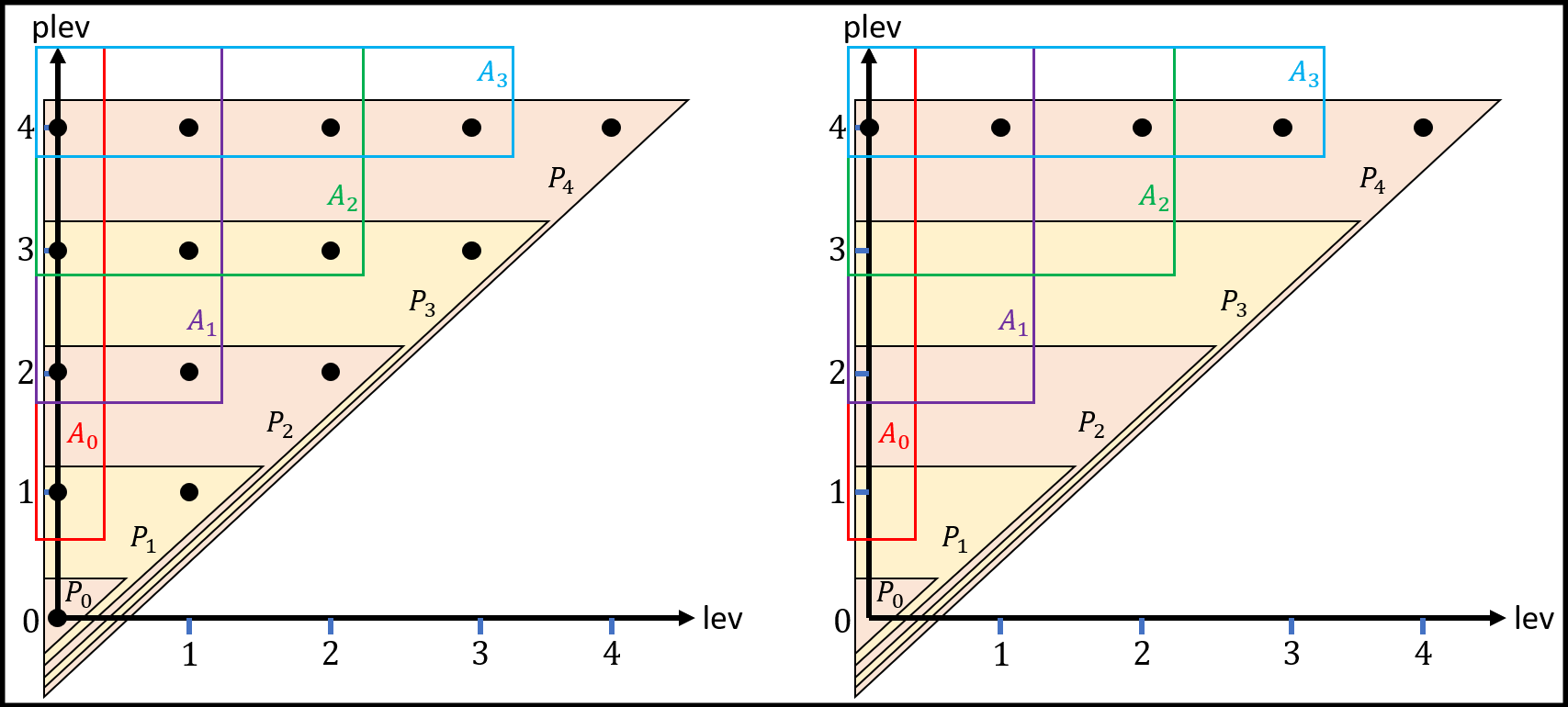}} 
		%\caption{\label{fig2} $3$-active (green) and $3$-passive (red) elements (left) and $4$-active (green) and $4$-passive (red) elements (right) at the same time $t$. We can see the following properties, from left to right: If an element is in $P_k$ it is also in $P_{k'}$ for any $k'>k$ \shay{$k' < k$}. If an element is in $A_k$ it is also in $A_{k'}$ for any $k'<k$ \shay{incorrect (maybe its level is smaller than $k'$}. But if an element is in $A_k$ we do not know if it is in $A_{k'}$ for any $k'>k$, and likewise if an element is in $P_k$ we do not know if it is in $P_{k'}$ for any $k'<k$ \shay{incorrect} (depends on the passive level). This does not necessarily mean though that if $P_k$ is a large enough fraction of $A_k$ then so is $P_{k'}$ a large enough fraction of $A_k'$ for any $k'>k$, since as can be seen in the figure, there could be several active elements (with respect to $k'$) at a level higher than $k$.}
        \caption{\label{fig2} On the left: Each black circle represents a single element with level and passive level up to $4$, and is assigned to different collections $A_k$ and $P_k$. Recall that for each element $e$, $\plev(e) \geq \lev(e)$. We can see the following properties: If an element is in $P_k$ it is also in $P_{k'}$ for any $k'>k$. If an element is in $A_k$ it is also in $A_{k'}$ for any $k'<k$ if its level is up to $k'$. For each $k$ we have that $A_k \cup P_k$ is the collection of all elements up to level $k$, and for each $k$ $A_k \cap P_k = \emptyset$. On the right: Following a reset up to level $3$, we have that $P_i = \emptyset$ for any $i \leq 3$, and the level of all elements that were at level up to $3$ are now at a level up to $4$, but the passive level of each such element is now at least $4$. The black circles here can represent several elements. Notice that ten elements participated in this reset, and these elements are represented by one of the five black circles. Any element that was in $P_4$ before the reset is still there after.}
        \label{firstfi}
	\end{figure}

%\color{black}

\subsubsection{A Lossless Deamortization} \label{sec:deam}
%\shay{should merge this section and the next into one, and give appr' title}

%\shay{the contents of this section are quite irrelevant, since \cite{bhattacharya2021dynamic} already dealt with most of these issues; we should put most of the focus in (1) the efficient switching of the pointers, (2) the removal of the $C$ factor. Also, maybe explain that 
%(a) it seems more difficult to dynamize and globalize the greedy vs primal-dual, so it makes sense that de-amoritizing it (with plev etc) would be more challenging. 
Recall that 
an efficient deamortization approach was given in the low-frequency regime,
where \cite{bhattacharya2021dynamic} 
deamortized the fully global amortized primal-dual algorithm of \cite{bhattacharya2019new}. However, the worst-case update time exceeds the amortized bound by a factor of $\frac{\log(C n)}{\eps}$.

The focus of this work is the high-frequency regime,
which, as mentioned already, appears to be more challenging when it comes to the dynamic setting. 
Having obtained a fully global algorithm with amortized update time that matches the previous best amortized bounds \cite{gupta2017online,solomon2023dynamic}, our next challenge is to deamortize it to achieve a good worst-case update time. We develop a {\em lossless} deamortization approach for the high-frequency regime, using which we achieve a worst-case update time that matches the best amortized bounds, and actually shaves off a $\eps^{-3}$ factor from \cite{solomon2023dynamic}.
We then apply our deamortization approach also in the low-frequency regime, first to shave off a $\log (Cn)$ factor, and then to remove the dependency on the aspect ratio $C$.
%The first question is whether we can develop an efficient deamortization approach for the high-frequency regime, and the second is whether it can be {\em lossless}.
%As mentioned already, the high-frequency regime seems to be more challenging when it comes to the dynamic setting. Nonetheless, when it comes to deamortization of a fully global algorithm 

Our deamortization approach is reminiscent of the one in the low-frequency regime \cite{bhattacharya2021dynamic}, since we need to cope with similar technical difficulties. Nonetheless, our approach has to deviate from the previous one in several key points.
%, partly as our goal is a lossless deamoritzation. 
We next discuss some of those technical difficulties, highlighting the new hurdles that we overcame on the way to achieving a {\em lossless} deamorization.
%and for the high-frequency regime is quite similar to
%For the sake of brevity, we defer the discussion ... in this discussion we will only focus

Consider a reset to level $k$.
If $k$ is large enough, then the reset cannot be carried out within a single update step, but rather needs to be {\em simulated} on \EMPH{the background} within a long enough time interval, where in each update step we can execute a small amount of computational steps. 
We shall denote by $\reset(k)$ a reset instance to level $k$;  roughly speaking, we would like to execute $O(f/\eps)$ computational steps of $\reset(k)$ for any possible level $k$ following each update step, so that the worst-case update time will be the number 
$O(\log_\beta (Cn))$
of levels times $O(f/\eps)$, namely $O(\frac{f \log(Cn)}{\eps^2})$.
Importantly, before a $\reset(k)$ instance can start, we first need to copy the contents of the current \EMPH{foreground} (output) solution  up to level $k$, as well as the underlying data structures,
to a chunk of {\em local memory} on the background, which is disjoint from the solution and data structures on the foreground as well as from those of any other reset instance that is running on the background. It is crucial that the contents of memory in any $\reset(k)$ instance will form an {\em independent} copy of the foreground solution and data structures up to level $k$.
Only after we have copied those contents, we turn to simulating the execution of the reset on the background. Finally, after termination of the reset in the background, we need to bring back the new solution and data structures up to level $k$ that we have in the background to the foreground (and overwrite it);  a central crux
(discussed below) is that this last part needs to be carried out within a single update step. 

If only one reset were to run in the background at any point in time, things would be easy.
However,
%The main challenges derive from the fact that
multiple resets at different levels
need to run together, and they all need to be simulated on the background at the same time; this issue, alas, may lead to various types of conflicts and inconsistencies.
%that need to be resolved.
Indeed, as mentioned, when a reset to level $k$ starts,  we first copy the contents of the foreground solution and data structures up to level $k$
to the background. However, these contents in the foreground may be partially or fully overwritten by resets that get terminated before the one that has just started, since any terminating reset is supposed to bring back (and overwrite) its background solution and data structures to the foreground.
This creates inconsistencies between the views of the foreground by different reset instances.

The key question is: {\em How should we resolve such inconsistencies?}
%???????????
%The first issue is that while a reset to $k$ is working in the background, it could be that in the foreground 
%we would like to initiate a reset to $k'$ as well. This would obviously lead to some conflicts and inconsistencies that need to be solved. For example, if the reset to $k'$ finishes after the reset to $k$, should it ``overwrite" it? And if before, should it be overwritten by the reset to $k$? 
It is natural to give a higher precedence to a reset instance at a higher level than to a lower level instance, as it essentially operates on a super set-system (a superset of sets and elements), however the time  needed to complete a reset grows with the reset level, hence a reset at a lower level might have started the reset well after the higher level reset,  so it should hold a more up-to-date foreground view. 

We will not discuss the answer to this question in detail here;
the formal answer appears in \Cref{threes}.
Instead, we wish to highlight a key difference between our approach and that of \cite{bhattacharya2021dynamic},
%Before unfolding the technical details, let us explain on a high level why we are able 
which allows us to shave the extra $\log_\beta (Cn)$ factor in the time bound. 

Both our algorithm and that of \cite{bhattacharya2021dynamic} assign elements and sets to levels at most $O(\log_\beta(Cn))$, and for each level $k$ there is a $\reset(k)$ instance that is running on a separate chunk of local memory on the background.
%, which is disjoint from the hierarchy data structure on the foreground. 
%In both algorithms, when the instance of $\reset(k)$ starts we need to ensure that the contents of its local memory is an independent copy of the foreground data structure up to level $k$. 
The executions of $\reset(k)$ in the two algorithms are quite different.
First, while \cite{bhattacharya2021dynamic}  simulates the water-filling primal-dual algorithm, we need to simulate the greedy algorithm.
There are also other differences, including the exact way that the algorithms cope with adversarial element updates that occur during the resets' simulations. 
The key difference, however, is in the manner in which we resolve the aforementioned inconsistencies, briefly described next.
%prevent inconsistencies between the solutions and data structures of different reset instances, and their independence from the foreground.

%, this issue is handled differently in the two algorithms. 
In both algorithms, when $\reset(k)$ terminates, it switches its local memory to the foreground and {\em aborts} all other lower level instances $\reset(i)$, for all $i<k$. 
To ensure that all the aborted instances $\reset(i)$ will have an independent local copy of the current data structures up to level $i$, the approach of \cite{bhattacharya2021dynamic}
is that, besides executing the water-filling procedures, the instance $\reset(k)$ will also be responsible for initializing an independent copy of the data structures up to level $i$ for instance $\reset(i)$, for all $i 
< k$, right after $\reset(i)$ is aborted by $\reset(k)$. This is the main reason that the algorithm of \cite{bhattacharya2021dynamic} has a quadratic dependency on $\log_\beta(C n)$, as $\reset(k)$ needs to prepare the initial memory contents for all other instances below it after it terminates, and it is crucial to carry this out within a single update step, again to avoid inconsistencies.

In our approach, to save the extra $\log_\beta (Cn)$ factor in the update time,  
%do the same as our algorithm in the high frequency regime. Basically, in our algorithm, 
the $\reset(k)$ instance will no longer be responsible for initializing the memory contents of $\reset(i)$, for all $i<k$, right after $\reset(i)$ is aborted by $\reset(k)$.
 %--- this process, which needs to be carried out in a single update step, already incurs a quadratic dependency on $\log_\beta(Cn)$ in the update time.
Instead, each instance $\reset(i)$ will initialize its own memory in the background by copying data structures in the foreground up to level $i$, and only when the initialization phase is done, the actual simulation procedure begins (of either the greedy algorithm in our case, or the water-filling algorithm as in \cite{bhattacharya2021dynamic}). Moreover, we would like to carry out the termination of any $\reset(k)$ instance in a single update step, meaning within $O(\log_\beta (Cn))$ time. Alas, the caveat of such a modification is that we are no longer able to determine in constant time the levels of sets and elements on the foreground (although we are able to do so in each $\reset(k)$ instance running on the background). 
%(We are still able to determine the level of any set and element in any $\reset(k)$ instance in constant time.)
Instead, we propose an \EMPH{authentication} process for determining the foreground level of any set or element in $O(\log_\beta(Cn))$ time.  We demonstrate that despite this caveat, we are able to achieve the desired update time of $O\brac{\frac{f\log(Cn)}{\epsilon^2}}$, see \Cref{updtime} for details. %\textcolor{orange}{Amitai: Maybe explain how? (That it is not the bottleneck)}. 

%Next, we describe the key ideas for removing the time dependency on the aspect ratio $C$.

%Solving the issues above we are able to deamortize the amortized algorithm as follows. Keep updating element insertions/deletions to the foreground system, until exists $k$ such that $|P_k| \geq \epsilon \cdot |A_k|$ (\textcolor{red}{Amitai: not sure if we should state here that we don't wait, as it seems maybe too much information for tech highlights}). Then, duplicate the relevant subsystem of sets and elements needed to compute this reset in the background. 
%By simulating a reset, the meaning is that we compute this reset in the background, $O(\frac{f}{\epsilon})$ computations per update step (which we prove is fast enough), and once done transfer the relevant subsystem back to the foreground and abort all other ongoing resets to any level $k'<k$. There are $O(\log_{\beta} (Cn))$ levels, thus in the worst case where all resets are working simultaneously, we obtain a worst case update time of $O_\epsilon(f\log(Cn))$ (later improved to $O_\epsilon(f\log n)$). We prove that at all times for each $k$ we have  $|P_k| < 2\epsilon \cdot |A_k|$, and so we obtain the approximation factor of $(1+\epsilon) \cdot \ln n$. 
%\shay{rephrase}

%\subsubsection{Removing Dependency on Aspect Ratio.}
\paragraph{Removing Dependency on Aspect Ratio.~}
The approach suggested above can only achieve a worst-case update time of $O_\epsilon(f\log(Cn))$, which could be prohibitively slow for a sufficiently large aspect ratio $C$. To remove the dependence on the aspect ratio, the first natural approach is to apply our algorithm only on the {\em lowest window} of $10\log_{\beta}n$ consecutive levels, which starts with the lowest non-empty level (i.e., which contains at least one element), and directly add all sets to our set cover solution on all higher levels (after the window). 
The intuition behind this approach is that sets belonging to levels higher than the lowest window have negligible costs compared to sets inside the lowest window, so adding those sets to our set cover solution does not change our approximation ratio significantly.

The main issue with this approach is that the lowest non-empty level, and thus the lowest window, changes dynamically. In particular, the
adversary could delete elements in the lowest window. Once this window becomes empty, the algorithm must switch its attention to a different window at higher levels. Alas, since the algorithm did not maintain 
any structure on higher levels, and in particular the underlying invariants could be completely violated outside the lowest window, restoring the necessary structures and invariants on high levels due to a sudden switch would be a heavy computational task, which cannot fit in our worst-case time constraints. If instead of considering the lowest window of $10 \log_\beta n$ {\em consecutive} levels, we consider a window that consists of the lowest $10 \log_\beta n$ {\em non-empty} levels, we will still run into the same problem --- the adversary could make all those non-empty levels empty (and thus to trigger a switch to a higher window) much earlier than the algorithm may hope to restore the invariants at higher levels, since it is possible that the lower non-empty levels occupy far less elements than the higher ones.
%\shay{I slightly edited; should also add a few sentences saying that, if instead of considering the lowest window of $10 \log_\beta n$ consecutive levels, we consider the lowest $10 \log_\beta n$ non-empty levels, we'll still run into the same problem - the adversary can make all those levels empty much earlier than we can hope to take care of the ``dirt'' or invariants at higher levels, since it's possible that the lower non-empty levels occupy far less elements than the higher ones}

To fix this issue, let us partition the entire level hierarchy  into a sequence of fixed non-overlapping windows, each consisting of $c\log_{\beta}n$ consecutive levels for a constant $c$; for concreteness, we assume in this discussion that $c = 10$. Instead of maintaining the validity of the data structures and invariants only for the lowest nonempty window, we will maintain them across all windows, by applying the previous algorithm (with update time that depends on the aspect ratio) for every window as a black-box, and the output solution would be the union of all set covers ranging over all the windows. 
For efficiency purposes, we would like to somehow {\em map} every element to a {\em single} window (instead of up to $f$ windows, one per each set to which the element belongs), so that for each element update, we will only need to apply as a black-box our previous dynamic set cover algorithm on that window, and do nothing for all other windows.
Obtaining such a mapping, where each element is mapped to only one window, is problematic in terms of the approximation
factor. We will not get into this issue, since even ignoring it,
%\shay{this is unclear/inaccurate; if we use a ``partition'' into disjoint windows (with non-overlapping levels), I think we can't map each element to a single window; in any case, I think it's useful to describe the mapping of an element to a window - naively, we'll put an element in up to $2f$ windows (each set $s$ that contains $e$ may belong to a different window, so should explain that we only consider the cheapest set; also, the cheapest set may belong to two non-overlapping windows in principle, due to the number of elements that it currently covers)}
%\tianyi{This is not the final algorithm. There could be many ways to map an element to a single window, like what we will do with two partitions. But my main point is that, even if we map each element in a single window in some ways, there will always be troubles. That is why we need two partitions}
% \shay{unclear: should explain the mapping between elements and windows; I think that an element should belong to two consecutive windows, as this is crucial for the approx factor. Also, maybe stress that only the couple of windows that contain the updated element should simulate the black-box algorithm, and that we don't do anything for the remaining windows, otherwise the update time will be the same as in the previous section}
this approach may only give a $(2(1+\epsilon)\ln n)$-approximation, rather than a $((1+\epsilon)\ln n)$-approximation. Indeed, consider the case where elements in the lowest window are all lying towards the higher end of the window; more specifically, assume that in the lowest window of levels $[0, 10\log_{\beta}n]$, all elements are on levels $[9\log_\beta n, 10\log_{\beta}n]$. In this case, the costs of sets on levels $>10\log_{\beta}n$ are not negligible compared to those on levels $[9\log_{\beta}n, 10\log_{\beta}n]$. Consequently, although all sets in the third lowest window and all higher ones have negligible costs with respect to the lowest window,
we can only argue that the approximation ratios in the two lowest windows are both $\brac{(1+\epsilon)\ln n}\cdot \opt$, 
which results in a $(2(1+\epsilon)\ln n)$-approximation.

%\shay{maybe should also consider the third window, and say that the costs of all sets there are negligible wrt level $\le 10\log_\beta n$ (they contribute together less than an $\eps$-fraction of the first window's cost), and then we just have two windows to consider, the first and the second, hence the factor 2 overhead (I think it's more intuitive in this way)}

% \shay{if there is just one (or $O(1)$) nonempty level in the lowest window, then the approx for the lowest window set system is actually $O(1)$ and not $(1+\eps)\ln n$ (I can prove this), which makes our message a bit problematic}

To resolve this issue, we will use two overlapping sequences of windows instead of one; that is, the first sequence of windows is roughly $[0, 20\log_{\beta}n]\cup [20\log_{\beta}n + 1, 40\log_{\beta}n]\cup\ldots$, and the second sequence partitions the levels as roughly $[0, 10\log_{\beta}n], [10\log_{\beta}n+1, 30\log_{\beta}n]\cup [30\log_{\beta}n+1, 50\log_{\beta}n]\cup \ldots$. Then, for each of the two partitions, we apply our aspect-ratio-dependent algorithm as a black-box within each of the windows independently, so in the end we are maintaining two different candidate set cover solutions, where the one of smaller cost would be presented to the adversary. We argue that at any moment, at least one of the candidate set cover solutions provides a $((1+\epsilon)\ln n)$-approximation. To see this, consider the lowest non-empty window in each of the two sequences; we can show that in one of those windows, the lowest non-empty level lies in the lower half of that window, and the set cover solution corresponding to that window provides the required approximation, since all sets in the second lowest window and all higher ones in that sequence have negligible costs with respect to the lowest window, due to the half-window ``buffer''  that we have between the lowest nonempty level and the second window.

Our approach, which employs a fixed partition into windows, has two advantages over alternative possible suggestions that use dynamically changing windows. First, it is more challenging to maintain the data structures and the required invariants when using dynamically changing windows (and it is not even clear whether such alternative suggestions could work).
Second, and perhaps more importantly,
our approach enables us to apply the aspect-ratio-dependent algorithm as a {\em black-box} in each window, whereas it is unclear how to apply the algorithm as a black-box when using dynamically changing windows. See \Cref{remove} for details.

\paragraph{A Unified Approach.}
In \Cref{primdu} we demonstrate that our deamortization approach extends seamlessly to the low-frequency regime. This also applies to the removal of the aspect ratio dependency from the time bound, which, as mentioned above, is achieved via a black-box reduction. Our approach thus {\em unifies the landscape} of dynamic set cover algorithms with worst-case time bounds.

\section{Our Algorithm I: Aspect Ratio Dependence in Update Time} \label{threes}
We first prove a weaker version of our result, where the worst-case update time depends on the aspect ratio. 

\begin{theorem}\label{warmup}
	For any set system $(\univ, \sets)$, along with a cost function $\cost: \sets\rightarrow [1/C, 1]$, that undergoes a sequence of element insertions and deletions, where the frequency is always bounded by $f\geq \ln n$, and for any $\epsilon \in (0, \frac{1}{4})$, there is a dynamic algorithm that maintains a $((1+\epsilon)\ln n)$-approximate minimum set cover in $O\brac{\frac{f\log(Cn)}{\epsilon^2}}$ deterministic worst-case update time.
\end{theorem}

\subsection{Preliminaries, Invariants and Approximation Factor Analysis} \label{warmupsec}

Without loss of generality, assume that $\max_{s\in \sets}\cost(s) = 1$. Let $\beta = 1+\epsilon$. All sets $s \in \mathcal{S}$ will be assigned a level value $\lev(s)\in [-1, L]$ where $L = \ceil{\log_\beta(Cn)} + \ceil{10\log_\beta 1/\epsilon}$. Throughout the algorithm, we will maintain a valid set cover $\sets_\alg\subseteq \sets$ for all elements. We will assign each element $e\in \univ$ to one of the sets $s\in \sets_\alg$, which we will denote by $\asn(e)$, and conversely, for each set $s\in \sets$, define its {\em covering set} $\cov(s)$ to be the collection of elements in $s$ that are assigned to $s$, namely $\cov(s) = \{e\mid \asn(e) = s\}$. The level of an element $e$ is defined as the level of the set it is assigned to, namely $\lev(e) = \lev(\asn(e))$, and we make sure that $\lev(e) = \max\{lev(s) \vert s \ni e\}$, meaning $e$ is assigned to the set with the highest level containing $e$. We define the level of each set $s \notin \sets_\alg$ to be $-1$, whereas the level of each set $s \in \sets_\alg$ will lie in  $[0,L]$, so in particular we will have $\lev(e)\in [0, L]$,
for each element $e \in \mathcal{U}$. Let $S_i = \{s\mid \lev(s) = i \}, \forall i\in [-1, L]$, and $E_i = \{e\in \univ\mid \lev(e) = i\}, \forall i\in [0, L]$.

Besides the level value $\lev(e)$ for elements $e$, we will also maintain a value of \emph{passive level} $\plev(e)$ such that $\lev(e)\leq \plev(e) \leq L$, which plays a major role in our algorithm. In contrast to the level $\lev(e)$ of an element $e$, which may decrease (as well as increase) by the algorithm, its passive level $\plev(e)$ will be monotonically non-decreasing throughout its  lifespan.

An element is said to be \emph{dead} if it was deleted by the adversary, hence 
it is supposed to be deleted from $\univ$ --- but it currently resides in $\univ$ as our algorithm has not removed it yet. An element is said to be \emph{alive} if it is not dead. To avoid confusion, we will use the notation $\univ^+\supseteq \univ$ to denote the set of all dead and alive elements (i.e., the elements in the view of the algorithm), while $\univ$ is the set of alive elements (i.e., the elements in the eye of the adversary). We next introduce the following key definitions.

\begin{definition} \label{actpass}
For each level $k$, 
an element 
$e\in \univ^+$
is called \EMPH{$k$-active} (respectively, \EMPH{$k$-passive}) if 
$\lev(e) \leq k< \plev(e)$
(resp., $\plev(e)\leq k$)
and let $A_k = \{e\in \univ^+\mid \lev(e) \leq k< \plev(e)\}$ and $P_k = \{e\in \univ^+\mid \plev(e)\leq k\}$
be the sets of all $k$-active and $k$-passive elements, respectively.
Notice that $A_k \cup P_k$ is the collection of all elements at level $\leq k$, and $A_k \cap P_k = \emptyset$. Moreover, if $A_k \cap P_{j} \neq \emptyset$ for two levels $k \ne j$, then $k < j$. For each set $s\in \sets$, define $N_k(s) = A_k\cap s$.
\end{definition}
While previous works on primal-dual dynamic set cover algorithms \cite{bhattacharya2019new,bhattacharya2021dynamic,bukov2023nearly}
also use the terminology of {\em active} and {\em passive} elements, it has a completely different meaning there.
Moreover, importantly, while in previous work an element may be either active or passive, here we refine this binary distinction by introducing a level parameter; in particular, an element might be $k$-active and yet $j$-passive (for indices $k < j$). 

This refinement is crucial for our algorithm to efficiently maintain the following invariant (Invariant \ref{inv}), which is key to  bounding the approximation factor (see \Cref{pre-toy-approx} and \Cref{toy-approx}).
The first part of the invariant essentially aims at achieving a global analog of the local \Cref{global}(2).
It actually
provides a strict upper bound on
$|N_k(s)|$ for any set $s$ and each level $k$, which might seem too good to be true.
The reason such a strict, local upper bound can be efficiently maintained by the algorithm is that $N_k(s)$ is restricted to the $k$-active elements in set $s$, or in other words, \emph{all $k$-passive elements in $s$ are simply ignored} --- which is where the global relaxation comes into play.
Indeed, to prevent the accumulation of too many $k$-passive elements  --- which is crucial for bounding the approximation ratio ---
the third part of the invariant
restricts the ratio between the $k$-passive elements and the $k$-active elements to be at most $2\eps$ at all times.
Thus, although the upper bound on $|N_k(s)|$ holds  ``locally'' (i.e., for any set $s$ and each level $k$), it only holds ``globally'' (i.e., for an {\em average} set and each level $k$) if we take into account the ignored $k$-passive elements. 
%nature of this invariant arises from the fact we ignore $k$-passive elements. 
%this is where the global relaxation comes into play [[S: again]]. 
In order for the algorithm to maintain the third part of the invariant, a natural thing to do would be to turn $k$-passive elements into active across all levels (or \EMPH{fully-active}). Alas, if we turned a $k$-passive element into fully-active, that could violate the first part of the invariant across multiple levels. 
%as the algorithm might have to turn possibly many $k$-passive elements into fully-active in a short time window, the first part of the invariant would be very significantly violated. 
To circumvent this hurdle, our algorithm will turn elements into \EMPH{partially-active}, i.e., active in a precise interval of levels;
specifically, element $e$ will become active in the interval $[\lev(e),\plev(e) - 1]$ (as in \Cref{actpass}), and to perform efficiently --- the algorithm will have to carefully choose the right values for $\lev(e)$ and $\plev(e)$;
the exact details are given in the algorithm's description (Section \ref{sec:basAlg}).
Finally, we note that the second part of Invariant \ref{inv} coincides with \Cref{global}(1). Here too, the invariant seems like a local bound since it holds for any $s$, but it uses again the global relaxation provided by the third part of the invariant, since it considers dead elements as well.

\begin{invariant} {\ } \label{inv}
	%There are several properties that will be maintained by the algorithm.
	\begin{enumerate}[(1),leftmargin=*]
		\item For any set $s\in \sets$ and for each $k \in [0, L]$, we have $\frac{|N_k(s)|}{\cost(s)} < \beta^{k+1}$.
		
		%\shay{Later on we only need $\beta^{k+2}$}
		
		%\tianyi{Yes, but $\beta^{k+1}$ looks cleaner.}
		
		\item For any set $s\in \sets_\alg$, we have $\frac{|\cov(s)|}{\cost(s)}\geq \beta^{\lev(s)}$;
  we note that $\cov(s)$ may include dead elements, i.e., elements in $\univ^+ \setminus \univ$.
  In particular, $\lev(s)\leq \ceil{\log_\beta(Cn)}$.
%that are no longer in the system.
		Moreover, for each $s\notin \sets_\alg$, $\lev(s) = -1$.
		
		%\shay{better use $\cov^+(s)$}
		
		%\tianyi{The adversary does not see $\cov(s)$, so I am only using one notation $\cov(s)$.}
		
		\item For each $k \in [0, L]$, we have $|P_k|\leq 2\epsilon \cdot|A_k|$. We note that our algorithm does not maintain the values of $|P_k|, |A_k|$.

	\end{enumerate}
\end{invariant}

The following lemma shows that the approximation factor is in check (the term $1+O(\eps)$ can be reduced to $1+\eps$ by scaling).
The proof of Lemma \ref{toy-approx} is inherently different from and more challenging than the  approximation factor proofs of the previous dynamic greedy set cover algorithms \cite{gupta2017online,solomon2023dynamic}; while the proofs of
\cite{gupta2017online,solomon2023dynamic}
are obtained by introducing natural tweaks over the standard analysis of the static greedy algorithm, our approximation factor proof has to deviate significantly from the standard paradigm,
since \Cref{inv}
%--- aiming for efficiency --- 
is inherently weaker than those of \cite{gupta2017online,solomon2023dynamic} --- particularly as \Cref{inv}(1) ignores all $k$-passive elements. 

\bigskip

\color{black}

%An element is said to be \emph{dead} if it was deleted by the adversary, hence it is supposed to be deleted from $\univ$ --- but it currently resides in $\univ$ as our algorithm has not removed it yet. An element is said to be \emph{alive} if it is not dead. To avoid confusion, we will use the notation $\univ^+\supseteq \univ$ to denote the set of all dead and alive elements (i.e., the elements in the view of the algorithm), while $\univ$ is the set of alive elements (i.e., the elements in the eye of the adversary). We next introduce the following key definitions.

\begin{lemma}\label{pre-toy-approx}
	Let $\sets^*$ be an optimal set cover for $\univ$ (i.e., of all alive elements), and let $n'$ be an upper bound to the size of each set throughout the update sequence. If \Cref{inv} is satisfied, then it holds that $\cost(\sets_\alg)\leq (1+O(\eps))\cdot \ln n' \cdot \cost(\sets^*)$.
\end{lemma}
\begin{proof}
	By \Cref{inv}(3), we have for all $k \in [0,L-1]$:
	\begin{equation} \label{eq:inv313}
	\brac{\beta^{-k} - \beta^{-k-1}}\cdot |P_k| ~\leq~ 2\epsilon\cdot \brac{\beta^{-k} - \beta^{-k-1}}\cdot|A_k|.
	\end{equation}
	Taking a summation over all  $k\in \left[0, L-1\right]$, the left-hand side of \Cref{eq:inv313} becomes:
	$$\begin{aligned}
		\sum_{k = 0}^{L-1}(\beta^{-k} - \beta^{-k-1})\cdot |P_k| 
		&= \sum_{e\in \univ^+}\sum_{k = 0}^{L-1}(\beta^{-k} - \beta^{-k-1}) \cdot \ind[e\in P_k]\\
		&=\sum_{e\in \univ^+}\sum_{k = \plev(e)}^{L-1}(\beta^{-k} - \beta^{-k-1})\\
		& = \sum_{e\in \univ^+}\brac{\beta^{-\plev(e)}-\beta^{-L}}
	\end{aligned}$$
	and the right-hand side of \Cref{eq:inv313} becomes:
	$$\begin{aligned}
		2\epsilon\sum_{k = 0}^{L-1}(\beta^{-k} - \beta^{-k-1})\cdot |A_k| 
		&= 2\epsilon\sum_{e\in \univ^+}\sum_{k = 0}^{L-1}(\beta^{-k} - \beta^{-k-1}) \cdot \ind[e\in A_k]\\
		&= 2\epsilon\sum_{e\in \univ^+}\sum_{k = \lev(e)}^{\plev(e)-1}(\beta^{-k} - \beta^{-k-1})\\
		& = 2\epsilon\sum_{e\in \univ^+}\brac{\beta^{-\lev(e)}-\beta^{-\plev(e)}},
	\end{aligned}$$
	which yields:	
	$$\sum_{e\in \univ^+}\brac{\beta^{-\plev(e)}-\beta^{-L}}\leq 2\epsilon\cdot \sum_{e\in \univ^+}\brac{\beta^{-\lev(e)} - \beta^{-\plev(e)}}$$
	or equivalently, by adding $\sum_{e\in \univ^+}\brac{\beta^{-\lev(e)} - \beta^{-\plev(e)}}$ on the both sides, 
	\begin{equation} \label{eq:basicub}
	\sum_{e\in \univ^+}\brac{\beta^{-\lev(e)}-\beta^{-L}} ~\leq~ (1+2\epsilon)\cdot \sum_{e\in \univ^+}\brac{\beta^{-\lev(e)} - \beta^{-\plev(e)}}.
	\end{equation}
	We emphasize the point that $\univ^+$ also includes dead elements.

	%\shay{should clarify that $\univ$ here also counts dead elements}
	
	Next, let us lower bound $\cost(\sets^*)$ using the term $\sum_{e\in \univ^+}\brac{\beta^{-\lev(e)} - \beta^{-\plev(e)}}$. For any $s\in \sets^*$, consider the following three cases for any index $k\in [L]$: 
	\begin{itemize}
	\item $k < \log_\beta(1/\cost(s))-1$.
		
		By \Cref{inv}(1), we have: $|N_k(s)| < \beta^{k+1}\cdot \cost(s) < 1$, so $|N_k(s)| = 0$.

		\item $\log_\beta (1/\cost(s))-1\leq k \leq \log_\beta (n'/\cost(s))$.
		
		By \Cref{inv}(1), we have: %\shay{should be $\beta^{-k-2}|N_k(s)|$, also in the third case}
		
		$$\frac{1}{\epsilon}\brac{\beta^{-k} - \beta^{-k-1}}|N_k(s)| = \beta^{-k-1}|N_k(s)| < \cost(s).$$

		\item $k > \ceil{\log_\beta (n' / \cost(s))} = k_0$.
		
		In this case, we use the trivial bound: $|N_k(s)| \leq n' \leq \beta^{k_0}\cdot \cost(s)$, and so we have:
		$$\frac{1}{\epsilon}\brac{\beta^{-k} - \beta^{-k-1}}|N_k(s)| = \beta^{-k-1}|N_k(s)| \le \beta^{k_0 - k-1}\cdot \cost(s).$$
		
	%\shay{should be $\beta^{k_0 - k-2}$}
	\end{itemize}
%	Taking a summation over all $k\in [L]$, we can prove:
Observe that:
	\begin{equation} \label{lefthand}
		\begin{aligned}
		\frac{1}{\epsilon} \sum_{k = 0}^{L-1}(\beta^{-k} - \beta^{-k-1})\cdot |N_k(s)| 
		&= \frac{1}{\epsilon}\sum_{e\in s}\sum_{k = 0}^{L-1}(\beta^{-k} - \beta^{-k-1}) \cdot \ind[e\in N_k(s)]\\
		&= \frac{1}{\epsilon}\sum_{e\in s}\sum_{k = \lev(e)}^{\plev(e)-1}(\beta^{-k} - \beta^{-k-1})\\
		& = \frac{1}{\epsilon}\sum_{e\in s}\brac{\beta^{-\lev(e)}-\beta^{-\plev(e)}}.
	\end{aligned}
	\end{equation}
	By the above case analysis, we have:
	\begin{equation}
	\begin{aligned} \label{righthand}
	\frac{1}{\epsilon} \sum_{k = 0}^{L-1}(\beta^{-k} - \beta^{-k-1})\cdot |N_k(s)|
	&= \frac{1}{\epsilon} \sum_{0 \le k < \log_\beta(1/\cost(s))-1}(\beta^{-k} - \beta^{-k-1})\cdot |N_k(s)| 
	\\&~~~+
	\frac{1}{\epsilon} \sum_{ \log_\beta(1/\cost(s))-1 \le k \le  \log_\beta (n'/\cost(s))}(\beta^{-k} - \beta^{-k-1})\cdot |N_k(s)| 
	\\&~~~+
		\frac{1}{\epsilon} \sum_{\log_\beta (n'/\cost(s)) < k \le L-1}(\beta^{-k} - \beta^{-k-1})\cdot |N_k(s)| 
	\\&<   \sum_{0 \le k < \log_\beta(1/\cost(s))-1} 0 
		\\&~~~+ 
  \sum_{ \log_\beta(1/\cost(s))-1 \le k \le  \log_\beta (n'/\cost(s))} \cost(s) 
		\\&~~~+  \sum_{\log_\beta (n'/\cost(s)) < k \le L-1}\beta^{k_0 - k - 1} \cost(s) 		
	\\&< \brac{0 + (\log_\beta(n') + 2) + 1/\epsilon}\cdot \cost(s).
	\end{aligned}
	\end{equation}
Combining \Cref{lefthand} with \Cref{righthand} yields
	$$\frac{1}{\epsilon}\sum_{e\in s}\brac{\beta^{-\lev(e)} - \beta^{-\plev(e)}}\leq \brac{\log_\beta(n') + 2 + 1/\epsilon}\cdot \cost(s).$$
	Therefore, as $\ln(1+\epsilon) = \epsilon + O(\epsilon^2)$, under the assumption that $\eps =  \Omega(1 / \log n')$ we have:
	$$\sum_{e\in s}\brac{\beta^{-\lev(e)} - \beta^{-\plev(e)}}\leq
	(1+O(\epsilon))\ln n'\cdot \cost(s).$$
	Since $\sets^*$ is a valid set cover for all elements in $\univ$ (all the alive elements) and as for each dead element $e$ (in $\univ^+ \setminus \univ$) we have $\brac{\beta^{-\lev(e)} - \beta^{-\plev(e)}} = 0$, it follows that:
	
	%\shay{since $\univ$ might contain dead elements, should add another inequality for all alive elements but we don't have a notation for this} 
	%\shay{must assume that $\eps =  \Omega(1 / \log n)$}
	
	\begin{equation} \label{eq:conclude} \sum_{e\in \univ^+}\brac{\beta^{-\lev(e)} - \beta^{-\plev(e)}} ~\leq~
	\sum_{s \in S^*} \sum_{e \in s} (\beta^{-\lev(e)} - \beta^{-\plev(e)})
	~\le~ (1+O(\epsilon))\ln n'\cdot \cost(\sets^*).
	\end{equation}
	
%	Finally, let us upper bound $\cost(\sets_\alg)$ with the term $\sum_{e\in \univ}\brac{\beta^{-\lev(e)} - \beta^{-L}}$. 
\noindent We conclude that
%\todo{I changed $\univ$ by $\univ^+$ in the following (and made other edits)}
%, we have the following inequality which finishes the proof:
$$\begin{aligned}
	\cost(\sets_\alg) &= \sum_{s\in \sets_\alg}\cost(s)\leq \sum_{s\in \sets_\alg}\beta^{-\lev(s)}\cdot |\cov(s)| ~=~ \beta\cdot \sum_{s\in \sets_\alg}\sum_{e\in \cov(s)\cap\univ^+}\beta^{-\lev(e)} 
	\\ &\le (1+O(\epsilon))\cdot \sum_{e\in\univ^+}\brac{\beta^{-\lev(e)} - \beta^{-L}}
	~\le~ (1+O(\epsilon))\cdot \sum_{e\in \univ^+}\brac{\beta^{-\lev(e)} - \beta^{-\plev(e)}}
	\\ &\le (1+O(\epsilon))\ln n'\cdot \cost(\sets^*),
\end{aligned}$$
where the first inequality holds by \Cref{inv}(2), the second holds as $\lev(e)\leq L/2$ and hence $\beta^{-\lev(e)} - \beta^{-L} \ge \beta^{-\lev(e)}(1- \beta^{-\ceil{10\log_\beta 1/\epsilon}}) \ge \beta^{-\lev(e)}(1-\eps)$, and the last two follow from \Cref{eq:basicub} and \Cref{eq:conclude}, respectively.
% \shay{it should be $\beta^{\lev(s)}$ rather than $\beta^{1-\lev(s)}$}
\end{proof}

\begin{corollary} \label{toy-approx}

Since $n' \leq n$, we get that if \Cref{inv} is satisfied, then it holds that $\cost(\sets_\alg)\leq (1+O(\eps))\cdot \ln n \cdot \cost(\sets^*)$. 
    
\end{corollary}

\subsection{Algorithm Description} \label{sec:basAlg}

We will skip the details for the fully global algorithm that maintains \Cref{inv}, with \emph{amortized} update time of $O(\frac{f \cdot \log (Cn)}{\epsilon^2})$; for the general outline of this algorithm, see \Cref{fullyg}. Instead, we will dive straight into our ultimate goal of providing an algorithm that maintains \Cref{inv}, with a \emph{worst-case} update time of $O(\frac{f \cdot \log (Cn)}{\epsilon^2})$ --- this is the result that underlies \Cref{warmup}.
The main procedure of the algorithm is a {\em reset} operation, denoted by $\reset(k)$ when initiated for a level $k$. Simply put, performing a reset to level $k$ amounts to running the static greedy algorithm on the subuniverse of elements and sets at level up to $k$.
Our algorithm distinguishes between procedures and data structures that are executed and maintained in the \emph{foreground} and those in the \emph{background}. The foreground procedures can be executed from start to finish between one adversarial update step to the next --- and as such are very basic procedures, and the foreground data structures support the foreground procedures and are used for explicitly maintaining the {\em output} solution at every update step.
In contrast, the background procedures can take a long time to run; the algorithm {\em simulates} their execution over a possibly long time interval in the background, and only upon termination of the execution, the main algorithm may ``copy'' the background data structures and their induced output
into the foreground data structures and their induced output.
In particular, the aforementioned reset operation will be running in the background, while adversarial insertions and deletions will be handled in a rather straightforward manner in the foreground. 
We note that for an amortized algorithm, there is no need for any background procedures, since everything can be executed on the foreground 
when needed. Meaning, once a reset needs to be executed, we just execute it in a single update step in the foreground. In this section we give a high level description of the algorithm, refer to \Cref{updtime} for more lower level details regarding the exact data structures maintained etc. 

\subsubsection{Foreground}

The set cover solution $\sets_\alg$, which serves as the interface to the adversary (i.e., the output), will be maintained in the foreground.
Element deletions and insertions will be handled in the foreground as follows.

\begin{itemize}[leftmargin=*]
	\item \textbf{Deletions in the Foreground.} When an element $e\in \univ$ is deleted by the adversary, 
	we set $\plev(e) \leftarrow \lev(e)$, and mark element $e$ as dead.
%	We then increment (by one) all the counters $|P_k|, \forall k\in [\lev(e), \plev(e)-1]$ by one, and decrement all the counters $|A_k|, \forall k\in [\lev(e), \plev(e)-1]$ by one \todo{it seems we don't need to maintain these counters, they are only for the analysis; still, it's useful for the readres perhaps to see how the counters are updated, but it's too tedious. I'd omit it and leave a comment}.  
	Finally, for each $k\geq \lev(e)$, we feed the deletion of $e$ to instance $\reset(k)$ (if operating).  Note that there is no need to feed the deletion of $e$ to instances of $\reset(k)$ with $k \leq \lev(e) - 1$, since $\reset(k)$ is not affected by (and does not affect) levels larger than $k+1$. See \Cref{del} for the pseudo-code.
\begin{algorithm} 
	\caption{$\textsf{Foreground-Delete}(e)$}\label{del}
	%\For{$k$ from $\lev(e)$ to ($\plev(e) - 1$)}{
	%	$|P_k| ++$\;
	%	$|A_k| --$\;
	%}
	$\plev(e) \leftarrow \lev(e)$\;
	mark $e$ as dead\;
	\For{$k$ from $\lev(e)$ to $L$}{
		\If{$\reset(k)$ is operating}{
			feed this deletion of $e$ to background system working on $\reset(k)$\;
		}
	}
\end{algorithm}

	\item \textbf{Insertions in the Foreground.} When an element $e$ is inserted by the adversary, go over all sets $s\ni e$ and check if there is one in our set cover solution $\sets_\alg$. If so, let $s\ni e$ be such a set at the highest level, and assign $\lev(e) = \plev(e) = \lev(s), \asn(e) = s$. 
	%and increment all the counters $|P_k|, k\geq \plev(s)$. 
	If $e$ is not covered by any set currently in $\sets_\alg$, add an arbitrary $s\in \sets\setminus \sets_\alg$ to $\sets_\alg$, (which was at level $-1$, as guaranteed by \Cref{inv}(2)), and assign $\plev(e) = \lev(e) = \lev(s) = 0, \asn(e) = s$.
	%, and increment all the counters $|P_k|, k\in [0, L]$. 
	(Note that after adding $s$ to $\sets_\alg$, \Cref{inv}(2) is still satisfied, as $\frac{|\cov(s)|}{\cost(s)} \geq 1= \beta^{\lev(s)}$.) Finally, for each $k\geq \plev(e)$, feed the insertion $e$ to instance $\reset(k)$ if operating.
	See \Cref{ins} for the pseudo-code, and \Cref{fig2} for an illustration of a deletion and insertion.

\begin{algorithm}
	\caption{$\textsf{Foreground-Insert }(e)$} \label{ins} 
	%where $S_e$ is the collection of sets containing $e$)} 
	let $s \ni e$ be highest level set containing $e$\;
	\If{$\lev(s) > -1$}{
		$\lev(e),\plev(e) \leftarrow \lev(s)$\;
		$\asn(e) \leftarrow s$\;
	}\Else{
		$\lev(e),\plev(e),\lev(s) \leftarrow 0$\;
		$\asn(e) \leftarrow s$\;
            $\sets_\alg \leftarrow \sets_\alg \cup \{ s \}$\; 
	}
	%\For{$k$ from $\plev(e)$ to $L$}{
	%	$|P_k| ++$\;
	%}
	\For{$k$ from $\plev(e)$ to $L$}{
		\If{$\reset(k)$ is operating}{
		feed this insertion of $e$ to background system working on $\reset(k)$\;
		}
	}
\end{algorithm}

\begin{figure}
		\center{\includegraphics[scale=0.35]{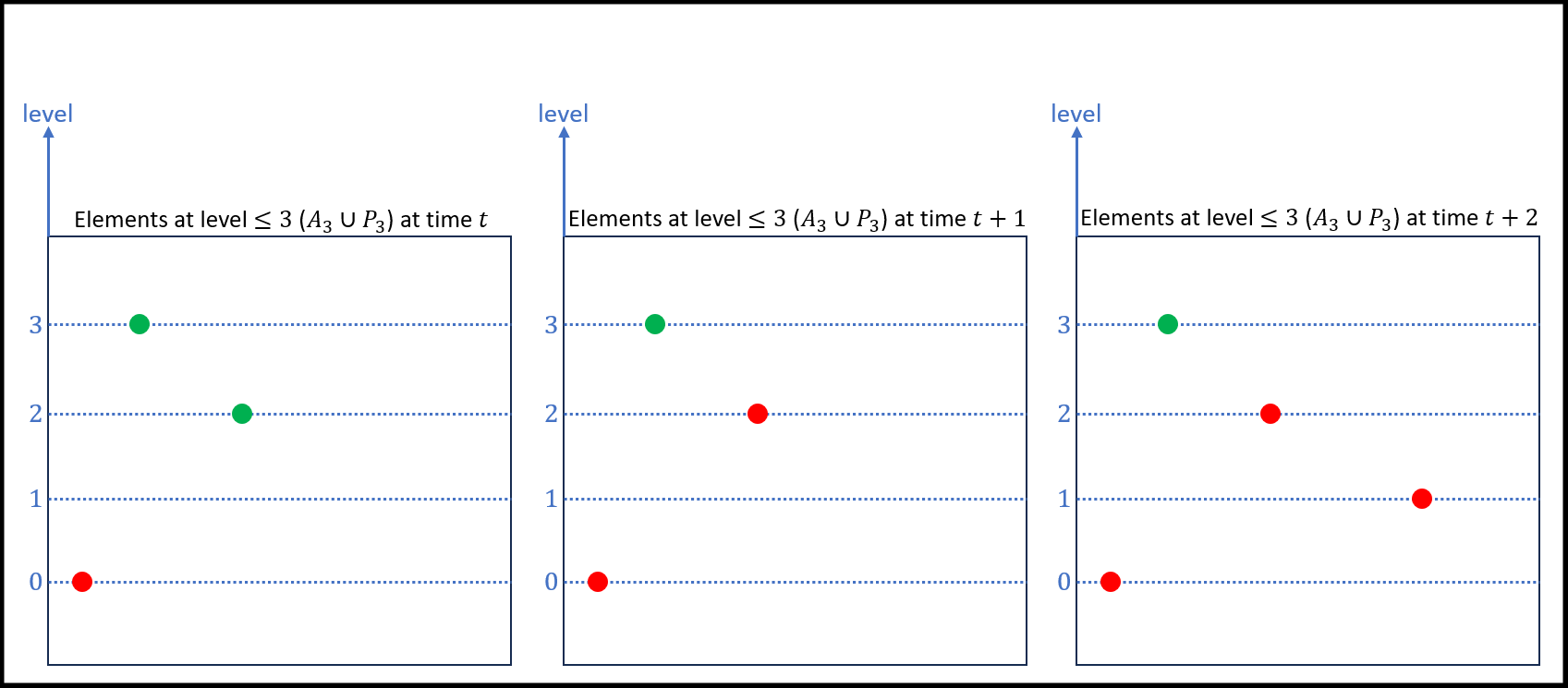}}
		\caption{\label{fig3} $3$-active (green) and $3$-passive (red) elements at time $t$ (left), $t+1$ (middle) and $t+2$ (right). Between $t$ and $t+1$ the element at level $2$ is deleted, thus becomes $3$-passive (because its passive level becomes $2$). Between $t+1$ and $t+2$, an element is inserted to level 1. It is $3$-passive because its passive level is $1$.}
        \label{fig2}
	\end{figure}

	\item \textbf{Termination of $\reset(\cdot)$ Instances.} Upon any element update (deletion or insertion), go over all levels $0\leq k\leq O(\log_\beta(Cn))$ and check if any instance $\reset(k)$ has just terminated right after the update. If so, take the largest such index $k$, and {\em switch its memory} to the foreground; we will describe how a memory switch is done later on in \Cref{termsec}. After that, abort all instances of $\reset(i)$, for $0\leq i<k$.
	
%	\item \textbf{Scheduling resets in the foreground.} For each level $k$, if the inequality $|P_k| < \epsilon |A_k|$ is violated, and there is currently no instance $\reset(k)$, then start an instance $\reset(k)$ in the background. Then, after each update from the adversary, go over all $k\in [0, L]$, and execute $O\brac{f/\epsilon}$ \emph{planned} steps (\Cref{plan}) of each subroutine $\reset(k)$ (if operating).
	
%	Intuitively, during the execution of an instance $\reset(k)$, if an element is inserted or deleted on some level in the range $[0, k]$, then the background procedure $\reset(k)$ should also handle it. When an instance of $\reset(k)$ finishes, it will update all information on levels $\leq k$, as well as all the counters $|P_{i}|, |A_{i}|, i\in [0,k]$, and partly on level $(k+1)$ in $O(\log n)$ time by switching pointers to memory locations, and then aborts all other instances $\reset(i), i<k$.

\item \textbf{Initiating $\reset(\cdot)$ Instances.} Upon any element update (deletion or insertion), go over all levels $0\leq k\leq L$ and check if there is currently an instance $\reset(k)$. Denote by $k_1,k_2,k_3, \ldots$ the levels that do not have such an instance, where $k_1 < k_2 < k_3 < \ldots$. Next, we want to partition all levels $k_i$ into \emph{short levels} and \emph{non-short levels}. All of the short levels will be lower than the non-short levels, meaning exists $i$ such that $k_{i'}$ is a short level for any $i'<i$ and a non-short level for any $i' \geq i$. In a nutshell, we will be able to execute a short level reset in a single update step, since the number of elements participating in the reset is small enough. Recall that upon termination of a reset we abort all instances of lower level resets, thus there is no reason to run all short level resets, only the highest one. Regarding the non-short levels, we initiate a reset to each and every one of them. To find the highest short level given $k_1,k_2,k_3, \ldots$ we do as follows. First, count all the first $\frac{L}{f}$ elements, from level $0$ upwards. Say that the $\frac{L}{f}$-th element is at level $j$. Thus, we know that $|\bigcup_{i=0}^{j'} E_i| < \frac{L}{f}$ for any $0 \leq j'<j$. Define $i$ to be the highest such that $k_i < j$. If no such $i$ exists then there are no short levels, otherwise $k_i$ is the highest short level, and we initiate the resets for levels $k_i,k_{i+1}, \ldots$, where again $k_i$ is a short level and the rest are non-short levels.

\end{itemize}

\subsubsection{Background}

For each level   $k\in [0, L]$, any reset instance $\reset(k)$ that operates (in the background) maintains a partial copy of the foreground 
%data structures 
in the background. Specifically, 
%it uses the following data structures.
we maintain and define the following:

	\begin{enumerate}[(1),leftmargin=*]
		\item We maintain subsets of elements $\univ^{(k)}, \univ^{(k)+}\subseteq \univ^+$
		($\univ^{(k)}$ are the alive elements in $\univ^{(k)+}$), and for each element $e\in \univ^{(k)+}$, maintain the two level indices $\lev^{(k)}(e)$ and $\plev^{(k)}(e)$. In addition, we maintain a subset of sets $\sets^{(k)} \subseteq \sets$, and a level value $\lev^{(k)}(s)$ for each $s\in \sets^{(k)}$, as well as a partial set cover solution $\sets^{(k)}_\alg$ that covers all elements in $\univ^{(k)}$
		
		\item For each level $i\in [-1, k+1]$, let $S_i^{(k)} = \{s\mid \lev^{(k)}(s) = i \}$ and for each level $i\in [0, k+1]$, let $E_i^{(k)} = \{e\mid \lev^{(k)}(e) = i\}$.
		
		\item For each element $e\in \univ^{(k)}$, we maintain the assignment $\asn^{(k)}(e)\in\sets^{(k)}_\alg$, and for each set $s\in \sets$ maintain the set $\cov^{(k)}(s)$, where $\cov^{(k)}(s) = \{e \mid \asn^{(k)}(e) = s\}$.
        \end{enumerate}

\begin{definition}\label{plan}
	The procedure $\reset(k)$ could make two different types of steps, \EMPH{immediate} and \EMPH{planned}: An immediate step of the algorithm is executed right away,
whereas a  planned step of the algorithm is stored implicitly in the background, and only executed when scheduled by the main algorithm in the foreground in reaction to element updates. 
\end{definition}

\noindent As mentioned, for each level $k$, if there is currently no instance $\reset(k)$, then we start an instance $\reset(k)$ in the background if $k$ is either a non-short level or it is the highest short level. Then, after each adversarial update step, go over all non-short levels $k\in [0, L]$, and execute $O\brac{f/\epsilon}$ \emph{planned} steps (see \Cref{plan}) of each instance of $\reset(k)$ (if operating), and execute the full reset of the short level reset. Roughly speaking, during the execution of an instance of $\reset(k)$, if an element is inserted or deleted on some level in the range $[0, k]$, then the background procedure $\reset(k)$ should also handle it. When an instance of $\reset(k)$ terminates, it will update all information on levels $\leq k$, and partly on level $k+1$,
%, within $O(\log n)$ time, by switching pointers to memory locations,
and then abort all other instances of $\reset(i), i<k$. 
The main technical part of our algorithm is the procedure $\reset(k)$, which runs in the background. 
Next, we describe the reset procedure, which consists of three phases: (I) initialization, (II) greedy set cover algorithm, and (III) termination.

% as well as all the counters $|P_{i}|, |A_{i}|, i\in [0,k]$,

\paragraph{Phase I: Initialization.} When an instance of $\reset(k)$ has been initiated by the foreground, it sets the following:

\begin{itemize}
    \item $\sets^{(k)}_\alg = S_{i}^{(k)}  \leftarrow \emptyset, \forall i\in [0, k+1]$. 
    %Meaning, the collection of sets that will be chosen to cover the elements up to level $k$ begins empty.
    \item $E_{i}^{(k)} \leftarrow \emptyset, \forall i\in [0, k+1]$ 
    %Elements at each level up to $k$ will be determined by the reset (static greedy algorithm), and begins empty for each level.
    \item $\univ^{(k)}  = \univ^{(k)+} \leftarrow (\bigcup_{i=0}^k E_i) \cap \univ$. Meaning, the elements participating in the reset are the alive elements up to level $k$ in the foreground.
    \item $\sets^{(k)}\leftarrow$ all sets that contain an element in $\univ^{(k)}$. Note that we cannot create $\sets^{(k)}$ directly from the sets $S_k, S_{k-1}, \cdots, S_{-1}$, as there might be several sets at level $-1$ not containing any element in $\univ^{(k)}$, and we do not want these sets to participate in a reset, since it could blow up the update time. 
    \item $\lev^{(k)}(e)\leftarrow -1, \forall e\in \univ^{(k)}$
    \item $\plev^{(k)}(e)\leftarrow \max\{\plev(e), k+1 \}, \forall e\in \univ^{(k)}$. Intuitively, following a reset to $k$ we want all elements participating in this reset ($\univ^{(k)}$) to be active up to at least $k+1$ (without decreasing).
    \item $\lev^{(k)}(s) \leftarrow -1, \forall s\in\sets^{(k)}$
\end{itemize}

%initializes its own data structures by setting $\sets^{(k)}_\alg = S_{i}^{(k)} = E_{i}^{(k)} \leftarrow \emptyset, \forall i\in [0, k+1]$, and 
%$\univ^{(k)}  = \univ^{(k)+} \leftarrow \bigcup_{i=0}^k E_i, \sets^{(k)}\leftarrow \bigcup_{i=-1}^k S_i$. To obtain all the pointers to the sets $E_i, S_i, 0\leq i\leq k$, we use the respective arrays whose entries provide the required pointers, from $E_k, S_k$ down to $E_0, S_{-1}$. 
%One technical issue is that, during the numeration of the sets from $E_k, S_k$ down to $E_0, S_{-1}$, some pointers might have been switched by other instances of resets $\reset(i), i<k$; we will show how to deal with this issue in \Cref{consistency} \todo{hard to understand, as at this stage it's not clear what's the meaning of switching pointers}.

%The collection of sets $\sets^{(k)}$ will contain all sets up to level $k$ (including $k$ and $-1$) that contain an element in $\univ^{(k)+}$ (so $|\sets^{(k)}| \leq f|\univ^{(k)+}|$).

%We also assign $\lev^{(k)}(s) =\lev^{(k)}(e)\leftarrow -1, \plev^{(k)}(e)\leftarrow \max\{\plev(e), k+1 \}, \forall e\in \univ^{(k)}, s\in\sets^{(k)}$. 
\noindent While the level $\lev^{(k)}(e)$ of elements $e$, initialized as $-1$, will be assigned 
a value from $0$ to $k+1$
throughout the execution of $\reset(k)$, the passive level $\plev^{(k)}(e)$ is assigned a value during initialization and does not change throughout the execution of $\reset(k)$. We note that $\plev^{(k)}(e)$ is no smaller than the foreground passive level
$\plev(e)$ of any element $e$, which will guarantee that the passive level of an element is monotone non-decreasing. 
Moreover, $\plev^{(k)}(e)$ is at least $k+1$, which will guarantee that none of the elements that participate in $\reset(k)$ from the initialization may belong to $P_i$, for any level $i \le k$. In addition, there is no effect to any level $j > k+1$, meaning if an element $e$ was $j$-passive before the reset to $k$, it will still be after, and if it was $j$-active, it will still be after.
For each set $s\in \sets^{(k)}$, we store the set $s\cap \univ^{(k)}$ (in a linked list). The stated steps incur a high running time, and as such cannot be executed in the foreground as immediate steps before the next update step occurs (aiming for a low worst-case update time), hence they will be scheduled as planned steps in the background. 
%During the initialization, some elements of $\univ^{(k)+}$ may become dead, hence we remove them from $\univ^{(k)}$ \todo{said in the next sentence}. 
For any new element $e$ that is inserted by the adversary during the initialization, we assign $\lev^{(k)}(e) = -1$ and $\plev^{(k)}(e) \leftarrow k+1$; for any old element $e\in \univ^{(k)+}$ that is deleted by the adversary during the initialization, and as such becomes dead, we remove it from $\univ^{(k)}$. Specific implementation of this phase is described in \Cref{initsec}.

%\tianyi{We need $O(\log_\beta(Rn))$ worst-case runtime to copy the pointers to $\bigcup_{i=0}^k E_i$. Otherwise, these pointers might be overwritten by memory switches below level $k$. So the overall worst-case time is at least $O(\log^2_\beta(Rn))$.}

\paragraph{Phase II: Greedy Set Cover Algorithm.} The algorithm consists of $k+2$ rounds, iterating from level $i = k+1$ down to $i = 0$;
in what follows, by writing ``the $i$th round'' we refer to the round that corresponds to level $i$. 
%corresponds to level $k+2-i$, for $i = 1,\ldots,k+2$. 
During the process, the algorithm maintains a collection $U\subseteq \univ^{(k)}$, which is the collection of all alive elements that have not been covered yet by the gradually growing $\sets^{(k)}_\alg$, and for each set $s\in \sets^{(k)}\setminus \sets^{(k)}_\alg$, it maintains all elements in $s\cap U$ (in a linked list). 
\begin{comment}
Every set $s\in \sets^{(k)}\setminus \sets^{(k)}_\alg$ will be stored in a {\em truncated max-heap} (see \Cref{heap}) with key equal to $|s\cap U| / \cost(s)$; later on \todo{give ref} we describe the implementation of this max-heap so that heap operations take constant time rather than $O(\log m)$ time \todo{we anyway spend $f$ time per iteration, and it's OK to assume that $f = \Omega(\log n)$, otherwise we'll use the primal dual algorithm}.    
\end{comment}
%$\brac{\ceil{\log_\beta \cnt(s)/\cost(s)}, \ceil{\log_2 1/\cost(s)}}$; in words, this heap gives higher priority for sets with higher value of $\cnt(s) / \cost(s)$, in favor of smaller set costs for breaking ties; later on we will discuss the implementation of this max-heap so that heap operations take constant time rather than $O(\log n)$ time.
At the beginning of the $i$th round, we make the assumption below, which will be proven in \Cref{assum}.
\begin{assumption}\label{i-round}
	The following two conditions hold {\em at the beginning} of the $i$th round. Importantly, these conditions do not necessarily hold throughout the $i$th round.
	\begin{itemize}[leftmargin=*]
		\item All elements in $U$ are alive; this holds for any round $i = k+1,\ldots, 0$.
		\item For any round $i = k,\ldots,0$, $|s\cap U| / \cost(s) < \beta^{i+1}$; for $i = k+1$, there is no upper bound on $|s\cap U| / \cost(s)$.
	\end{itemize}
\end{assumption}

\noindent During the $i$th round, the following steps will be scheduled as planned.

\begin{framed}
	\noindent \textbf{Planned Steps in the $i$th Round.}
	During the $i$th round, the algorithm iteratively chooses a set $s\in \sets^{(k)}\setminus \sets^{(k)}_\alg$ that maximizes $|s\cap U| / \cost(s)$ such that $|s\cap U| / \cost(s) \geq \beta^i$. This will be implemented by a \emph{truncated max-heap} (see \Cref{GSCAP} for details). If no such $s$ exists and $i>0$, we proceed to the next ($i-1$) round.  
 %of top priority
	%from the truncated max-heap, with a lower bound threshold of $\beta^i$. 
	%As long as the top priority set $s$ satisfies $|s\cap U| / \cost(s) \geq \beta^i$, we perform the following steps. 
	Add $s$ to $\sets^{(k)}_\alg$, assign $\lev^{(k)}(s) \leftarrow i$, and then go over all alive elements $e\in s\cap U$ and assign $\lev^{(k)}(e)\leftarrow i, \asn^{(k)}(e)\leftarrow s$; note that $\plev^{(k)}(e)$ was already assigned for elements that existed at the beginning of this phase, and, as described below, it is also assigned for newly inserted elements.  
After that, we remove $e$ from $U$ and enumerate all sets $s^\prime \ni e$ to maintain
%the linked list 
$s^\prime\cap U$. 
%[[S: should also update the key in the truncated heap]]

	%\tianyi{Doesn't this take $O(\log^2n)$ time in total per element? We should do something nontrivial to calculate the counters.}
\end{framed}

\noindent See \Cref{reset} for the pseudo-code of the planned steps in the greedy set cover algorithm. 

\begin{comment}
	\begin{algorithm}
		\caption{Phase II of $\reset(k)$ - Greedy Set Cover Algorithm} \label{reset}
		\begin{algorithmic}
			\FOR {$i$ from $k+1$ downto 0}
			\STATE {NoSets $\leftarrow$ FALSE}
			\STATE {\textbf{while} (!NoSets)}
			\STATE {Choose $s$ from max-heap with top priority}
			\IF {$|s \cap U|/\cost(s) < \beta^i$}
			\STATE {NoSets $\leftarrow$ TRUE}
			\ELSE
			\STATE {$\sets^{(k)}_\alg \leftarrow \sets^{(k)}_\alg \cup \{s\}$}
			\STATE {$\lev^{(k)}(s) \leftarrow i$}
			\FOR {$e \in s \cap U$}
			\STATE {$\lev^{(k)}(e) \leftarrow i$}
			\STATE {$\asn^{(k)}(e) \leftarrow s$}
			\STATE {$U \leftarrow U \setminus \{e\}$}
			\FOR {$s' \ni e$}
			\STATE {$s' \cap U \leftarrow s' \cap U \setminus \{e\}$}
			\ENDFOR
			\ENDFOR
			\ENDIF
			\STATE {\textbf{end while}}
			\FOR {$j$ from $i$ to $k$}
			\STATE {Update $|P^{(k)}_j|$ and $|A^{(k)}_j|$}
			\ENDFOR
			\ENDFOR
		\end{algorithmic}
	\end{algorithm}	
\end{comment}

\begin{algorithm}
	\caption{Phase II of $\reset(k)$ - Greedy Set Cover Algorithm} \label{reset}
	\For{$i$ from $k+1$ to 0}{
		NoSets $\leftarrow$ FALSE\;
		\While{(!NoSets)}{
			choose a set $s\in \sets^{(k)}\setminus \sets^{(k)}_\alg$ that maximizes $|s \cap U|/\cost(s)$\;
			\If{($|s \cap U|/\cost(s) < \beta^i$) \textbf{\emph{or}} ($\sets^{(k)}\setminus \sets^{(k)}_\alg = \emptyset$)}{
				NoSets $\leftarrow$ TRUE\;
			}
			\Else{
				$\sets^{(k)}_\alg \leftarrow \sets^{(k)}_\alg \cup \{s\}$\;
				$\lev^{(k)}(s) \leftarrow i$\;
				\For{$e \in s \cap U$}{
					$\lev^{(k)}(e) \leftarrow i$\;
					$\asn^{(k)}(e) \leftarrow s$\;
					$U \leftarrow U \setminus \{e\}$\;
					\For{$s' \ni e$}{
						$s' \cap U \leftarrow s' \cap U \setminus \{e\}$\;
                            Update heap\;
					}
     }
				}
			}
		}
		%\For{$j$ from $i$ to $k$}{
		%	update $|P^{(k)}_j|$ and $|A^{(k)}_j|$\;
		%}
\end{algorithm}

Finally, we describe how to handle adversarial element updates that are fed to the background during the $i$th round.
\begin{itemize}[leftmargin=*]
	\item \textbf{Deletions in the Background.} 
	\begin{comment}
		First we describe how to maintain the counter information. There are two different cases.
		\begin{itemize}[leftmargin=*]
			\item If $e\notin U$, then increase all the counters $|P^{(k)}_j|$, and decrease all the counters $|P^{(k)}_j|$, $\forall \lev^{(k)}(e)\leq j< \plev^{(k)}(e)$, by applying the data structure of \Cref{interval}.
			
			If $e\in U$, then decrease all the counters $|P^{(k)}_j|$, $\forall i\leq j< \plev^{(k)}(e)$, by applying the data structure of \Cref{interval}.
		\end{itemize}
		
		After updating the counter information,
	\end{comment}
	 Suppose that an element $e$ is deleted by the adversary during the $i$th round. We mark $e$ as dead (thus it joins $\univ^{(k)+} \setminus \univ^{(k)}$), and assign $\plev^{(k)}(e)\leftarrow \lev^{(k)}(e)$. If $e$ is in $U$ at the moment, we remove $e$ from $U$, enumerate all sets $s\ni e$, update the linked list $s\cap U$, and update the truncated max-heap.
	
	\item \textbf{Insertions in the Background.} 
	Suppose that an element $e$ is inserted by the adversary during the $i$th round.
	Enumerate all sets $s \ni e$, and proceed as follows:
	\begin{itemize}[leftmargin=*]
	\item If $e$ belongs to a set in $\sets^{(k)}_\alg$ (and thus covered by $\sets^{(k)}_\alg$), then find such a set $s\in \sets^{(k)}_\alg$ that maximizes $\lev^{(k)}(s)$, and then assign $\lev^{(k)}(e) = \plev^{(k)}(e)\leftarrow \lev^{(k)}(s)$. Note that such elements will not have $\plev^{(k)} \ge k+1$ as elements that exist at the beginning of the execution of $\reset(k)$.
	
	%After that, increment all counters $|P^{(k)}_j|$, $\plev^{(k)}(e)\leq j\leq k$, by applying the data structure of \Cref{interval}.
	
	\item Otherwise, $e$ is not covered by $\sets^{(k)}_\alg$, in which case we assign $\plev^{(k)}(e)\leftarrow i, \lev^{(k)}(e)\leftarrow -1
$. Such elements might be covered throughout this or subsequent rounds of $\reset(k)$, which will change their $\lev^{(k)}(e)$ to be the round in which they are covered, i.e.,  at most $i$, but their $\plev^{(k)}(e)$ will remain $i$; note also the difference from elements that existed at the beginning of the execution of $\reset(k)$.
%After that, increment all counters $|P^{(k)}_j|$, $i\leq j\leq k$, by applying the data structure of \Cref{interval}.

        \item Regardless of whether $e$ belongs to a set in $\sets^{(k)}_\alg$ or not, we need to add all sets not in $\sets^{(k)}$ containing $e$ to $\sets^{(k)}\setminus \sets^{(k)}_\alg$. There are at most $f$ such sets, and for each such set $s'$ we know that $s' \cap U = \{e\}$, since otherwise $s'$ would have already been in $\sets^{(k)}$. Therefore, we can update the heap in $O(f)$ time following the insertion of $e$.

	\end{itemize}
\end{itemize}

\paragraph{Phase III: Termination.} When all $(k+2)$ rounds of the greedy set cover algorithm terminate, we set $S_i$ and $E_i$ (foreground) to be $S_i^{(k)}$ and $E^{(k)}_i$, respectively, for each $i\in [0, k]$. Then, append the linked list of $S_{k+1}^{(k)}$ and $E_{k+1}^{(k)}$ to $S_{k+1}$ and $E_{k+1}$. Finally, we abort all lower-level reset instances. Specific implementation is described in \Cref{termsec}.

%Finally, for each $0\leq i\leq k$, update $|P_{i}|\leftarrow |P^{(k)}_{i}|, |A_{i}|\leftarrow |A^{(k)}_{i}|$.
%\tianyi{The updates to counters for $i>k$ is wrong, because of possible conflicts with other schedulers. In fact, we don't need to update the counters at all upon termination.}

\subsection{Implementation Details and Update Time Analysis} \label{updtime}
%\shay{note that some of the update time analysis is given in Section 2.4; perhaps the titles of these sections should change}
In this section we will describe the maintained data structures and analyze the worst-case update time for each part of the algorithm separately.

\begin{framed}
	\noindent \textbf{Data Structures that Link between the Foreground and Background.}
	\begin{enumerate}[(a)]

        \item For each $k$, we have pointers to the sets $S^{(k)}_i$ and $E^{(k)}_i$, stored in two arrays of size $L+2$ and $L+1$, respectively (an entry for every $i \in [-1,L]$ and $i \in [0,L]$, respectively). In addition, the head of the list $S^{(k)}_i$ and $E^{(k)}_i$ keeps a Boolean value which indicates whether it is in the foreground or not.
        
		\item We store an array in the foreground $\lev[\cdot]$ indexed by $s\in \sets$ and $e\in \univ$. So the size of this array should be $O(|\sets| + |\univ|)$. Here we have assumed that sets and elements have unique identifiers from a small integer universe (if the sets and elements belong to a large integer universe and assuming we would like to optimize the space usage, we can use hash tables instead of arrays). For each $s\in \sets$ and $0\leq k\leq L$, $\lev[s][k]$ stores a pointer to the memory location containing the value of $\lev^{(k)}(s)$, as well as a pointer to the list head of $S^{(k)}_i$ if $\lev^{(k)}(s) = i$ (and $i \neq -1$). Similarly, $\lev[e][k]$ stores a pointer to the memory location of $\lev^{(k)}(e)$, as well as a pointer to the memory location of the pointer to $E^{(k)}_i$ if $\lev^{(k)}(e) = i$.
		%[[S: unclear; what's the size of this array? the assumption is that all $s$ and $e$ have unique identifiers from a small enough universe, may need to stress that this can be handled using hashing otherwise (which requires randomization)]]
		
		%\item For each $0\leq k\leq L$ and $0\leq i\leq k$, if $S^{(k)}_i\neq \emptyset$, then keep a flag bit $\flag[k][i]\in \{0, 1\}$ that indicates whether the set $S^{(k)}_i$ is currently on the foreground. If it is on the foreground, then $\flag[k][i] = 1$; if it hasn't been switched to the foreground, or switched away from the foreground by some other $\reset(*)$, then $\flag[k][i] = 0$.
		
		%[[S: first explain that we have pointers to the sets $S_i^{(k)}$ and $E_i^{(k)}$, which are stored in two arrays of size $L+1$ each;
		%then can explain that we also have this additional pointer to the flag]].		
		%Each set $S_i^{(k)}$ and $E_i^{(k)}$ should contain a pointer to the memory location of $\flag[k][i]$. So given any pointer to $S_i^{(k)}$ or $E_i^{(k)}$, we can check in constant time if this collection of sets or elements are currently in the foreground.
		
	%\item For each $s\in \sets$ and $0\leq k\leq L$, $\lev[s][k]$ stores a pointer to the memory location of $\lev^{(k)}(s)$, as well as a pointer to the list head of $S^{(k)}_i$ if $\lev^{(k)}(s)\neq -1$. Similarly, $\lev[e][k]$ stores a pointer to the memory location of $\lev^{(k)}(e)$, as well as a pointer to the memory location of the pointer to $E^{(k)}_i$.

        %\item The head of the list $S^{(k)}_i$ and $E^{(k)}_i$ keeps a Boolean value which indicates whether it is in the foreground or background.

    \end{enumerate}
\end{framed}

\subsubsection{Foreground Operations}
The collections $S_i$ and $E_i$ for each $i$ will be maintained in doubly linked lists in the foreground. To access the level value $\lev(s)$ in the foreground, we can enumerate all indices $k\in [0, L]$ and check the entry $\lev[s][k]$ that points to $\lev^{(k)}(s) = i$ and list head $S^{(l)}_{i}, l\geq i$. If either $\lev[s][k]$ is a null pointer, or $\lev[s][k]$ pointers to a value $\lev^{(k)}(s) = i$ but $S_i^{(l)}$ is not in the foreground, then we know $\lev^{(k)}(s)\neq\lev(s)$. Therefore, accessing the foreground level value $\lev(s)$ takes time $O(L)$. Similarly, accessing the foreground level value $\lev(e)$ takes time $O(L)$ as well. We conclude that upon insertion of element $e$, enumerating all sets $s \ni e$ to decide which set needs to cover $e$ takes time $O(f \cdot L) = O(f \cdot \log_\beta(Cn))$. As can be seen in \Cref{del} and \Cref{ins}, the runtime of other parts is constant, except for the part of feeding to a reset, which will be analyzed in \Cref{GSCAP}. Regarding initiating resets, in $O(L) = O(\log_{\beta}(Cn))$ time we can go over all levels to check which ones need to be initiated, and also check which is the highest short level, by using $E_i$ which is maintained in the foreground. We remind that a short level is a level $k$ such that $\reset(k)$ is not working currently in the background, and $f \cdot \sum_{i=0}^k |E_i| < L$. Moreover, there is only one short level reset operating per update step (the highest one).

\subsubsection{Reset - Initialization Phase} \label{initsec}
%Copying the memory locations of the pointers of $\{E_k, E_{k-1}, \cdots, E_0\}$ and $\{S_k, S_{k-1}, \cdots, S_{-1}\}$ takes time at most $O(f|\univ^{(k)+}|)$ [[S: why not $O(k)$?]] since we have a linked list that connects pointers to nonempty sets. The bottleneck of this part is enumerating all the elements in $\univ^{(k)+}$ and their sets [[S: this indeed takes more time, which we need, but it's not the copying of the memory locations]], which takes a planned total time of $O(f|\univ^{(k)+}|)$. If the initialization is interrupted by an element update, then the worst-case update time is $O(1)$. Obtaining the collection $\mathcal{S}^{(k)}$ takes $O(f|\univ^{(k)}|)$ time, and this is the bottleneck of the initialization phase.
When an instance $\reset(k)$ has been scheduled, it initializes its own data structures by setting $\sets^{(k)}_\alg = S_{i}^{(k)} = E_{i}^{(k)} \leftarrow \emptyset, \forall i\in [0, k+1]$, and $\univ^{(k)} = \univ^{(k)+} \leftarrow \bigcup_{i=0}^k E_i$, each maintained in a doubly-linked list. After that, enumerate all elements of $\univ^{(k)}$ and let $\sets^{(k)}$ be all the sets containing at least one element in $\univ^{(k)}$. So $\sets^{(k)} \subseteq \{s\mid -1\leq\lev(s)\leq k\}$. 
%Note that the reason why we don't create $\sets^{(k)}$ directly from the sets $S_k, S_{k-1}, \cdots, S_0$ is because there might be some missing sets on level $-1$ containing some elements in $\univ^{(k)+}$.
To obtain all the pointers to $E_i$, $0\leq i\leq k$, we follow the linked list consisting of these pointers, from $E_k$ down to $E_0$. One remark is that, during the numeration of the list from $E_k$ down to $E_0$, which could take several update steps, some pointers might already be switched by other resets $\reset(i), i<k$; we will show that this is still fine in \Cref{consistency}. 
%The collection of sets $\sets^{(k)}$ will contain all sets up to level $k$ (including $k$ and $-1$) that contain an element in $\univ^{(k)+}$ (so $|\sets^{(k)}| \leq f|\univ^{(k)+}|$).
We also assign $\lev^{(k)}(s) =\lev^{(k)}(e)\leftarrow -1, \plev^{(k)}(e)\leftarrow \max\{\plev(e), k+1 \}, \forall e\in \univ^{(k)}, s\in\sets^{(k)}$. For each set $s\in \sets^{(k)}$, store the set $s\cap \univ^{(k)}$ as a linked list. All the above steps will be planned in the background for the non-short levels. 
%During the initialization, remove all dead elements from $\univ^{(k)+}$. 
If a new element $e$ is inserted by the adversary during initialization, assign $\lev^{(k)}(e)\leftarrow -1, \plev^{(k)}(e)\leftarrow k+1$, and add all sets $s \notin \sets^{(k)}$ containing $e$ to $\sets^{(k)}$, and enumerate $s \cap \univ^{(k)}$; if an old element $e\in \univ^{(k)}$ is deleted by the adversary during the initialization, remove it from $\univ^{(k)}$.

Copying the memory locations of the pointers of $\{E_k, E_{k-1}, \cdots, E_0\}$ and $\{S_k, S_{k-1}, \cdots, S_{-1}\}$ takes time at most $O(k)$. Obtaining the collection $\mathcal{S}^{(k)}$ and calculating $s\cap \univ^{(k)}$ for each set $s\in \sets^{(k)}$ takes $O(f|\univ^{(k)}|)$ time, so the total running time of this phase is $O(f|\univ^{(k)}| + k) = O(f|\univ^{(k)}| + L)$. Since all the non-short levels satisfy $f \cdot |\univ^{(k)}| \geq L$, we get that the running time of this phase for non-short levels is $O(f|\univ^{(k)}|)$, and for the short level reset it is $O(L) = O(\log_{\beta}(Cn))$.

\subsubsection{Reset - Greedy Set Cover Algorithm Phase} \label{GSCAP}

%When the algorithm is interrupted by an element update, we need to scan all the sets containing this updated element and calculate set levels both in the background and in the foreground, which takes time $O(f\log_\beta(Cn))$ worst-case time across all levels $k\in [0,L]$.

According to the algorithm, for each element in $\univ^{(k)}$, the greedy set cover algorithm procedure scans the list of all the sets containing this element at most once, and so the planned number of sets the algorithm goes through is $O(f|\univ^{(k)}|)$. For each such set $s'$, we must update $s' \cap U$ (line $15$ in \Cref{reset}), which indeed takes $O(1)$ time, but we must store the values of $|s' \cap U|$ for each $s' \in \sets^{(k)}$ in some data structure that will allow us to update values (line $16$ in \Cref{reset}) and extract the maximum value (line $4$ in \Cref{reset}) in $O(1)$ time, if we want the reset to run in $O(f|\univ^{(k)}|)$ time. We shall implement this with a \emph{truncated max-heap}:

%When it is interrupted by element updates, according to \Cref{interval}, the worst-case update time of maintaining the counters $|A^{(k)}_i|, |P^{(k)}_i|$ is $O(1)$ for each $k$.

%[[S: I haven't read the heap implementation, since this seems minor (and as I pointed out, we can do using an ordinary heap)]]
%[[S: should also update the key in the truncated heap (Amitai: After we remove $e$ from $U$ and enumerate all sets $s' \ni e$ to maintain the linked list $s' \cap U$)]]

\begin{definition}\label{heap}
	Let $X$ be a set of objects with key values. A \emph{truncated} max-heap data structure on $X$ supports the following operations.
	\begin{itemize}[leftmargin=*]
		\item Removal of any object in $X$.
		\item Change the key value of any object in $X$.
		\item Given a threshold value $k$, return any object with maximum key value or whose value is $\geq k$.
	\end{itemize}
\end{definition}

\noindent During the $i$th round we need to repeatedly choose sets from the max-heap with top priority with threshold $\beta^i$. To implement heap operations of the truncated max-heap in constant time, store all sets $s\in \sets^{(k)}\setminus \sets^{(k)}_\alg$ in an array of length $L$, each entry of the array is a linked list of elements in $\sets^{(k)}\setminus \sets^{(k)}_\alg$. We will ensure the following property of this array.

%\todo{S: I think we can make do with standard heap, since the log factor overhead is dominated by the $f$ time that we need to process any element (assuming $f = \Omega(\log n)$}

\begin{invariant}[heap invariant]\label{heap-inv}
	During the $i$-round of the greedy set cover algorithm phase of $\reset(k)$, for any index $0\leq j\leq i$, if $j<i$, then the $j$th entry of the array is a linked list of sets in $\sets^{(k)}\setminus \sets^{(k)}_\alg$, such that for each of these sets $s$, we have $\floor{\log_\beta (\frac{|s\cap U|}{\cost(s)})} = j$; otherwise, if $j = i$, then $\floor{\log_\beta (\frac{|s\cap U|}{\cost(s)})} \geq j$; if $j>i$, then the entry of the array is an empty list.
\end{invariant}

Initialization of this heap data structure takes $O(f|\univ^{(k)}| + k)$ time. Upon insertions/deletions occurring during this initialization, we can update the heap data structure under construction in a straightforward manner, since we are not yet updating the levels nor building the set cover solution. During the $i$th round of the greedy set cover algorithm phase, removal of any object from the heap can be done in constant time by linked list operations. If the key value $j = \floor{\log_\beta (\frac{|s\cap U|}{\cost(s)})}$ has changed either by the adversary or by the background algorithm itself, we can attach it to the ($\min\{j, i\}$)th linked list in the array of the heap. During the $i$th round, extraction operation on the heap always has threshold $\beta^i$, so it suffices to check if the $i$th entry of the array is empty. All of these operations take constant time. When the $i$th round has finished, it must be the case that the $i$th entry in the array of the heap has now become empty. So, when we enter the $(i-1)$th round of the greedy set cover algorithm phase, \Cref{heap-inv} still holds.

Since we must pass through all levels from $k+1$ to $0$, we conclude that the total running time of this phase of the reset is also $O(f|\univ^{(k)}| + k) = O(f|\univ^{(k)}| + L)$. Again, for all non-short levels this means $O(f|\univ^{(k)}|)$, and for the short level this means $O(L) = O(\log_{\beta}(Cn))$. We conclude that the first two phases together of the reset (initialization and greedy set cover algorithm), run in $O(\log_{\beta}(Cn))$ time for the short level reset. Thus, we will compute these two phases all within a single update step, without affecting the worst-case update time. As for the other resets, as mentioned we will compute $O(\frac{f}{\epsilon})$ steps per update steps, thus the total number of computations per update step will be $O(\frac{f \cdot L}{\epsilon}) = O(\frac{f \cdot \log_{\beta}(Cn)}{\epsilon})$, which is the bottleneck. 

\subsubsection{Reset - Termination Phase} \label{termsec}

\noindent When $\reset(k)$ has finished and called upon to switch its memory to the foreground, for every index $0\leq i\leq k$, we replace every pointer to $S_i$ with the pointer to $S^{(k)}_i$, and connect all pointers to nonempty lists $S^{(k)}_i, 0\leq i\leq k$ with a linked list. 

Merging the list $S_{k+1}^{(k)}$ with the list $S_{k+1}$ on the foreground is in fact done in the greedy set cover algorithm phase, following round $k+1$, as it does not have worst-case runtime guarantee. Nevertheless it will be explained here as part of the termination phase. Upon merging $S_{k+1}^{(k)}$ with $S_{k+1}$, suppose $S_{k+1}$ is equal to some $S_{k+1}^{(l)}$ for some $l\geq k+1$; in other words, $S_{k+1}$ was computed in some instance $\reset(l)$, where $l\geq k+1$. Then, for each set $s\in S^{(k)}_{k+1}$, redirect the pointer $\lev[s][k]$ from the list head $S^{(k)}_{k+1}$ to the list head $S_{k+1}^{(l)}$. After that, concatenate the two lists $S_{k+1}^{(k)}, S_{k+1}$. We can merge the two lists $E_{k+1}^{(k)}$ and $E_{k+1}$ in a similar manner. The running time of this is clearly $O(f \cdot |\univ^{(k)}|)$, thus not the bottleneck of the second phase.

As for the runtime of the termination phase, the number of memory pointers to be switched is at most $O(L) = O(\log_\beta(Cn))$, so the worst-case runtime is $O(L) = O(\log_\beta(Cn))$.
%\paragraph{Basic operations.} To implement fast access to the level value $\lev^{(k)}(s)$ and $\lev^{(k)}(e)$ for any set $s$ and element $e$ in constant time, we also need to specify the following data structures.
For any set $s\in \sets^{(k)}$, to access any level value $\lev^{(k)}(s)$ in the background, we can follow the pointer $\lev[s][k]$ and check $\lev^{(k)}(s)$ in constant time; similarly we can check the value of $\lev^{(k)}(e)$ for any $e\in \univ^{(k)+}$ in constant time.

Since our data structures maintain multiple versions of the level values, and our algorithm keeps switching memory locations between the foreground and the background, we need to argue the consistency of the foreground data structures; this is done in \Cref{consistency}.

By the algorithm description, upon an element update, we scan all the levels, and if any instance of $\reset(\cdot)$ has just terminated, we switch the one $\reset(k)$ on the highest level $k$ to the foreground as we have discussed in the previous paragraph, while aborting all other instances of $\reset(i), i<k$. As we have seen, switching a single instance of $\reset(\cdot)$ to the foreground takes $O(L) = O(\log_\beta(Cn))$ time, and so the worst-case time of this part is $O(L) = O(\log_\beta(Cn))$. Therefore, the termination can be done in a single update step.

\bigskip

\noindent To conclude \Cref{updtime}, any insertion/deletion in the foreground can be dealt with in $O(f \cdot L)$ time, thus it can be done within a single update step. Dealing with the short level reset (finding the highest one, initializing the reset, and executing the greedy set cover algorithm) takes $O(L)$ time, thus it will all be done within a single update step as well. Likewise, termination can be done in $O(L)$ time (updating data structures of the highest finished reset and aborting the rest). Every other $\reset(k)$ will take $O(f \cdot |\univ^{(k)}|)$ time (initialization and greedy set cover algorithm). Since this cannot be executed within a single update step, for each $k$ we will execute $O(\frac{f}{\epsilon})$ computations per update step. In \Cref{proofcorrect} we will prove that by working in such a pace, we can ensure that \Cref{inv} holds, which in turn will be enough to prove that the approximation factor holds, by \Cref{toy-approx}.

\subsection{Proof of Correctness} \label{proofcorrect}

Let us first show that our memory switching scheme preserves consistency of the data structures on the foreground. This is nontrivial because different foreground sets $S_k$ may come from different background copies $S_k^{(l)}$, and they might not be compatible with each other as different instances of $\reset(\cdot)$ may have different views of the data structures.
\begin{claim}\label{consistency}
	The foreground data structures $\{S_k\}_{-1\leq k\leq L}, \{E_k\}_{0\leq k\leq L}$ are consistent; that is, they satisfy their specifications in \Cref{warmup}.
\end{claim}
\begin{proof}
    This statement is proved by an induction on time. Consider any instance of $\reset(k)$. First we show that during its initialization phase, $\reset(k)$ obtains a set of pointers $S_i, E_i, 0\leq i\leq k$ which comes from a valid realization of the basic data structure defined in \Cref{warmup}. This is nontrivial because the when initialization phase scans from $i = k, k-1, \ldots, 0$, some $\reset(j), j<k$ could terminate and change part of the data structures on the foreground. We will show that this is still fine.
	
	Recall that, during the initialization phase of $\reset(k)$, we have planned a procedure that scans and copy all the pointers to lists $S_i, E_i$ where $i$ goes from $k$ down to $0$. If no $\reset(j), j<k$ has switched its memory to the foreground, then in the end, $\reset(k)$ can obtain a prefix of the foreground data structures. Otherwise, suppose an instance of $\reset(j), j<k$ has switched its memory to the foreground while $\reset(k)$ hasn't finished copying all the pointers.
	\begin{itemize}[leftmargin=*]
		\item If $j<i$ at the moment of memory switching, then since $\reset(j)$ only modifies levels in $[0, j+1]$, $\reset(k)$ will still copy a prefix of the new version of the foreground data structures.
		\item If $j\geq i$, then $\reset(k)$ will be copying the old version of $S_0, S_1, \cdots, S_i$, which will be switched off from the foreground, but still consistent on its own. In other words, $\reset(k)$ will in the end obtain an stale version of the foreground data structure which is still consistent with itself.
	\end{itemize}
	
	Next, let us consider the moment when $\reset(k)$ finishes and switch its local memory to the foreground. As $\reset(k)$ was not aborted, no other instances $\reset(j), j>k$ has finished; also, by the algorithm description, any $\reset(j), j<k$ cannot change anything in the foreground on levels higher than $k$. Therefore, when $\reset(k)$ switches, the foreground data structure is still consistent.
\end{proof}

Next, we prove that \Cref{i-round} always holds during Phase II of procedure $\reset(k)$ (the greedy set cover algorithm) in the background.
\begin{claim} \label{assum}
	\Cref{i-round} always holds.
\end{claim}
\begin{proof}
	Let us prove this statement (that the two conditions of \Cref{i-round} hold) by a reverse induction on $i$, from $i = k+1$ to 0. 
	For the induction basis $i = k+1$, this assumption trivially holds: For the first condition, any element deletion in $\univ^{(k)}$ triggers directly a removal of the element from $U$, rather than marking it as dead; the second condition holds vacuously.
	Next, consider the induction step. For the first condition, when an element in $U$ is deleted, it is removed right away from $U$, so $U$ always contains alive elements only. For the second condition, the terminating condition of the $i$th round implies that when the $i$th round terminates, any set $s\in \sets^{(k)}\setminus \sets^{(k)}_\alg$ satisfies $|s\cap U| / \cost(s) < \beta^i$, implying that the second condition for round $i-1$ holds at the beginning of round $i-1$.
\end{proof}

Finally, we prove that \Cref{inv} always holds in the foreground, which concludes the proof of \Cref{warmup}.

\begin{claim}\label{inv1}
	\Cref{inv}(1) always holds in the foreground.
\end{claim}

\begin{proof}

Consider the insertion of element $e$ in the foreground to level $j$. By the description of the algorithm, we would have $\lev(e) = \plev(e) = j$. Since $\lev(e) = \plev(e)$, by definition $e$ cannot be part of $A_i$ for any $i$. Thus, the insertion cannot raise $|N_i(s)|$ for any $i$ and $s$. Likewise, upon deletion of element $e$, we set $\plev(e) \leftarrow \lev(e)$, thus again this deletion cannot raise $|N_i(s)|$ for any $i$ and $s$.

Since insertions and deletions to the foreground cannot cause a violation to \Cref{inv}(1), the only possibility left is if \Cref{inv}(1) is violated in the background in some system working on $\reset(k)$, and it is then transferred to the foreground. Assume by contradiction that exists a set $s$ such that upon termination of $\reset(k)$, we have $\frac{|N_j(s)|}{\cost(s)} \geq \beta^{j+1}$ for some $j$.

If $j \geq k+1$, then for $|N_j(s)|$ to grow between right before the initialization of $\reset(k)$ and right after the termination of $\reset(k)$, there must exist an element $e \in s$ that joined $N_j(s)$ sometime between those two time steps. For this to happen $e$ must change its passive level, since its level cannot fall from $>j$ to $\leq j$ in an instance of $\reset(k)$ where $j \geq k+1$ (because $e$ would not participate in such a reset). This means that when $\reset(k)$ is initialized, either $\plev(e) \leq j$ or $e$ has not been inserted yet, and $\plev(e) > j$ when $\reset(k)$ terminates. Consider the first case where $\plev(e) = i$ when the reset is initialized, where $i \leq j$. Then by the algorithm description $\plev(e)$ is $\max\{i,k+1\}$ when the reset terminates. But clearly $j \geq \max\{i,k+1\}$ and so we reach a contradiction. Now consider the second case where $e$ was inserted sometime during the execution of $\reset(k)$. By the algorithm description $\plev(e)$ is $\leq k+1$ when the reset terminates, and since $j \geq k+1$ we again reach a contradiction.

If $j < k+1$, we claim that all elements in $N_j(s)$ at the termination of $\reset(k)$ were in $U$ in the beginning of round $j$. Assume by contradiction that exists an element $e \in N_j(s)$ when the reset terminates, such that $e \notin U$ in the beginning of round $j$. If $e$ was inserted during round $j$ or after, then at termination $\plev(e) \leq j$ by the algorithm description, a contradiction to the fact that $e \in N_j(s)$ when the reset terminates. If $e$ existed at the beginning of round $j$, but was not in $U$ at that time, it means it has already been covered in a previous round, meaning covered at a level $>j$, which again is a contradiction to the fact that $e \in N_j(s)$ when the reset terminates. So indeed all elements in $N_j(s)$ at the termination of $\reset(k)$ were in $U$ in the beginning of round $j$. But since we assumed by contradiction that upon termination $\frac{|N_j(s)|}{\cost(s)} \geq \beta^{j+1}$, we get that $\frac{|s \cap U|}{\cost(s)} \geq \beta^{j+1}$ in the beginning of round $j$, a contradiction to \Cref{i-round}. 
\end{proof}

\begin{claim}\label{inv2}
	\Cref{inv}(2) always holds in the foreground.
\end{claim}

\begin{proof}
	It is immediate by the algorithm description that the second half of \Cref{inv}(2) (that $\lev(s) = -1$ for each $s \notin \sets_\alg$) always holds. It remains to show that $|\cov(s)| / \cost(s) \ge \beta^{\lev(s)}$, for any set $s \in \sets_\alg$. Since $\cov(s)$ can contain dead elements as well, clearly a deletion cannot cause a violation to this invariant, as well as an insertion. Thus the only possibility left is if \Cref{inv}(2) is violated in the background in some system working on $\reset(k)$, and it is then transferred to the foreground. Assume by contradiction that exists a set $s$ such that upon termination of $\reset(k)$, we have $|\cov(s)| / \cost(s) < \beta^{\lev(s)}$.
		
		By the description of the greedy set cover algorithm, when a set $s$ joins the partial solution $\sets^{(k)}_\alg$ during the $i$th round, it is guaranteed that $|\cov(s)| / \cost(s) = |s\cap U| / \cost(s) \geq \beta^i = \beta^{\lev^{(k)}(s)}$;
		recalling that $\cov(s)$ also includes dead elements, this ratio $|\cov(s)| / \cost(s)$ may only increase later on.
	%When $\reset(k)$ terminates and switch its memory to the foreground, it does not harm the inequality $|N_{j}(s)| / \cost(s) < \beta^{j+1}$ for $j\geq k+1$, the passive level of any element processed by $\reset(k)$ has increased to at most $k+1$ (if it was higher than $k+1$, then it does not change). Therefore, the memory switch step does not harm \Cref{inv}(1), and this concludes the proof.
	%[[S: should also explain that switching to the foreground won't affect anything up to level $k+1$ as well.
	%Should go over each of the relevant invariants and argue shortly why they hold as a result of the memory switch for every level (up to $k+1$ and above $k+1$]] 
\end{proof}

\begin{claim} \label{pre4inv}
	Consider an instance of $\reset(k)$ and denote by $U_i$ the collection of uncovered alive elements at the beginning of round $i$ in the greedy set cover algorithm phase, where $i\in [0, k+1]$. Then by working at a pace of $O(f / \epsilon)$ per update step, the reset will terminate after less than $\frac{\epsilon}{2} |U_i|$ element updates following the beginning of round $i$.
 \end{claim}

\begin{proof}
	Note that at the beginning of round $i$, the total number of remaining planned computations by $\reset(k)$ is $O(f|U_i|)$, since each remaining alive element will change its level only once, and it will take $O(f)$ time for it to update all relevant data structures regarding this change. Each time the adversary makes an update, the number of alive elements can increase by at most one, so if we mark by $x$ the total number of update steps from the beginning of round $i$ until termination, there are at most $|U_i| + x$ elements that we need to assign to levels, taking $O(f(|U_i|+x))$ time. Notice though that $x$ is roughly $\epsilon |U_i|$, since the reset takes care of $\Theta(\frac{1}{\epsilon})$ elements per update step ($\Theta(\frac{f}{\epsilon})$ computations per update step, each element taking $\Theta(f)$ time), in which time only one can be inserted. So $O(f(|U_i|+x)) = O(f|U_i|)$, meaning the reset takes $O(f|U_i|)$ time for planned and immediate steps. By working at a pace of $O(\frac{f}{\epsilon})$ computations per update step, we can finish all steps in less than $\frac{\epsilon}{2} |U_i|$ update steps. 
\end{proof}

\begin{claim} \label{claim4inv}
	Following the termination of an instance $\reset(k)$, we have $|P_i| < \epsilon \cdot |A_i|$ for all $i \leq k$.
\end{claim}

\begin{proof}
	Consider the moment when the $i$th round of the greedy set cover algorithm begins in an instance of $\reset(k)$ (within Phase II of $\reset(k)$), where $i\in [0, k+1]$. Denote by $U_i$ the set of all currently uncovered alive elements.
	%Note that at this point, the total number of (planned and immediate) remaining steps by $\reset(k)$ is bounded by $O(f|U|)$, where recall that $U$ is the set of all currently uncovered alive elements. Each time the adversary makes an update, $|U|$ can increase by at most one, and $O(f / \epsilon)$ steps of $\reset(k)$ will be executed. Therefore, the reset can work at a pace such that there are less than $\frac{\epsilon}{2} |U|$ element updates from the adversary before $\reset(k)$ terminates. 
	By Claim \ref{pre4inv}, the reset will terminate in less than $\frac{\epsilon}{2} |U_i|$ update steps. By the algorithm description,
	each element $e\in U_i$ 
	 at the beginning of the $i$th round will be assigned $\lev(e)$ at most $i$ in round $i$ onwards, while $\plev(e) > i$. The only elements that could have level $\leq i$ with passive level also $\leq i$, and thus
	belong to $P^{(k)}_{i}$, are the ones inserted/deleted by the adversary from the $i$th round onwards. As mentioned, there are less than $\frac{\epsilon}{2} |U_i|$ insertions/deletions from this point until termination, so at the end of the reset $|P^{(k)}_{i}|<\frac{\epsilon}{2} |U_i|$. Moreover, $(|A^{(k)}_{i}| + |P^{(k)}_{i}|)$ at the end of the reset is the total number of elements that were assigned to a level up to $i$. 
 %To give a lower bound to this, we know that there are less than $\frac{\epsilon}{2} |U_i|$ deletions from the $i$th round onwards [[S: I think that the deleted ``dead'' elements, are also counted in $P^{(k)}_{i}$, so I'm not sure we need to consider them separately]]. 
 So eventually there will be at least $|U_i|$
 %$ - \frac{\epsilon}{2} |U_i|$ 
 elements (alive and dead) covered at level $\leq i$, meaning $(|A^{(k)}_{i}|  + |P^{(k)}_{i}|) \geq |U_i|$
 %- \frac{\epsilon}{2} |U_i|$ [[S: my understanding is that	$(|A^{(k)}_{i}|  + |P^{(k)}_{i}|) \ge |U_i|$, since even the dead elements are counted here]]. 
 Combining this altogether we obtain that $|P^{(k)}_{i}| < \epsilon |A^{(k)}_{i}|$ for any $\epsilon < 0.5$. 
	%and only one such insertion is enough for there would be more than $|U|$ elements eventually at level $\leq i$, we get that $|A^{(k)}_{i}| \geq |U| \geq \frac{1}{\epsilon}|P^{(k)}_{i}|$.
\end{proof}

\begin{claim} \label{larger}
	When $\reset(k)$ terminates and is transferred to the foreground, it does not change $A_{k'}$, and $P_{k'}$ can only reduce in size, for any $k'>k$.
\end{claim}

\begin{proof}

First, we will show that each element in $P_{k'}$ in the foreground before the termination of $\reset(k)$ remains there after the termination. Consider such an element $e \in P_{k'}$. $\plev^{(k)}(e) \leq k'$ upon termination, since otherwise that would mean that $\plev(e) > k'$ following $\reset(k)$, and since this cannot happen due to the reset (the passive level would be at least $k+2$), we must have had $\plev(e) > k'$ in the foreground upon termination, a contradiction to the fact that $e \in P_{k'}$ in that time. The only exception is if $e$ was a dead element, therefore it was in $P_{k}$ thus in $P_{k'}$ as well, and $e$ would be completely removed during the reset, reducing $|P_{k'}|$ by one.

Second, we will show that each element in $A_{k'}$ in the foreground before the termination of $\reset(k)$ remains there after the termination. Consider such an element $e \in A_{k'}$. Clearly $\lev^{(k)}(e) \leq k'$ upon termination, since $e$ cannot rise to a level above $k+1$, which is $\leq k'$. Since $e \in A_{k'}$, the passive level of $e$ in the foreground right before termination of $\reset(k)$ is $> k'$, thus it is also $> k+1$. Therefore, the passive level of $e$ would not have changed due to $\reset(k)$, and $e$ remains in $A_{k'}$.

Next, we will show that no new elements join $P_{k'}$ following the termination of $\reset(k)$. Assume by contradiction that exists such an element $e$. Either $\lev(e) > k'$ or $\plev(e) > k'$ in the foreground right before the termination. If $\lev(e) > k'$, then the only way $e$ would have participated in $\reset(k)$, is if before the termination of $\reset(k)$ there was a termination of some $\reset(k'')$ where $k'' \geq k'$. This would abort $\reset(k)$ though. Likewise, if $\plev(e) > k'$ in the foreground right before the termination, then since the passive level cannot decrease, $e$ cannot join $P_{k'}$.
%then before the termination of $\reset(k)$ there was a termination of some $\reset(k'')$ where $k'' \geq k'$, which again, would abort $\reset(k)$.

Lastly, we will show that no new elements join $A_{k'}$ following the termination of $\reset(k)$. Assume by contradiction that exists such an element $e$. Either $\lev(e) > k'$ or $\plev(e) \leq k'$ in the foreground right before the termination. If $\lev(e) > k'$, then the only way $e$ would have participated in $\reset(k)$, is if before the termination of $\reset(k)$ there was a termination of some $\reset(k'')$ where $k'' \geq k'$. This would abort $\reset(k)$ though. If $\plev(e) \leq k'$, then since $k' \geq k+1$, by the algorithm description $\plev^{(k)}(e)$ will remain $\leq k'$, a contradiction, and the claim follows. 
%$P_k$ doesn't become too large - first, no element that was previously in $A_k$ joins $P_k$, as they had plev at least $k+1$, and the active level will not exceed $k$, so they will remain in $A_k$. 
%	Next, what about newly inserted elements and elements of $A_{k-1} \cup P_{k-1}$ (I focus on the case $k'+1 = k$)? 
%	They may belong to $P_k$. As for the ``old elements'', we can argue that they must have already belonged to $P_k$ - if they move to $P_k$, then their passive level is at most $k$ also before the reset (if it were higher than $k$ previously, and as it won't decrease, they wouldn't move to $P_k$). As for the new elements; first, we can employ the previous claim to reason that only $\eps/2 \cdot |U_{k'}| \le \approx \eps \cdot |A_k|$ new elements may join $P_k$, but I think we can say something stronger: all these newly inserted elements that might increase $|P_{k}|$ in the background are anyway increasing $P_k$ in the foreground, so it seems that there is no special overhead due to the background treatment of new elements.
\end{proof}
%Since each element update can raise $|P^{(k)}_{i}|$ by only one, at the end of the reset procedure we have $|P^{(k)}_{i}|\leq \epsilon |U|$. On the other hand, b

%Next, we show that the third part of \Cref{inv} holds.
\begin{lemma}\label{inv3}
	\Cref{inv}(3) always holds in the foreground.
\end{lemma}

\begin{proof}

Assume by contradiction that there exists $k$ such that $|P_k| > 2\epsilon \cdot |A_k|$ for some $k$, at the end of update step $t$ (right before update step $t+1$). Let $t'$ be the last time step before $t$ that $\reset(k)$ was initiated. $\reset(k)$ was initiated at time $t'$ right after the termination of some $\reset(k')$, where $k' \geq k$, by the algorithm description. By \Cref{claim4inv}, we know that at time $t'$ (following the termination of $\reset(k')$) we have $|P_k| < \epsilon \cdot |A_k|$ in the foreground.

%time step such that at the end of update step $t'-1$ we have $|P_k| < \epsilon \cdot |A_k|$, but at the end of update step $t'$ we have $|P_k| \geq \epsilon \cdot |A_k|$, and this continues to hold from update step $t'$ all the way to update step $t$. Notice that at update step $t'$, $\reset(k)$ is already working in the background, since by the algorithm description once an instance of $\reset(k')$ for any $k'\geq k$ terminates, we immediately begin an instance of $\reset(k)$. The trigger that caused $|P_k|$ to surpass $\epsilon \cdot |A_k|$ at update step $t'$ must be an element insertion/deletion, since the termination of some $\reset(k')$ for any $k'<k$ does not effect $A_k$ or raise the size of $P_k$ by \Cref{larger}, and for any $k' \geq k$ we would get that $|P_k| < \epsilon \cdot |A_k|$ by Claim \ref{claim4inv}. 
 
%	Since due to one insertion/deletion only one element can leave $A_k$ and join $P_k$, it follows that $|P_k| < \epsilon \cdot (|A_k|+1) + 1$ holds immediately after update step $t'$.  
Between $t'$ and $t$, $|P_k|$ can rise and $A_k$ can change only due to insertions/deletions, since by definition of $t'$ no termination of $\reset(k')$ for any $k' \geq k$ exists, and any termination of $\reset(k')$ for any $k' < k$ would not raise $|P_k|$ or change $A_k$ by \Cref{larger}. By denoting $A_k$ at times $t'$ and $t$ by $A_k^{t'}$ and $A_k^{t}$ respectively, we conclude that there must be more than ($2\epsilon \cdot |A_k^{t}| - \epsilon \cdot |A_k^{t'}|$) update steps between $t'$ and $t$, since again due to one insertion/deletion only one element can join $P_k$. 
	
%As mentioned, at update step $t'$, $\reset(k)$ is already working in the background. Say this reset was initiated at time step $t'' \leq t'$, and let $\tau = t'' - t'$ (so $\tau \geq 0$). 
If we denote by $x$ the total number of update steps throughout which $\reset(k)$ is executed, we get by Claim \ref{pre4inv} that:

%$x < \frac{\epsilon}{2}(|P_k^{t'}| + |A_k^{t'}|)$, 

%Thus, the remaining number of update steps throughout which $\reset(k)$ is executed following $t'$, would be $x - \tau < \frac{\epsilon}{2}(|P_k^{t''}| + |A_k^{t''}|) - \tau$. Now, notice that in each step between $t''$ to $t'$, $(|P_k| + |A_k|)$ cannot reduce. An insertion can only raise it, a deletion would not change it, termination of a reset to a level $< k$ would not change it by \Cref{larger}, and termination of a reset to a level $\geq k$ would abort this reset. Thus:
 
% $$ x - \tau < \frac{\epsilon}{2}(|P_k^{t''}| + |A_k^{t''}|) - \tau \leq \frac{\epsilon}{2}(|P_k^{t'}| + |A_k^{t'}| + \tau) - \tau, $$
 
%    \noindent meaning that:

 %The initial collection $U = U_{k+1}$ of alive elements that participate in the reset is all of the elements at level up to $k$ (including), which is exactly $|P_k^{t'}| + |A_k^{t'}|$, plus up to all inserted elements during the reset process. If we denote by $x$ the total number of update steps throughout which $\reset(k)$ is executed, we get by Claim \ref{pre4inv}:
	
	\begin{equation} \label{inv3eq1}
		x < \frac{\epsilon}{2}(|P_k^{t'}| + |A_k^{t'}|),
	\end{equation} 
	
	\noindent since the collection of all elements participating in $\reset(k)$ initiated at time $t'$ is $P_k^{t'} \cup A_k^{t'}$. We assume that at update step $t$, $\reset(k)$ is still running, thus we will reach a contradiction if:
	
	\begin{equation} \label{inv3eq2}
		\frac{\epsilon}{2}(|P_k^{t'}| + |A_k^{t'}|) < 2\epsilon \cdot |A_k^{t}| - \epsilon \cdot |A_k^{t'}|.
	\end{equation} 
	
	\noindent Now, notice that:
	
	\begin{equation} \label{inv3eq3}
		|A_k^{t}| \geq |A_k^{t'}| - x,
	\end{equation} 
	
	\noindent since again, in each update step during the reset up to one element can be removed from $A_k$, and we assumed that $x \geq t-t'$. Thus, we need to show that:
	
	\begin{equation} \label{inv3eq4}
		\frac{\epsilon}{2}(|P_k^{t'}| + |A_k^{t'}|) < 2\epsilon \cdot (|A_k^{t'}|-x) - \epsilon \cdot |A_k^{t'}|.
	\end{equation} 
	
	\noindent Plugging in $x$ from \Cref{inv3eq1} and rearranging, we get that we need to show:
	
	\begin{equation} \label{inv3eq5}
		(\frac{1}{2} + \epsilon)|P_k^{t'}| < (\frac{1}{2} - \epsilon)|A_k^{t'}|.
	\end{equation} 
	
	\noindent Since we know that $|P_k^{t'}| < \epsilon|A_k^{t'}|$, it is enough to show that:

        \begin{equation} \label{inv3eq5.5}
		(\frac{1}{2} + \epsilon) \cdot \epsilon|A_k^{t'}| < (\frac{1}{2} - \epsilon)|A_k^{t'}|,
	\end{equation} 

\noindent meaning that 

        \begin{equation} \label{inv3eq6}
		2\epsilon^2 + 3\epsilon -1 < 0,
	\end{equation} 

 \noindent which holds for any $\epsilon < \frac{1}{4}$.	
%	\noindent Rearranging, we get that this inequality holds for any $\epsilon < 0.1$ and $|A_k^{t'}| > \frac{5}{\epsilon}$. In the special case where $t'=t$, $\reset(k)$ will terminate within a single update step, at the end of time step $t'$, since a reset with $O(\frac{1}{\epsilon})$ participating elements is computed in $O(\frac{f}{\epsilon} + \log_{\beta}n)$ running time. 
Thus, we reach our contradiction and the lemma follows.
\end{proof}

We conclude that our algorithm indeed maintains \Cref{inv}, in worst-case update time of $O(\frac{f \cdot L}{\epsilon}) = O\brac{\frac{f\log(Cn)}{\epsilon^2}}$ as proved in \Cref{updtime}. Since \Cref{inv} holds, the approximation factor of the maintained minimum set cover is $(1+\epsilon)\ln n$, as shown in \Cref{toy-approx}. This concludes the proof of \Cref{warmup}.

\section{Removing Dependency on Aspect Ratio} \label{remove}
In this section we prove \Cref{wc}, by removing the dependency on the aspect ratio $C$ in the update time.

\subsection{Preliminaries and Basic Data Structures}
As before, an element is \emph{dead} if it is supposed to be deleted from $\univ$ by the adversary, but currently resides in the system as our algorithm has not removed it yet; an element is \emph{alive} if it is not dead. $\univ^+\supseteq \univ$ is the collection of all dead and alive elements. Define $\beta = 1+\epsilon$. For each $s\in \sets$, define the \emph{top level} $\tlev(s) = \ceil{\log_\beta (n / \cost(s))}$. 
Define a parameter $K = \ceil{10\log_\beta n}$. We will partition the hierarchy into two sequences of windows, the \emph{even}
and the \emph{odd} indexed for parameter $l$:
$[l K, (l+2)K)$, and note that every even-indexed window overlaps two odd-indexed windows, and vice versa (except the top/bottom ones); see \Cref{fig1} for an illustration.
%\shay{I suggest to start by explaining that we partition the hierarchy into two sequences of windows, the even and the odd indexed for parameter $l$:$[l K, (l+2)K)]$, and to note that every even-indexed window overlaps two odd-indexed windows (except for the bottom/top ones, which overlap one), and to refer to the picture}
\begin{definition}\label{window}
For any $l\geq -1$, define $\sets_l = \{s\in\sets\mid \tlev(s)\in [l K, (l+2)K) \}$ and let $\mathcal{U}_l$ be the collection of all elements that the cheapest set containing them is in $\mathcal{S}_l$. 
\end{definition}

\begin{obs} \label{ob1}
Each set $s \in \mathcal{S}$ is in two consecutive set collections, $\mathcal{S}_l$ and $\mathcal{S}_{l+1}$, and in those two only. Each element $e \in \mathcal{U}$ is in two consecutive element collections,  $\mathcal{U}_{l'}$ and $\mathcal{U}_{l'+1}$, and in those two only.
\end{obs}

\begin{obs} \label{ob2}
By Observation \ref{ob1}, the union of all odd-indexed element collections  equals the union of all even-indexed element collections, which equals $\mathcal{U}$. Namely, $\univ = \bigcup_{j\geq 0}\univ_{2j} = \bigcup_{j\geq 0}\univ_{2j-1}$. 
\end{obs}

\noindent For each $l\geq -1$, we will maintain a set cover $\sets_{l, \alg}\subseteq \sets_l$ that covers $\univ_l$ by applying \Cref{warmup} as a black-box
on the set cover instance $(\univ_l, \sets_l)$. So, for each $l\geq -1$, the algorithm will maintain a super set $\univ_l^+\supseteq \univ_l$ which contains both alive and dead elements. Overall, by \Cref{ob2}, we will have two different set cover solutions for $\univ$:
$$\sets_{\even, \alg} := \bigcup_{j\geq 0}\sets_{2j, \alg} \; , \; \sets_{\odd, \alg} := \bigcup_{j\geq 0}\sets_{2j-1, \alg}.$$
The solution of  smaller total cost among $\sets_{\even, \alg}$ and $ \sets_{\odd, \alg}$ will be the output solution that is presented to the adversary.

To illustrate how \Cref{warmup} is applied on the set cover instance $(\univ_l, \sets_l)$, let us slightly open the black box and define the following notations. 
%\shay{it's slightly inaccurate: defining the hierarchy of intervals is part of the algorithm; for example, a set that's not in the solution goes to level -1 in the alg of sec 2, and here it'll go to the min level of $I_l$} \tianyi{So, I will withdraw this explanation since it would confuse readers.}
For any index $l\geq -1$, for each set $s\in \sets_l$, we will assign a level:
%\shay{unclear, I think the meaning is that this is a dynamic level that the black-box alg assigns to set $s$ in the $l$th batch; also, may want to explain the need for $I_l$ (it's the batch of levels in which we operate now (instead of $-1$ to $L$). Update, I see that some clarifications are given in Invariant 3.1, but this needs to be discussed here to build the intuition for the readers}
$$\lev_l(s)\in I_l = \left[\max\{lK - \ceil{\log_\beta n}-1, -1\}, (l+2)K + \ceil{10\log_\beta 1/\epsilon}\right].$$
Intuitively, $\lev_l(s)$ is the dynamic level that the black-box algorithm assigns to set $s$ in the $l$-th window $\sets_l$, and $I_l$ is the batch or interval of levels in which we operate now, instead of $-1$ to $L$. Notice that $s \in \sets_l$ cannot cover at a level higher than $\tlev(s)$, which is less than the upper limit of $I_l$. On the other hand, $s \in \sets_l$ cannot cover at a level lower than $\floor{\log_\beta (\frac{1}{\cost(s)})}$ (its level if it covers just one element), which is larger than the lower limit of $I_l$. Therefore, $I_l$ contains all levels that $s \in \sets_l$ can cover at (and has some slack on both ends). The lowest level in each interval $I_l$ is reserved for sets in $\sets_l \setminus \sets_{l,\alg}$, just as we reserved the level $-1$ for sets not in the cover in the previous section. See \Cref{fig1} for an illustration.

%the algorithm will ensure that $\lev_l(s) \le (l+2)K$ and thus $\lev_l(s) \le \max(I_l) -\ceil{10\log_\beta 1/\epsilon}$.

%\shay{should state the properties that the alg of Sec 2 needs from this batch (including $\lev_l(s) \le (l+2)K$, hence $\lev_l(s) \le \max(I_l) -\ceil{10\log_\beta 1/\epsilon}$, which is crucial for the proof of Lemma 3.1 (or 2.1)), and explain that they hold also for the $l$th batch; see also the email that I sent}.

For any element $e\in \univ_{l}^+$, we will assign it to set $\asn_l(e)\in \sets_{l, \alg}$, and conversely, for each $s\in \sets_l$, define $\cov_l(s) = \{e\mid \asn_l(e) = s\}$. The level of an element $e$ is defined to be $\lev_l(e) = \lev_l(\asn_l(e))$. Besides, we will also maintain a value of \emph{passive level} $\plev_l(e)\geq \lev_l(e)$. For any $k\in I_l$, define $A_{l, k} = \{e\in \univ_{l}^+ \mid \lev_l(e) \leq k < \plev_l(e) \}$, and $P_{l, k} = \{e\in \univ_{l}^+ \mid \plev_l(e)\leq k\}$. For each set $s\in \sets_l$, define $N_{l, k}(s) = A_{l, k}\cap s$. For dead element $e$ in $\univ_l^+\setminus \univ_l$, we will have $\plev_l(e) = \lev_l(e)$.

	\begin{figure}
		\center{\includegraphics[scale=0.45]{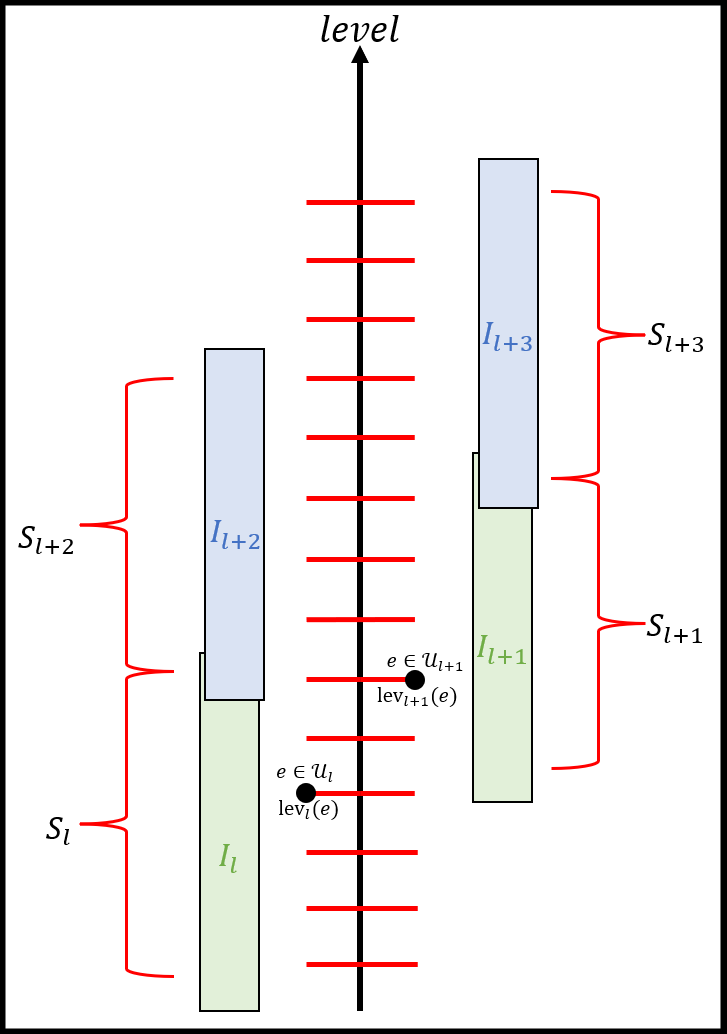}}
		\caption{\label{fig1} \emph{Element $e$ is in $\mathcal{U}_l$ and $\mathcal{U}_{l+1}$. Thus it must be covered by a set in $\mathcal{S}_l$, at level $\lev_l(e) \in I_l$, and it must be covered by a set in $\mathcal{S}_{l+1}$, at level $\lev_{l+1}(e) \in I_{l+1}$. Notice the overlap between the different intervals.}}
	\end{figure}

\begin{comment}
\begin{figure}
	\centering
	\input{figs/scale-free}
	\caption{\label{fig1} \emph{Element $e$ is in $\mathcal{U}_l$ and $\mathcal{U}_{l+1}$. Thus it must be covered by a set in $\mathcal{S}_l$, at level $\lev_l(e) \in I_l$, and it must be covered by a set in $\mathcal{S}_{l+1}$, at level $\lev_{l+1}(e) \in I_{l+1}$. Notice the overlap between $I_l$ and $I_{l+1}$. \shay{I think there's a problem here - $I_l$ and $I_{l+2}$ overlap, and same goes for $I_{l+1}$ and $I_{l+3}$, I think the figure illustrates the windows of $\sets_l$ and not of $I_l$}}}
\end{figure}
\end{comment}

Our algorithm will maintain the following data structures for each index $l\geq -1$.

\begin{framed}
	\noindent \textbf{Basic Data Structures for each $l\geq -1$.}
	
	\begin{enumerate}[(a)]
		\item For each level $k\in I_l$, we will maintain linked lists $S_{l, k} = \{s\mid \lev_l(s) = k \}$ and $E_{l, k} = \{e\mid \lev_l(e) = k\}$.
		
		% \item Plus, we will explicitly maintain counters $|A_{l, k}|$ and $|P_{l, k}|$. Note that we don't try to maintain the sets $N_{l, k}(s) = A_{l, k}\cap s$.
		
		\item For each element $e\in \univ_{l}$, maintain the assignment $\asn_l(e)$, and then for each set $s\in \sets_l$ maintain the set $\cov_l(s)$.
	\end{enumerate}
\end{framed}

%\shay{the following invariant should look pretty much identical to 2.1: the analog to $k \in [0,L]$ should not be $k \in I_l$, since level -1 isn't included; condition (2) should look more similar; in condition (3) should remind that we don't maintain these counters; condition (4) should be moved outside of the invariant}
%\tianyi{I think there is no better places for (4)}

\begin{invariant}\label{inv-scale}
    The algorithm will maintain the following invariant for each $l\geq -1$:
	\begin{enumerate}[(1),leftmargin=*]
		\item For any set $s\in \sets_l$ and any $k > \min\{I_l\}$, we have $\frac{|N_{l, k}(s)|}{\cost(s)} < \beta^{k+1}$.
		
		\item For any set $s\in \sets_{l, \alg}$, we have $\frac{|\cov_l(s)|}{\cost(s)}\geq \beta^{\lev_l(s)}$; note that $\cov_l(s)$ could count dead elements in $\univ_{l}$ that are already deleted by the adversary.
        In particular, $\lev_l(s)\leq \tlev(s)$ for any $s\in \sets_l$. Moreover, for each $s\notin \sets_{l, \alg}$, $\lev_l(s) =\min\{I_l\}$.
		
		%Besides, for each $s\notin \sets_{l, \alg}$, $\lev_l(s) =\min(I_l)$. Plus, $\lev_l(s)\leq \tlev(s)$ for any $s\in \sets_l$.
		
		\item For each $k\in I_l$ and $k > \min\{I_l\}$, we have $|P_{l, k}|\leq 2\epsilon |A_{l, k}|$. Note that the algorithm does not maintain the counters $|A_{l, k}|, |P_{l, k}|$ explicitly.
		
		%\item For dead element $e$ in $\univ_l^+\setminus \univ_l$, we have $\plev_l(e) = \lev_l(e)$.
		
		%\item For each element $e\in \univ$, if $e\in \univ_{2j}$, then any set $s\in \bigcup_{i>j}\sets_{2i}$ does not contain $e$; similarly, if $e\in \univ_{2j-1}$, then any set $s\in \bigcup_{i>j}\sets_{2i-1}$ does not contain $e$.
	\end{enumerate}
\end{invariant}

\noindent To analyze the approximation ratio, let $\sets^*$ denote the optimal set cover for all elements in $\univ$.

%\shay{I suggest to make the following change; Applying \Cref{toy-approx} for the set system of $(\mathcal{U}_l, \mathcal{S}_l)$, we directly get that $\cost(\sets_{l, \alg}) \le (1+O(\epsilon))\ln n \cdot \cost(\sets^*_l) \le (1+O(\epsilon))\ln n \cdot \cost(\sets^*)$, which saves most of the proof of the following lemma. So there are two options: either we keep the following statement as is, and then the proof starts with this one line observation, and then it continues with the ``Finally ...'' part in the end which shows that $\sum_{j\geq 1}\cost(\sets_{l+2j, \alg})\leq \epsilon\cdot \cost(\sets_{l,\alg})$. Or we can have a claim for  $\sum_{j\geq 1}\cost(\sets_{l+2j, \alg})\leq \epsilon\cdot \cost(\sets_{l,\alg})$, and then derive as a corollary of this claim the statement of this lemma}

% \shay{I stopped reading here for now}
%\shay{perhaps some explanation is required, especially to address the difference between the window $I_l$ and the window $[-1,L]$ (maybe this needs to be addressed before this lemma)}
%\tianyi{I have changed Sec 2 so that $L = \ceil{\log_\beta(Cn)} + \ceil{10\log_\beta 1/\epsilon}$ rather than $L = 2\ceil{\log_\beta(Cn)}$. So $I_l$ and $[-1, L]$ are the same now}

\begin{lemma}\label{approx}
	For any $l\geq -1$, let $k_0\in I_l$ be the smallest index such that there exists a set in $\mathcal{S}_l$ covering at level $k_0$, or in other words, $k_0$ is the smallest level such that $\sets_{l, \alg}\cap S_{l, k_0}\neq \emptyset$. If $k_0 \leq (l+1)K$, then we have:
	$$\sum_{j \geq 0}\cost(\sets_{l + 2j, \alg})\leq (1+O(\epsilon))\ln n \cdot \cost(\sets^*).$$
\end{lemma}
\begin{proof}
    We apply \Cref{pre-toy-approx} and \Cref{toy-approx} for the set system of $(\mathcal{U}_l, \mathcal{S}_l)$, where instead of summations from $k=0$ to $L-1$ we have summations from $\min\{I_l\} +1$ to $\max \{I_l\} - 1$, instead of $L$ we have $\max\{I_l\}$, and instead of $\lev(e)$ and $\plev(e)$ we have $\lev_l(e)$ and $\plev_l(e)$, respectively. As a result we directly get that 
$\cost(\sets_{l, \alg}) \le (1+O(\epsilon))\ln n \cdot \cost(\sets^*_l) \le (1+O(\epsilon))\ln n \cdot \cost(\sets^*)$, where $\sets^*_l$ denotes the optimal set cover for all elements in $\univ_l$. 
%Next, let us prove that $\sum_{j\geq 1}\cost(\sets_{l+2j, \alg})\leq \epsilon\cdot \cost(\sets_{l,\alg})$. 
Consider any set $s_0\in \sets_{l, \alg}\cap S_{l, k_0}$. For any set $s\in \sets_{l+2j, \alg}$ and $j>0$, we have: 
	\begin{equation} \label{costs} \cost(s)\leq n \cdot \beta^{1-\tlev(s)}\leq n\cdot \beta^{1-(l+2)K} \leq n\cdot \beta^{1-k_0-K} \leq \frac{\epsilon}{(1+2\epsilon) n}\cdot\beta^{-k_0-1}\leq \frac{\epsilon\cdot\cost(s_0)}{(1+2\epsilon) n},
 \end{equation}
	
\noindent where the first inequality holds since  $\tlev(s) \leq \log_\beta (n / \cost(s)) + 1$
(by definition of $\tlev(s)$), the second holds since $\tlev(s) \geq (l+2j)K \ge (l+2)K$ for any $j>0$ by definition of $s$, the third follows from the initial assumption that $k_0 \leq (l+1)K$, the fourth by definition of $K$ and that $\epsilon > \frac{2}{n^8}$ is not too small, and the fifth holds since $s_0$ is covering at level $k_0$. %Since each element is assigned to one set in the set cover, and only up to a $\frac{2\epsilon}{1+2\epsilon}$-fraction of the elements can be dead (otherwise \Cref{inv-scale}(3) is violated) 
Since at most $n$ elements are alive and at most $2\eps \cdot n$ elements may be dead (otherwise
\Cref{inv-scale}(3) is violated), and as each element (dead or alive) is assigned to at most one set in the set cover solution, it follows that the total number of sets in
$\sets_{l+2j, \alg}$ over all $j>0$
is bounded by $(1+2\epsilon) n$. By \Cref{costs}, we get $$\sum_{j>0}\cost(\sets_{l+2j, \alg}) ~\leq~
(1+2\eps)n \cdot \left(\frac{\epsilon}{(1+2\epsilon) n}\cdot\cost(s_0)\right)
~\le~
\epsilon\cdot\cost(s_0) ~\leq~ \epsilon\cdot \cost(\sets_{l, \alg}).$$ We conclude that
$$\sum_{j \geq 0}\cost(\sets_{l + 2j, \alg}) ~=~ \cost(\sets_{l, \alg}) + \sum_{j>0}\cost(\sets_{l+2j, \alg}) ~\leq~ (1+\epsilon) \cdot \cost(\sets_{l, \alg}) ~\leq~ (1+O(\epsilon))\ln n \cdot \cost(\sets^*),$$
\noindent and the lemma follows.
\end{proof}

\begin{lemma} \label{approxf}
$\min\{\cost(\sets_{\even, \alg}), \cost(\sets_{\odd, \alg}) \}\leq (1+O(\epsilon))\ln n \cdot \cost(\sets^*)$
\end{lemma}
\begin{proof}
Let $l'$ be the smallest value such that $\mathcal{U}_{l'} \neq \emptyset$. Pick an arbitrary element $e$ in both $\mathcal{U}_{l'}$ and $\mathcal{U}_{l'+1}$ with the lowest $\lev_{l'+1}(e)$ value; such an element exists by definition of $l'$ and by Observation \ref{ob1}. Write $k_1 = \lev_{l'+1}(e)$, and denote the set that covers $e$ at level $k_1$ by $s_1 \in \mathcal{S}_{l'+1,\alg}$. By Observation \ref{ob1}, $s_1$ is also in either $\mathcal{S}_{l'}$ or $\mathcal{S}_{l'+2}$. If $s_1$ were in $\mathcal{S}_{l'+2}$, then since $s_1$ contains $e$, by definition $e$ would be in $\mathcal{U}_{l'+2}$, but it is not. Thus $s_1 \in \mathcal{S}_{l'},\mathcal{S}_{l'+1}$. Since $s_1 \in \mathcal{S}_{l'}$, by definition of $\mathcal{S}_{l'}$ we know that $\tlev(s_1) < (l'+2)K$. Moreover, $k_1 \leq \tlev(s_1)$ since $s_1$ cannot cover at a level higher than its top level. Thus, we get that $k_1 < (l'+2)K$. 

Now, denote by $k_0 \in I_{l'+1}$ the smallest index such that there exists a set in $\mathcal{S}_{l'+1}$ covering at level $k_0$. We have $k_0 \leq k_1$ by definition, thus $k_0 < (l'+2)K$. Plugging this in Lemma \ref{approx} (with our $l'+1$ as $l$ in the lemma), we obtain:
$$\sum_{j \geq 0}\cost(\sets_{l' +1 + 2j, \alg})\leq (1+O(\epsilon))\ln n \cdot \cost(\sets^*).$$
 
 Since for $j<0$ there are no elements in $\mathcal{U}_{l'+1+2j}$,
 Observation \ref{ob2} implies that the set $\bigcup_{j\geq 0}\sets_{l'+1+2j, \alg}$ is a valid set cover for all elements in $\mathcal{U}$. 
 Note also that $\bigcup_{j\geq 0}\sets_{l'+1+2j, \alg}$ 
 %is a valid set cover, and it 
 is equal to either $\sets_{\even, \alg}$ or $\sets_{\odd, \alg}$, depending on whether $l'$ is odd or even, respectively. Therefore, 
$$\min\{\cost(\sets_{\even, \alg}), \cost(\sets_{\odd, \alg}) \} ~\leq~ \sum_{j \geq 0}\cost(\sets_{l' +1 + 2j, \alg}) ~\leq~ (1+O(\epsilon))\ln n \cdot \cost(\sets^*),$$
which completes the proof. 
\end{proof}

\subsection{Algorithm Description}
%We apply the algorithm that is provided by \Cref{warmup} as a black-box, as follows.
%Following element insertions and deletions, 
For each index $l\geq -1$, 
we will  maintain
the subset of sets $\sets_l\in \sets$ and the subset of elements $\univ_l\subseteq \univ$, to coincide with \Cref{window};
we will also apply the algorithm provided by \Cref{warmup},
as a black-box, on the set system $(\univ_l, \sets_l)$. 
However,  we cannot afford to run the black-box algorithm on all set systems following every update.
Instead, since each element may belong to at most two set systems, element updates are handled in the following manner:
\begin{itemize}[leftmargin=*]
	\item \textbf{Insertion.} When an element $e$ is inserted, enumerate all sets containing $e$ to find the cheapest one, denoted $s$. Letting $l$ denote the index such that $s \in \mathcal{S}_{l},\mathcal{S}_{l+1}$, we add $e$ to both $\mathcal{U}_{l}$ and $\mathcal{U}_{l+1}$, and run the element insertion algorithm on the two set systems $(\univ_{l}, \sets_{l})$ and $(\univ_{l+1}, \sets_{l+1})$.
	
	\item \textbf{Deletion.} When an element $e \in \mathcal{U}_{l},\mathcal{U}_{l+1}$ is deleted, remove $e$ from $\mathcal{U}_{l}$ and $\mathcal{U}_{l+1}$, and run the element deletion algorithm on the two set systems $(\univ_{l}, \sets_{l})$ and $(\univ_{l+1}, \sets_{l+1})$.
\end{itemize}

\noindent Upon each update step, if the inserted/deleted element belongs to $\mathcal{U}_{l}$ and $\mathcal{U}_{l+1}$, then we only run the reset operations for these two systems. By \Cref{warmup}, the worst-case update time of our algorithm is   $O\brac{\frac{f\log n}{\epsilon^2}}$. Also, \Cref{warmup} implies that all conditions of \Cref{inv-scale} are preserved for every set system;
recalling that the solution of  smaller total cost among $\sets_{\even, \alg}$ and $ \sets_{\odd, \alg}$ is the output solution, \Cref{approxf} implies that the approximation factor is in check. This concludes the proof of \Cref{wc}.

\section{Extension to the Low-Frequency Regime} \label{primdu}
In this section we present a dynamic set cover algorithm with an improved worst-case update time, in the low-frequency regime of $f = O(\log n)$. We will mostly follow the known algorithm with worst-case update time $O\brac{\frac{f\log^2(Cn)}{\epsilon^3}}$ \cite{bhattacharya2021dynamic,bhattacharya2019new}, and focus on the adjustments that we make, omitting most of the details that remain the same.

%\begin{theorem}\label{extend}
%	For any set system $(\univ, \sets)$ with set cost range $[1/C, 1]$ that undergoes a sequence of element insertions and deletions, where the frequency is always bounded by $f$, and for any $\epsilon \in (0, 1)$, there is a dynamic algorithm that maintains a $((1+\epsilon)f)$-approximate minimum set cover in $O\brac{\frac{f\log n}{\epsilon^2}}$ deterministic worst-case update time.
%\end{theorem}

%\shay{I copied the following to Sec 2, and edited it there; it should be mostly removed from here, keeping just a couple of sentences}
Before unfolding the technical details, let us explain on a high level why we are able to shave the extra $\log_\beta n$ factor in the time bound, where recall that $\beta = 1 + \epsilon$. Similar to the high-frequency regime, the algorithm of \cite{bhattacharya2021dynamic} also assigns elements and sets to levels at most $O(\log_\beta(Cn))$. Moreover, as before, for each level $k$, there is a $\reset(k)$  instance running on a chunk of local memory in the background, which is {\em disjoint from} and {\em independent of} the hierarchical data structure on the foreground. 
%When the instance of $\reset(k)$ begins, we need to ensure that the content of its local memory is an independent copy of the foreground data structure up to level $k$. 
During the execution of $\reset(k)$, it performs a water-filling primal-dual algorithm while also handling element updates from the adversary, in a  similar way that we have done with the greedy set cover algorithm.

When the execution of $\reset(k)$ terminates, it switches its local memory to the foreground and aborts all other instances $\reset(i), \forall i<k$. To ensure that all the aborted instances $\reset(i)$ will have a local copy of the current data structure up to level $i$, besides executing the water-filling procedures, 
the approach of \cite{bhattacharya2021dynamic} is that
the instance $\reset(k)$ will also be responsible for initializing an independent copy of the data structures up to level $i$ for instance $\reset(i)$, for all $i < k$, right after $\reset(i)$ is aborted by $\reset(k)$. This is the main time bottleneck of the algorithm of \cite{bhattacharya2021dynamic}: 
%which incurs a quadratic dependency on $\log(C n)$, 
as the $\reset(k)$ instance prepares the initial memory contents for all other instances below it, this incurs a running time of at least $\sum_{i=0}^k O(i) = O(k^2)= O(\log^2_\beta (C n))$.

To save one extra $\log_\beta (Cn)$ factor in the runtime, we do the same as our algorithm in the high-frequency regime. In our algorithm, $\reset(k)$ will not be responsible for initializing the memory contents of $\reset(i)$, for any $i<k$. Instead, each instance $\reset(i)$ will initialize its own memory in the background by copying data structures in the foreground up to level $i$, and only when the initialization phase is done, should it begin with the water-filling procedure. By doing so, we can obtain an improved time bound of $O\brac{\frac{f\log(Cn)}{\epsilon^2}}$. In the end, to remove runtime dependency on the aspect ratio $C$, we will apply the same black-box reduction as we did in the high-frequency regime, refer to \Cref{removec} for details.

As mentioned above, our first goal will be to prove the following Lemma:

\begin{lemma}\label{extend-logC}
	For any set system $(\univ, \sets)$ with set cost range $[1/C, 1]$ that undergoes a sequence of element insertions and deletions, where the frequency is always bounded by $f$, and for any $\epsilon \in (0, 1)$, there is a dynamic algorithm that maintains a $((1+\epsilon)f)$-approximate minimum set cover in $O\brac{\frac{f\log (Cn)}{\epsilon^2}}$ deterministic worst-case update time.
\end{lemma}

\noindent To do so, we will follow the general lines of \Cref{threes}, highlighting two major points throughout. The first, is the differences between our low-frequency algorithm presented next and the high-frequency algorithm presented in \Cref{threes}, regarding certain definitions, invariants, properties, etc. The second is the differences between our low-frequency algorithm presented next and the algorithm presented in \cite{bhattacharya2019new} and \cite{bhattacharya2021dynamic}. We mention that there are some differences between \cite{bhattacharya2019new} and \cite{bhattacharya2021dynamic} regarding specific invariants/definitions/etc., and we mainly chose to compare to \cite{bhattacharya2021dynamic}, since the aim of that paper and this section is essentially the same, to deamortize the low-frequency amortized algorithm presented in \cite{bhattacharya2019new}. For the sake of brevity we will omit several details that remain the same from \cite{bhattacharya2019new}/\cite{bhattacharya2021dynamic} or \Cref{threes}, and refer to the relevant lemmas/properties/invariants/etc. when needed.

\subsection{Preliminaries and Invariants}
%  \shay{maybe stress that here we start from level 0 rather than -1 (here level 0 stands for sets not in the solution)}

Each set $s\in \sets$ is assigned a level value $\lev(s)\in [0, \ceil{\log_\beta(Cn)}+1]$; here the difference from the greedy set cover is that we start from level $0$ rather than $-1$. For each element $e$, its level is given by $\lev(e) = \max\{\lev(s) : s\in \sets , e \in s\}$. Following \cite{bhattacharya2019new,bhattacharya2021dynamic}, each element will be classified as \emph{alive}, or \emph{dead}, and further classify alive elements as \emph{active} or \emph{passive}. An element is dead when it has been deleted from $\univ$ by the adversary, but is still lingering in the system due to the algorithm. As before, $\univ^+\supseteq \univ$ denotes the set of all elements including dead ones, and $\univ$ denotes the set of existing elements from the perspective of the adversary. For each level $k$, define $E_k = \{e\in \univ^+\mid \lev(e) = k\}, A_k = \{e\in \univ\text{ is active}\mid \lev(e)\leq k \}, P_k = \{e\in \univ^+\text{ is not active}\mid \lev(e)\leq k \}$. Note that $P_k$ contains both passive and dead elements, as in \Cref{threes} and as opposed to \cite{bhattacharya2021dynamic}, but $A_k$ does not, as opposed to \Cref{threes}. Moreover, the partition now of $\univ$ to passive and active is binary, in contrast to what we had before where passiveness is parameterized by levels, which will simplify things.

As before, during the algorithm, each set $E_k$ is maintained as a linked list, and all pointers to non-empty sets $E_k$ are stored in a doubly-linked list. Moreover, sets $A_k, P_k$ will not be maintained explicitly in our algorithm, and they are only used for the analysis of approximation ratio, just like in \Cref{threes}, and in contrast to \cite{bhattacharya2021dynamic}.

Similarly to \cite{bhattacharya2021dynamic}, each active element $e$ is assigned a weight $\wts(e) = \beta^{-\lev(e)}$, and each passive element is assigned a weight $\wts(e)\leq \beta^{-\lev(e)}$. From the moment an element becomes dead, its weight in the system will not change until the algorithm removes it from the system. The weight of a set $s$ is given by $\wts(s) = \sum_{e\in s\cap\univ^+}\wts(e)$. Moreover:

%The \emph{extra} weight $\delta(s)$ of a set is the total weight that $s$ receives from elements on levels higher than $\lev(s)$; that is $\delta(s) = \sum_{e\in s\cap \univ^+, \lev(s)>\lev(s)}\wts(e)$.

\begin{definition}
	A set $s\in\sets$ is called \emph{tight} if $\wts(s) > \frac{\cost(s)}{\beta}$; otherwise it is called \emph{slack}.
\end{definition}

The algorithm will maintain the following invariant; 
refer to \Cref{inv} for comparison.

%\shay{I'd try to make the following invariant more similar to our invariants 2.1, as follows: remove the current condition (3) -- instead, near the explanation after the invariant that says that the set cover solution is comprised of all tight sets, say that we'll guarantee that every element is contained in at least one tight set. Then, switch between conditions (1) and (2), so that they coincide with invariant 2.1 (i.e., (1) is an upper bound, (2) is a lower bound (or alternatively to change invariant 2.1 instead))}

% \shay{should explain the meaning of fully global here (in the first invariant we include dead weight; in the second we exploit the fact that passive elements have lower weights wrt to their level)}
\begin{invariant}{\ } \label{lastinv}
    \begin{enumerate}[(1),leftmargin=*]
	\item For each $s\in \sets$, $\wts(s)\leq \cost(s)$.
        \item All sets $s$ for which $\lev(s) \ge 1$ are tight.
	\item $|P_k|\leq 2\epsilon\cdot |A_k|$ for each $k\geq 0$.
    \end{enumerate}
\end{invariant}
The above invariants underlie a fully global approach: for the first invariant, element insertions will be handled lazily and we will exploit the fact that passive elements have lower weights with respect to their levels; for the second invariant, element deletions will also be handled lazily and set tightness are preserved by dead weights.

The algorithm will make sure that every element is contained in at least one tight set, and the output set cover $\sets_\alg$ of our algorithm will be the set of all tight sets in $\sets$. Assuming \Cref{lastinv} holds, the approximation ratio is guaranteed by the following statement.
\begin{lemma}[Theorem $4.7$ in \cite{bhattacharya2019new}] \label{frombhatt}
	All tight sets form a $(1+O(\epsilon))f$-approximate minimum set cover of $(\univ, \sets)$.
\end{lemma}

\subsection{Algorithm Description and Update Time Analysis} \label{updtimepd}
We follow our deamortization approach from the previous sections; at a high-level, the resulting  deamortized algorithm is of similar flavor to the one in the high-frequency regime, but there are of course significant differences, starting with the fact that the basic subroutine here for the reset procedure is the water-filling primal-dual algorithm rather than the greedy algorithm. For each level $k$, there is an instance of $\reset(k)$ that operates on an independent chunk of memory, disjoint from the foreground memory where the data structure is stored, and computes a partial solution up to level $k$ using a water-filling algorithm. As in the deamortization for the high-frequency regime, we distinguish between operations on the foreground and the background.

\subsubsection{Foreground}
Element deletions and insertions will be handled in the foreground as follows, similarly to \cite{bhattacharya2019new} and \cite{bhattacharya2021dynamic}.
\begin{itemize}[leftmargin=*]
	\item \textbf{Deletions in the Foreground.} When an element $e\in \univ$ gets deleted by the adversary, mark $e$ as dead, and for each $k\geq \lev(e)$, feed the deletion to $\reset(k)$ which is running in the background.
    
%    \shay{should it be $\ge \lev(e)-1$ as before? (since the merge at level $k+1$ during termination of reset affects level $k+1$)? I think we can use $\lev(e)$ both here and before (since it suffices to delete $e$ from the foreground at level $k+1$), and similarly for insertion}
    
%    \tianyi{I thinks it is just $k\geq \lev(e)$. $\reset(k)$ does not need to worry about deletions on level $k+1$ on the foreground.}
 
	\item \textbf{Insertions in the Foreground.} When a new element $e$ is inserted by the adversary, go over all sets $s\ni e$ to compute $\lev(e) = \max_{s\in \sets}\{\lev(s) \}$. If exists a set $s\ni e$ that is tight then assign $\wts(e) \leftarrow 0$. Otherwise, assign $\wts(e) \leftarrow \min_{s \ni e}\{\cost(s) - \wts(s)\}$, which ensures us that each element has at least one tight set containing it. In both cases $e$ becomes passive. Finally, for each $s\ni e$, update the set weight $\wts(s)\leftarrow \wts(s) + \wts(e)$. After that, feed this new insertion to all instances of $\reset(k)$ for $k\geq \lev(e)$.
 
    %\shay{we need to feed the insertion to $\reset(k)$ in appropriate levels in the background}
	
	% plus, if $\lev(e) > \lev(s)$, update $\delta(s)\leftarrow \delta(s) + \wts(e)$.
	
	\item \textbf{Termination of $\reset(\cdot)$ Instances.} Upon any element update (deletion or insertion), go over all levels $0\leq k\leq \ceil{\log_\beta(Cn)}+1$ and check if any instance $\reset(k)$ has just terminated right after the update. If so, take the largest such level $k$, and switch its memory to the foreground; we will describe how a memory switch is done later on. Then, abort all instances of $\reset(i), \forall 0\leq i < k$.
	
	\item \textbf{Initiating $\reset(\cdot)$ Instances.} Similarly to \Cref{threes}, upon any element update (deletion or insertion), go over all levels $0\leq k\leq L$ and check if there is currently an instance $\reset(k)$. Denote by $k_1,k_2,k_3, \ldots$ the levels that do not have such an instance, where $k_1 < k_2 < k_3 < \ldots$. Next, we want to partition all levels $k_i$ into \emph{short levels} and \emph{non-short levels}. All of the short levels will be lower than the non-short levels, meaning exists $i$ such that $k_{i'}$ is a short level for any $i'<i$ and a non-short level for any $i' \geq i$. In a nutshell, we will be able to execute a short level reset in a single update step, since the number of elements participating in the reset is small enough. Recall that upon termination of a reset we abort all instances of lower level resets, thus there is no reason to run all short level resets, only the highest one. Regarding the non-short levels, we initiate a reset to each and every one of them. To find the highest short level given $k_1,k_2,k_3, \ldots$ we do as follows. First, count all the first $\frac{L}{f}$ elements, from level $0$ upwards. Say that the $\frac{L}{f}$-th element is at level $j$. Thus, we know that $|\bigcup_{i=0}^{j'} E_i| < \frac{L}{f}$ for any $0 \leq j'<j$. Define $i$ to be the highest such that $k_i < j$. If no such $i$ exists then there are no short levels, otherwise $k_i$ is the highest short level, and we initiate the resets for levels $k_i,k_{i+1}, \ldots$, where again $k_i$ is a short level and the rest are non-short levels.
\end{itemize}

\subsubsection{Foreground and Background Data Structures} 

Similar to our previous algorithm, for each level $k\in [0, \ceil{\log_\beta(Cn)}+1]$, any instance of $\reset(k)$ that operates (in the background) maintains a partial copy of the foreground data structures in the background. Specifically, it maintains the following data structures.
\begin{enumerate}[(1),leftmargin=*]
	\item Maintain subsets of elements $\univ^{(k)}, \univ^{(k)+}\subseteq \univ^+$, and for each element $e\in \univ^{(k)+}$, maintain its level $\lev^{(k)}(e)$ as well as its weight $\wts^{(k)}(e)$. In addition, maintain a subset of sets $\sets^{(k)}\subseteq \sets$, and for each set $s\in \sets^{(k)}$, maintain its level $\lev^{(k)}(s)$ as well its weight $\wts^{(k)}(s)$.
	
	\item For each level $i\in [0, k+1]$, let $S^{(k)}_i = \{s\mid \lev^{(k)}(s) = i\}$, $E^{(k)}_i = \{e\mid \lev^{(k)}(e) = i\}$.
\end{enumerate}

\noindent Similarly to the previous algorithm, we maintain the following data structures that link between the foreground and background, which help improve the update time of \cite{bhattacharya2021dynamic}. 
\begin{framed}
	\noindent \textbf{Data Structures that Link between the Foreground and Background.} 
	\begin{enumerate}[(a)]

        \item For each $k$, we have pointers to the sets $S^{(k)}_i$ and $E^{(k)}_i$, stored in two arrays of size $L+2$ and $L+1$, respectively (an entry for every $i \in [-1,L]$ and $i \in [0,L]$, respectively). In addition, the head of the list $S^{(k)}_i$ and $E^{(k)}_i$ keeps a Boolean value which indicates whether it is in the foreground or not.
        
		\item We store an array in the foreground $\lev[\cdot]$ indexed by $s\in \sets$ and $e\in \univ$. So the size of this array should be $O(|\sets| + |\univ|)$. Here we have assumed that sets and elements have unique identifiers from a small integer universe (if the sets and elements belong to a large integer universe and assuming we would like to optimize the space usage, we can use hash tables instead of arrays). For each $s\in \sets$ and $0\leq k\leq L$, $\lev[s][k]$ stores a pointer to the memory location containing the value of $\lev^{(k)}(s)$, as well as a pointer to the list head of $S^{(k)}_i$ if $\lev^{(k)}(s) = i$ (and $i \neq -1$). Similarly, $\lev[e][k]$ stores a pointer to the memory location of $\lev^{(k)}(e)$, as well as a pointer to the memory location of the pointer to $E^{(k)}_i$ if $\lev^{(k)}(e) = i$.

  	\item Similarly, we store an array in the foreground $\wts[\cdot]$ indexed by $s\in \sets$ and $e\in \univ$, of size $O(|\sets| + |\univ|)$. For each $s\in \sets$ and $0\leq k\leq L$, $\wts[s][k]$ stores a pointer to the memory location containing the value of $\wts^{(k)}(s)$. Similarly, $\wts[e][k]$ stores a pointer to the memory location of $\wts^{(k)}(e)$.

    \end{enumerate}
\end{framed}

\paragraph{Foreground Operations.} The collections $S_i$ and $E_i$ for each $i$ will be maintained in doubly linked lists in the foreground. To access the level and weight values in the foreground, $\lev(s)$ and $\wts(s)$ respectively, we can enumerate all indices $k\in [0, L]$ and check the entry $\lev[s][k]$ that points to $\lev^{(k)}(s) = i$ and list head $S^{(l)}_{i}, l\geq i$. If either $\lev[s][k]$ is a null pointer, or $\lev[s][k]$ points to a value $\lev^{(k)}(s) = i$ but $S_i^{(l)}$ is not in the foreground, then we know $\lev^{(k)}(s)\neq\lev(s)$. Once we reach $k'$ such that $\lev^{(k')}(s) = i'$ and $S_{i'}^{(l)}$ is in the foreground (determined by the Boolean value), we know that $\lev(s) = \lev^{(k)}(s) = i'$ and $\wts(s) = \wts^{(k)}(s)$. Therefore, accessing the foreground level value $\lev(s)$ and weight value $\wts(s)$ takes time $O(L)$. Similarly, accessing the foreground level value $\lev(e)$ and weight value $\wts(e)$ takes time $O(L)$ as well. We conclude that upon insertion of element $e$, enumerating all sets $s \ni e$ to decide which level $e$ should be at, determining the weight of element $e$, and updating the set weights of all sets that contain $e$, takes time $O(f \cdot L) = O(f \cdot \log_\beta(Cn))$. Regarding initiating resets, in $O(L) = O(\log_{\beta}(Cn))$ time we can go over all levels to check which ones need to be initiated, and also check which is the highest short level, by using $E_i$ which is maintained in the foreground. We remind that a short level is a level $k$ such that $\reset(k)$ is not working currently in the background, and $f \cdot \sum_{i=0}^k |E_i| < L$. Moreover, there is only one short level reset operating per update step (the highest one).

\subsubsection{Background}

\paragraph{Initialization Phase.} When an instance of $\reset(k)$ has been initiated, it initializes its own data structures by setting $\sets^{(k)}_\alg = S_{i}^{(k)} = E_{i}^{(k)} \leftarrow \emptyset, \forall i\in [0, k+1]$, and $\univ^{(k)} = \univ^{(k)+} \leftarrow \bigcup_{i=0}^k E_i$. After that, enumerate all elements of $\univ^{(k)}$ and let $\sets^{(k)}$ be all the sets containing at least one element in $\univ^{(k)}$. So $\sets^{(k)} \subseteq \{s\mid 0\leq\lev(s)\leq k\}$. Note that we cannot create $\sets^{(k)}$ directly from the sets $S_k, S_{k-1}, \ldots, S_0$, as there might be some tight sets on level $0$ containing some elements in $\univ^{(k)}$ and there might be some slack sets on level 0,
but finding the tight sets at level 0 out of all sets there might take too much time.

%\shay{I rephrased the previous sentence. Should also use it in Sec 2}

To obtain all the pointers to $E_i, 0\leq i\leq k$, we follow the linked list consisting of these pointers, from $E_k$ down to $E_0$. One remark is that, during the enumeration of the list from $E_k$ down to $E_0$, some pointers might already be switched by other instances of $\reset(i), i<k$; the correctness of this step can be argued equivalently to the proof of \Cref{consistency}.

%The collection of sets $\sets^{(k)}$ will contain all sets up to level $k$ (including $k$ and $-1$) that contain an element in $\univ^{(k)+}$ (so $|\sets^{(k)}| \leq f|\univ^{(k)+}|$).

%Besides, we also reassign $\lev^{(k)}(s) =\lev^{(k)}(e)\leftarrow k+1$, and $\wts^{(k)}(e)\leftarrow \min\{\wts(e), (1+\epsilon)^{-k-1} \}, \forall e\in \univ^{(k)+}$, and update the value of $\wts^{(k)}(s)$ accordingly for each $s\in \sets^{(k)}$.
During the initialization phase, we remove all dead elements from $\univ^{(k)+}$. 
All the above steps will be planned in the background for the non-short levels.
%All the above steps do not have worst-case runtime, so they will be planned in the background. 
%If a new element $e$ is inserted by the adversary during the initialization phase, assign level $\lev^{(k)}(e) \leftarrow k+1$, and weight $\wts^{(k)}(e)\leftarrow \min_{s\ni e}\{\cost(s) - \wts^{(k)}(s), (1+\epsilon)^{-k-1} \}$; if an old element $e\in \univ^{(k)+}$ is deleted by the adversary during the initialization, remove it from $\univ^{(k)}$.
Copying the memory locations of the pointers of $\{E_k, E_{k-1}, \cdots, E_0\}$ and $\{S_k, S_{k-1}, \cdots, S_{-1}\}$ takes time at most $O(k)$. Obtaining the collection $\mathcal{S}^{(k)}$ takes $O(f|\univ^{(k)}|)$ time, so the total running time of this phase is $O(f|\univ^{(k)}| + k) = O(f|\univ^{(k)}| + L)$. Since all the non-short levels satisfy $f \cdot |\univ^{(k)}| \geq L$, we get that the running time of this phase for non-short levels is $O(f|\univ^{(k)}|)$, and for the short level reset it is $O(L) = O(\log_{\beta}(Cn))$.

\paragraph{Water-Filling Phase.} Once an instance of $\reset(k)$ has been initiated, we run the algorithm \emph{EfficientRebuild($k$)} given in \cite{bhattacharya2021dynamic}. If an element $e$ is inserted or deleted by the adversary during the initialization phase or the water-filling phase, it is treated just as it is in \cite{bhattacharya2021dynamic}. The following properties hold following the execution of \emph{EfficientRebuild($k$)}, by \cite{bhattacharya2019new} and \cite{bhattacharya2021dynamic}.

\begin{enumerate} [(1)]
    \item Every element $e$ that was active at level up to $k$ in the foreground upon beginning this phase will end up active and at level up to $k+1$.

    \item Every element $e$ that was passive at level up to $k$ in the foreground upon beginning this phase will end up active at level up to $k+1$ or passive at level \emph{exactly} $k+1$. 

    \item Every element $e$ that was dead at level up to $k$ in the foreground upon beginning this phase will end up removed completely of the system.

    \item The procedure does not touch any element that was at level $\geq k+1$ in the foreground upon the beginning of this phase.

    \item Each set in $S^{(k)}_i$ for any $1 \leq i \leq k+1$ is tight.

    \item Each element participating in the reset is contained in at least one tight set.

    \item $|P_i| < \epsilon \cdot |A_i|$ for all $i \leq k$. Notice that by $(2)$ it may seem that $|P_i| = 0$ for any $i \leq k$, but $(2)$ refers only to elements that were in the system when this phase began, and not elements that were inserted during the execution. 
\end{enumerate}

\noindent The running time of this procedure is given in the following lemma:

\begin{lemma}[Claim $C.2$ in \cite{bhattacharya2021dynamic}] \label{bhatt_time}
	The running time of the water-filling phase of $\reset(k)$ is $O(f \cdot |\univ^{(k)+}|)$.
\end{lemma}

\paragraph{Termination Phase.} When the water-filling algorithm $\reset(k)$ is finished, we set $S_i, E_i$ to be $S_i^{(k)}, E_i^{(k)}, i\in [0, k]$, and append the linked list of $S^{(k)}_{k+1}, E^{(k)}_{k+1}$ to $S_{k+1}, E_{k+1}$. Finally, abort all lower-level reset instances. The implementation of switching the pointers is roughly the same as in the high-frequency regime, as explained next.
%\shay{maybe say a few words about switching pointers, which is dealt with as in the high-frequency regime; this part is important, because this is where our approach differs from the previous}

%Using the above algorithm
%and analysis very similar to the one used for the high-frequency regime, we derive the following statement. 
%\shay{rephrase}.

When $\reset(k)$ has finished and called upon to switch its memory to the foreground, for every index $0\leq i\leq k$, we replace every pointer to $S_i$ with the pointer to $S^{(k)}_i$, and connect all pointers to nonempty lists $S^{(k)}_i, 0\leq i\leq k$ with a linked list. 

Merging the list $S_{k+1}^{(k)}$ with the list $S_{k+1}$ on the foreground is in fact done in the water-filling phase, as it does not have worst-case runtime guarantee. Nevertheless it will be explained here as part of the termination phase. Upon merging $S_{k+1}^{(k)}$ with $S_{k+1}$, suppose $S_{k+1}$ is equal to some $S_{k+1}^{(l)}$ for some $l\geq k+1$; in other words, $S_{k+1}$ was computed in some instance $\reset(l)$, where $l\geq k+1$. Then, for each set $s\in S^{(k)}_{k+1}$, redirect the pointer $\lev[s][k]$ from the list head $S^{(k)}_{k+1}$ to the list head $S_{k+1}^{(l)}$. After that, concatenate the two lists $S_{k+1}^{(k)}, S_{k+1}$. We can merge the two lists $E_{k+1}^{(k)}$ and $E_{k+1}$ in a similar manner. The running time of this is clearly $O(f \cdot |\univ^{(k)}|)$, thus not the bottleneck of the water-filling phase.

As for the runtime of the termination phase, the number of memory pointers to be switched is at most $O(L) = O(\log_\beta(Cn))$, so the worst-case runtime is $O(L) = O(\log_\beta(Cn))$.
%\paragraph{Basic operations.} To implement fast access to the level value $\lev^{(k)}(s)$ and $\lev^{(k)}(e)$ for any set $s$ and element $e$ in constant time, we also need to specify the following data structures.
For any set $s\in \sets^{(k)}$, to access any level value $\lev^{(k)}(s)$ in the background, we can follow the pointer $\lev[s][k]$ and check $\lev^{(k)}(s)$ in constant time; similarly we can check the value of $\lev^{(k)}(e)$ for any $e\in \univ^{(k)+}$ in constant time.

Since our data structures maintain multiple versions of the level values, and our algorithm keeps switching memory locations between the foreground and the background, we need to argue the consistency of the foreground data structures; this was done in \Cref{consistency} for the high-frequency regime, which equivalently holds for the low-frequency regime as well.

By the algorithm description, upon an element update, we scan all the levels, and if any instance of $\reset(\cdot)$ has just terminated, we switch the one $\reset(k)$ on the highest level $k$ to the foreground as we have discussed in the previous paragraph, while aborting all other instances of $\reset(i), i<k$. As we have seen, switching a single instance of $\reset(\cdot)$ to the foreground takes $O(L) = O(\log_\beta(Cn))$ time, and so the worst-case time of this part is $O(L) = O(\log_\beta(Cn))$. Therefore, the termination can be done in a single update step (along with the short level reset).

\bigskip

\noindent To conclude \Cref{updtimepd}, any insertion/deletion in the foreground can be dealt with in $O(f \cdot L)$ time, thus it can be done within a single update step. Dealing with the short level reset (finding the highest one, initializing the reset, and executing it) takes $O(L)$ time, thus it will all be done within a single update step as well. Likewise, termination can be done in $O(L)$ time (updating data structures of the highest finished reset and aborting the rest). Every other $\reset(k)$ will take $O(f \cdot |\univ^{(k)+}|)$ time (initialization and water-filling phase). Since this cannot be executed within a single update step, for each $k$ we will execute $O(\frac{f}{\epsilon})$ computations per update step. In \Cref{proofcorrectpd} we will prove that by working in such a pace, we can ensure that \Cref{lastinv} holds, which in turn will be enough to prove that the approximation factor holds, by \Cref{frombhatt}.

\subsection{Proof of Correctness} \label{proofcorrectpd} 

In this section we prove that \Cref{lastinv} always holds in the foreground, mainly relying on the properties of \emph{EfficientRebuild($k$)} given in \cite{bhattacharya2021dynamic}. 

\begin{claim} \label{lastinv1}
	\Cref{lastinv}(1) always holds in the foreground.
\end{claim}

\begin{proof}
    Assume \Cref{lastinv}(1) holds right before time step $t$. Following the insertion/deletion at time $t$, it is easy to verify that \Cref{lastinv}(1) still holds. Indeed, a deletion does not raise the weight of any element, and due to an insertion of element $e$ either $\wts(e) = 0$, or $e$ is assigned a weight small enough such that the weight of any set $s \ni e$ does not surpass its cost, by the description of the algorithm. Following the termination of some $\reset(k)$, \Cref{lastinv}(1) holds from the description of the reset procedure presented in Section $C.4$ in \cite{bhattacharya2021dynamic}.
\end{proof}

\begin{claim} \label{lastinv2}
   \Cref{lastinv}(2) always holds in the foreground.
\end{claim}

\begin{proof}
     Assume \Cref{lastinv}(2) holds right before time step $t$. Following the insertion/deletion at time $t$, it is easy to verify that \Cref{lastinv}(2) still holds. Indeed, following the update the weight of any set at level $\geq 1$ does not reduce. Following the termination of some $\reset(k)$, \Cref{lastinv}(2) holds by Lemma C.10 in \cite{bhattacharya2021dynamic}.
\end{proof}

\begin{claim} [Lemma C.11 in \cite{bhattacharya2021dynamic}] \label{lastclaim4inv}
	Following the termination of an instance $\reset(k)$, we have $|P_i| < \epsilon \cdot |A_i|$ for all $i \leq k$.
\end{claim}

\begin{claim} \label{lastlarger}
	When $\reset(k)$ terminates and is transferred to the foreground, it does not change $A_{k'}$ or $P_{k'}$, for any $k'>k$.
\end{claim}

\begin{proof}
    The claim follows immediately from the description of the reset procedure presented in Section $C.4$ in \cite{bhattacharya2021dynamic}, and by Property $3.6(4)$ in \cite{bhattacharya2019new}.
\end{proof}

\begin{lemma}\label{lastinv3}
	\Cref{lastinv}(3) always holds in the foreground.
\end{lemma}

\begin{proof}

Assume by contradiction that there exists $k$ such that $|P_k| > 2\epsilon \cdot |A_k|$ for some $k$, at the end of update step $t$ (right before update step $t+1$). Let $t'$ be the last time step before $t$ that $\reset(k)$ was initiated. $\reset(k)$ was initiated at time $t'$ right after the termination of some $\reset(k')$, where $k' \geq k$, by the algorithm description. By \Cref{lastclaim4inv}, we know that at time $t'$ (following the termination of $\reset(k')$) we have $|P_k| < \epsilon \cdot |A_k|$ in the foreground.

Between $t'$ and $t$, $|P_k|$ can increase and $|A_k|$ can decrease only due to insertions/deletions, since by definition of $t'$ no termination of $\reset(k')$ for any $k' \geq k$ exists, and any termination of $\reset(k')$ for any $k' < k$ would not raise $|P_k|$ or change $A_k$ by \Cref{lastlarger}. By denoting $A_k$ at times $t'$ and $t$ by $A_k^{t'}$ and $A_k^{t}$ respectively, we conclude that there must be more than ($2\epsilon \cdot |A_k^{t}| - \epsilon \cdot |A_k^{t'}|$) update steps between $t'$ and $t$, since again due to one insertion/deletion only one element can join $P_k$. 
	
If we denote by $x$ the total number of update steps throughout which $\reset(k)$ is executed, we can execute the reset at a pace such that:

	\begin{equation} \label{inv3eq1pd}
		x < \frac{\epsilon}{2}(|P_k^{t'}| + |A_k^{t'}|),
	\end{equation} 
	
	\noindent since the collection of all elements participating in $\reset(k)$ initiated at time $t'$ is $P_k^{t'} \cup A_k^{t'}$. We assume that at update step $t$, $\reset(k)$ is still running, thus we will reach a contradiction if:
	
	\begin{equation} \label{inv3eq2pd}
		\frac{\epsilon}{2}(|P_k^{t'}| + |A_k^{t'}|) < 2\epsilon \cdot |A_k^{t}| - \epsilon \cdot |A_k^{t'}|.
	\end{equation} 
	
	\noindent Now, notice that:
	
	\begin{equation} \label{inv3eq3pd}
		|A_k^{t}| \geq |A_k^{t'}| - x,
	\end{equation} 
	
	\noindent since again, in each update step during the reset up to one element can be removed from $A_k$, and we assumed that $x \geq t-t'$. Thus, we need to show that:
	
	\begin{equation} \label{inv3eq4pd}
		\frac{\epsilon}{2}(|P_k^{t'}| + |A_k^{t'}|) < 2\epsilon \cdot (|A_k^{t'}|-x) - \epsilon \cdot |A_k^{t'}|.
	\end{equation} 
	
	\noindent Plugging in $x$ from \Cref{inv3eq1pd} and rearranging, we get that we need to show:
	
	\begin{equation} \label{inv3eq5pd}
		(\frac{1}{2} + \epsilon)|P_k^{t'}| < (\frac{1}{2} - \epsilon)|A_k^{t'}|.
	\end{equation} 
	
	\noindent Since we know that $|P_k^{t'}| < \epsilon|A_k^{t'}|$, it is enough to show that:

        \begin{equation} \label{inv3eq5.5pd}
		(\frac{1}{2} + \epsilon) \cdot \epsilon|A_k^{t'}| < (\frac{1}{2} - \epsilon)|A_k^{t'}|,
	\end{equation} 

\noindent meaning that 

        \begin{equation} \label{inv3eq6pd}
		2\epsilon^2 + 3\epsilon -1 < 0,
	\end{equation} 

 \noindent which holds for any $\epsilon < \frac{1}{4}$.	Thus, we reach our contradiction and the lemma follows.
\end{proof}

We conclude that our algorithm indeed maintains \Cref{lastinv}, in worst-case update time of $O(\frac{f \cdot L}{\epsilon}) = O\brac{\frac{f\log(Cn)}{\epsilon^2}}$ as proved in \Cref{updtimepd}. Since \Cref{lastinv} holds, the approximation factor of the maintained minimum set cover is $(1+\epsilon)f$, as shown in \Cref{frombhatt}. This concludes the proof of \Cref{extend-logC}.

\begin{comment}
\subsection{Removing Dependency on Aspect Ratio}
To remove the dependency on $C$ in the update time, we employ the same black-box reduction from \Cref{remove} for the high-frequency regime. 

%\shay{I stop here; I'll first read Section 3 and then return to this point}
Define the \emph{top level} of set $s$ as $\tlev(s) = \ceil{\log_\beta(n / \cost(s))}$. Define a parameter $K = \ceil{10\log_\beta n}$.
\begin{definition}
	For any $l\geq -1$, define $\sets_l = \{s\in \sets\mid \tlev(s)\in [lK, (l+2)K) \}$, and let $\univ_l$ be the collection of all elements such that the cheapest set containing them is in $\sets_l$.
\end{definition}
For each $l\geq -1$, we will maintain a set cover $\sets_{l, \alg}\subseteq \sets_l$ that covers $\univ_l$ by applying \Cref{extend-logC} as a black-box which takes $O\brac{\frac{f\log n}{\epsilon^2}}$ worst-case update time. Overall, we will have two different set cover solutions for $\univ$:
$$\sets_{\even, \alg} = \bigcup_{j\geq 0}\sets_{2j, \alg} \; , \; \sets_{\odd, \alg} = \bigcup_{j\geq 0}\sets_{2j-1, \alg}. $$
Similar to the proof argument for the high-frequency setting, we can prove that at least one solution from $\{\sets_{\even, \alg}, \sets_{\odd, \alg} \}$ is a $(1+\epsilon)f$-approximation, which concludes the proof of \Cref{wc2}.

\end{comment}

\subsection{Removing Dependency on Aspect Ratio} \label{removec}
To remove the dependency on $C$ in the update time, we employ the same black-box reduction from \Cref{remove} for the high-frequency regime. 
%\shay{I stop here; I'll first read Section 3 and then return to this point}
Define the \emph{top level} of set $s$ as $\tlev(s) = \ceil{\log_\beta(n / \cost(s))}$. Define a parameter $K = \ceil{10\log_\beta n}$.
\begin{definition}
	For any $l\geq -1$, define $\sets_l = \{s\in \sets\mid \tlev(s)\in [lK, (l+2)K) \}$, and let $\univ_l$ be the collection of all elements such that the cheapest set containing them is in $\sets_l$.
\end{definition}
For each $l\geq -1$, we will maintain a set cover $\sets_{l, \alg}\subseteq \sets_l$ that covers $\univ_l$ by applying \Cref{extend-logC} as a black-box which takes $O\brac{\frac{f\log n}{\epsilon^2}}$ worst-case update time. Overall, by \Cref{ob1} and \Cref{ob2} we will have two different set cover solutions for $\univ$:
$$\sets_{\even, \alg} = \bigcup_{j\geq 0}\sets_{2j, \alg} \; , \; \sets_{\odd, \alg} = \bigcup_{j\geq 0}\sets_{2j-1, \alg}. $$
Following similar lines to those in the argument for the high-frequency setting, it can be proved that at least one solution from $\{\sets_{\even, \alg}, \sets_{\odd, \alg} \}$ is a $(1+\epsilon)f$-approximation; we next provide this argument, for completeness, which would complete the proof of \Cref{wc2}.

\begin{lemma}\label{approxlow}
	For any $l\geq -1$, let $k_0\in I_l$ be the smallest index such that there exists a set in $\mathcal{S}_l$ covering at level $k_0$, or in other words, $k_0$ is the smallest level such that $\sets_{l, \alg}\cap S_{l, k_0}\neq \emptyset$. If $k_0 \leq (l+1)K$, then we have:
	$$\sum_{j \geq 0}\cost(\sets_{l + 2j, \alg})\leq (1+O(\epsilon))f \cdot \cost(\sets^*).$$
\end{lemma}
\begin{proof}
    We know that 
$\cost(\sets_{l, \alg}) \le (1+O(\epsilon))f \cdot \cost(\sets^*_l) \le (1+O(\epsilon))f \cdot \cost(\sets^*)$. 
%Next, let us prove that $\sum_{j\geq 1}\cost(\sets_{l+2j, \alg})\leq \epsilon\cdot \cost(\sets_{l,\alg})$. 
Consider any set $s_0\in \sets_{l, \alg}\cap S_{l, k_0}$. For any set $s\in \sets_{l+2j, \alg}$ and $j>0$, we have: 
	\begin{equation} \label{costs2} \cost(s)\leq n \cdot \beta^{1-\tlev(s)}\leq n\cdot \beta^{1-(l+2)K} \leq n\cdot \beta^{1-k_0-K} \leq \frac{\epsilon}{(1+2\epsilon) n}\cdot\beta^{-k_0-1}\leq \frac{\epsilon\cdot\cost(s_0)}{(1+2\epsilon) n},
 \end{equation}
	
\noindent where the first inequality holds since  $\tlev(s) \leq \log_\beta (n / \cost(s)) + 1$
(by definition of $\tlev(s)$), the second holds since $\tlev(s) \geq (l+2j)K \ge (l+2)K$ for any $j>0$ by definition of $s$, the third follows from the initial assumption that $k_0 \leq (l+1)K$, the fourth by definition of $K$ and that $\epsilon > \frac{2}{n^8}$ is not too small, and the fifth holds since $s_0$ is covering at level $k_0$. %Since each element is assigned to one set in the set cover, and only up to a $\frac{2\epsilon}{1+2\epsilon}$-fraction of the elements can be dead (otherwise \Cref{inv-scale}(3) is violated) 
Since at most $n$ elements are alive and at most $2\eps \cdot n$ elements may be dead (otherwise
\Cref{lastinv}(3) is violated), and as each element (dead or alive) is assigned to at most one set in the set cover solution, it follows that the total number of sets in
$\sets_{l+2j, \alg}$ over all $j>0$
is bounded by $(1+2\epsilon) n$. By \Cref{costs2}, we get $$\sum_{j>0}\cost(\sets_{l+2j, \alg}) ~\leq~
(1+2\eps)n \cdot \left(\frac{\epsilon}{(1+2\epsilon) n}\cdot\cost(s_0)\right)
~\le~
\epsilon\cdot\cost(s_0) ~\leq~ \epsilon\cdot \cost(\sets_{l, \alg}).$$ We conclude that
$$\sum_{j \geq 0}\cost(\sets_{l + 2j, \alg}) ~=~ \cost(\sets_{l, \alg}) + \sum_{j>0}\cost(\sets_{l+2j, \alg}) ~\leq~ (1+\epsilon) \cdot \cost(\sets_{l, \alg}) ~\leq~ (1+O(\epsilon))f \cdot \cost(\sets^*),$$
\noindent and the lemma follows.
\end{proof}

\noindent Next, by following the same lines of \Cref{approxf} (with $f$ instead of $\ln n$) and using \Cref{approxlow}, we can show that: $$\min\{\cost(\sets_{\even, \alg}), \cost(\sets_{\odd, \alg}) \}\leq (1+O(\epsilon))f \cdot \cost(\sets^*),$$

\noindent and we have thus concluded the proof of \Cref{wc2}.

\section{From Dynamic Dominating Set to Set Cover} \label{dssec}

\paragraph{The Standard Reduction from {\em Static} Dominating Set to Set Cover.}

Given our graph $G = (V, E)$ with $V = \{v_1, v_2, \ldots, v_n\}$, one can construct a set cover instance $(\univ, \sets)$ as follows. The universe is $\univ = \{e_1, e_2, \ldots, e_n\}$, and the family of sets is $\sets = \{S_1, S_2, \ldots, S_n\}$ such that $S_i$ consists of the element $e_i$ and all elements $e_j$ such that $v_j$ is adjacent to $v_i$ in $G$. Now if $D$ is a dominating set for $G$, then $\sets_{\alg} = \{S_i ~\vert~ v_i \in D\}$ is a feasible solution of the set cover problem, with $\cost(\sets_{\alg}) = \cost(D)$. Conversely, if $\sets_{\alg}$ is a feasible solution of the set cover problem, then $D = \{v_i ~\vert~ S_i \in \sets_{\alg}\}$ is a dominating set for $G$, with $\cost(D) = \cost(\sets_{\alg})$. Hence, the cost of a minimum dominating set for $G$ equals the cost of a minimum set cover for $(\univ, \sets)$, and in particular, if $\sets_{\alg}$ provides an $\alpha$-approximation for the set cover problem given by $(\univ, \sets)$, then $D = \{v_i ~\vert~ S_i \in \sets_{\alg}\}$ provides an $\alpha$-approximation for the dominating set problem given by $G$.

\subsection{A Reduction in the Dynamic Setting}

\paragraph{Handling Edge Insertions and Deletions.}

There is a key difference between the set cover problem and the dominating set problem in the dynamic setting.
In the set cover problem the adversary inserts or deletes an {\em element} upon each update step, whereas in the dominating set problem the adversary inserts or deletes an {\em edge} upon each update step, which can be thought of as creating or removing two \emph{connections} between an element to a set in the set cover problem. That is, if the edge $(v_i,v_j)$ is inserted, then in the set cover problem we would want to add $e_i$ to $S_j$ and $e_j$ to $S_i$; similarly, if the edge $(v_i,v_j)$ is deleted, then in the set cover problem we want to remove $e_i$ from $S_j$ and $e_j$ from $S_i$. Since such an operation does not exist in the basic dynamic setting of the set cover problem, we will treat such an operation as an element deletion followed by an element insertion of the same element, just with a different (by one set) collection of sets that can cover the element, as explained next. 

\paragraph{Algorithm.}

We will use \EMPH{as a black box} the algorithm for the set cover problem in the high-frequency regime, provided by \Cref{wc} (and presented in Sections \ref{threes} and \ref{remove}). We are given initially a graph $G$ with $n$ vertices and no edges. Thus, in our reduction to set cover, we will begin with $n$ active elements $e_1, e_2, \ldots, e_n$, and $n$ sets $S_1, S_2, \ldots, S_n$ such that each $S_i$ contains only $e_i$. Upon insertion of edge $(v_i,v_j)$ by the adversary, in our set cover system we delete $e_i$ and insert it with the same collection of sets that it was contained in before, \emph{plus} $S_j$; likewise, we delete $e_j$ and insert it with the same collection of sets that it was contained in before, \emph{plus} $S_i$. Upon deletion of edge $(v_i,v_j)$ by the adversary, in our set cover system we delete $e_i$ and insert it with the same collection of sets that it was contained in before, \emph{minus} $S_j$; likewise, we delete $e_j$ and insert it with the same collection of sets that it was contained in before, \emph{minus} $S_i$.
Overall, each adversary edge update is translated to four set cover updates.
Explicitly performing each of these four set cover updates in the set cover instance takes time that is linear in the degrees of the two endpoints of the updated edge, which is at most $O(\Delta)$.
By working at a pace four times faster than in the set cover algorithm, the worst case update time will still be $O\brac{\frac{f\log n}{\epsilon^2}}$. Notice that the frequency of each element is upper bounded by $\Delta + 1$, thus we actually obtain a worst-case update time of $O\brac{\frac{\Delta \log n}{\epsilon^2}}$. Regarding the approximation factor, recall that by \Cref{pre-toy-approx} we have that $\cost(\sets_\alg)\leq (1+O(\eps))\cdot \ln n' \cdot \cost(\sets^*)$, where $\sets^*$ is an optimal set cover for $\univ$, and $n'$ is an upper bound to the size of each set throughout the update sequence, which is $\Delta + 1$ in this setting. Thus, we obtain an approximation factor of $(1+O(\eps))\cdot \ln (\Delta + 1) = (1+O(\eps))\cdot \ln \Delta$, which concludes the proof of \Cref{ds-final}.

%\color{black}

%\vspace{5mm}
\clearpage
\bibliographystyle{alpha}
\bibliography{ref}

\end{document}